\documentclass[a4paper,twocolumn,11pt,accepted=2022-09-13]{quantumarticle}
\pdfoutput=1

\usepackage{graphicx}
\usepackage{dcolumn}
\usepackage{bm}
\usepackage[sort,square,numbers]{natbib}

\usepackage{amsmath}
\usepackage{amssymb}
\usepackage{amsthm}
\usepackage{graphicx}
\usepackage{algorithmic}
\usepackage{mathrsfs}
\usepackage{braket}
\usepackage{xcolor}
\usepackage{tikz}
\usepackage[roman]{complexity}
\usepackage[caption=false]{subfig}

 \newtheorem{thm}{\protect\theoremname}
 \theoremstyle{plain}
 \newtheorem{lem}[thm]{\protect\lemmaname}
 \theoremstyle{plain}
 \newtheorem{rem}[thm]{\protect\remarkname}
 \theoremstyle{plain}
 \newtheorem*{lem*}{\protect\lemmaname}
 \theoremstyle{plain}
 
 \theoremstyle{plain}

   \providecommand{\corollaryname}{Corollary}
   \providecommand{\lemmaname}{Lemma}
   \providecommand{\propositionname}{Proposition}
   \providecommand{\remarkname}{Remark}
 \providecommand{\theoremname}{Theorem}
 
 \usepackage[normalem]{ulem}

\newcommand{\Or}{\mathcal{O}}

\newcommand{\wt}{\widetilde}
\newcommand{\Tr}{\mathrm{Tr}}

\newcommand{\dd}{\mathrm{d}}
\newcommand{\ZZ}{\mathbb{Z}}

\newcommand{\ketbra}[2]{\ket{#1}\!\bra{#2}}

\usepackage{pdfpages}
\usepackage{pgffor}
\makeatletter
\AtBeginDocument{\let\LS@rot\@undefined}
\makeatother

\newcommand{\changed}[1]{#1}

\begin{document}

\title{Provably accurate simulation of gauge theories and bosonic systems}

\author{Yu Tong}
\affiliation{Google Quantum AI, Venice, CA, USA}
\affiliation{Department of Mathematics, University of California, Berkeley, CA, USA}

\author{Victor V. Albert}
\affiliation{Joint Center for Quantum Information and Computer Science,
NIST and University of Maryland, College Park, MD, USA}

\author{Jarrod R. McClean}
\affiliation{Google Quantum AI, Venice, CA, USA}

\author{John Preskill}
\affiliation{Institute for Quantum Information and Matter, Caltech, Pasadena, CA, USA}
\affiliation{AWS Center for Quantum Computing, Pasadena, CA, USA}

\author{Yuan Su}
\affiliation{Google Quantum AI, Venice, CA, USA}
\affiliation{Institute for Quantum Information and Matter, Caltech, Pasadena, CA, USA}


\date{April 4th, 2022}

\begin{abstract}
Quantum many-body systems involving bosonic modes or gauge fields have infinite-dimensional local Hilbert spaces which must be truncated to
perform simulations of real-time dynamics on classical or quantum computers. To analyze errors resulting from truncation, we develop methods for bounding the rate of growth of local quantum numbers such as the occupation number of a mode at a lattice site, or the electric field at a lattice link. Our approach applies to various models of bosons interacting with spins or fermions such as the Hubbard-Holstein, Fr\"ohlich, and Dicke models, and also to both abelian and non-abelian gauge theories.
We show that if states in these models are truncated by imposing an upper limit $\Lambda$ on each local quantum number, and if the initial state has low local quantum numbers, then
a truncation error no worse than $\epsilon$ can be achieved by choosing $\Lambda$ to increase polylogarithmically with $\epsilon^{-1}$, an exponential improvement over previous bounds based on energy conservation. 
For the Hubbard-Holstein model, we numerically compute an upper bound on the value of $\Lambda$ that achieves accuracy $\epsilon$, finding significant improvement over previous estimates in various parameter regimes.
We also establish a criterion for truncating the Hamiltonian with a provable guarantee on the accuracy of time evolution.
Building on that result, we formulate quantum algorithms for dynamical simulation of lattice gauge theories and of models with bosonic modes; the gate complexity depends almost linearly on spacetime volume in the former case, and almost quadratically on time in the latter case.
We establish a lower bound showing that there are systems involving bosons for which this quadratic scaling with time cannot be improved.
By applying our results on the truncation error in time evolution, we also prove that spectrally isolated energy eigenstates can be approximated with error at most $\epsilon$ by truncating local quantum numbers at $\Lambda=\textrm{polylog}(\epsilon^{-1})$.

\end{abstract}

\maketitle


\section{Introduction}

    Model physical systems are often formulated on spatial lattices, where the local Hilbert space residing on each site or link of the lattice is infinite dimensional. Examples include condensed-matter systems with bosonic degrees of freedom \cite{KlossReichmanTempelaar2019multiset,DelJavierEtAl2018tensor,GuoWeichselbaumEtAl2012critical,SchroderChin2016simulating,ReinhardMordovinaEtAl2019density,SandhoeferChan2016density,WoodsCramerPlenio2015simulating,MacridinEtAL2018digital,MacridinEtAl2018ElectronPhonon}, lattice gauge theories (LGTs) \cite{PichlerDalmonteEtAl2016real,KuhnZoharCiracEtAl2015non,MagnificoFelserEtAl2021lattice,Wiese2013ultracold,BanerjeeDalmonteEtAl2012atomic,ZoharCiracReznik2013cold,ZoharCiracReznik2012simulating,BanulsCarmenEtAl2017efficient,BenderZoharEtAl2018digital,ShawEtAl2020Schwinger,KanNam2021lattice,AbhishekRoggeroWiebe2013hybridized,Tran2021Faster,BanulsBlattEtAl2020simulating,KlcoSavageStryker20202,KlcoDumitrescuEtAl2018quantum,MuschikHeylMartinezEtAl2017u,ByrnesYamamoto2006simulating,Chakraborty2020DigitalSimulation,Davoudi2021Toward}, and other lattice field theories \cite{JordanLeePreskill2012quantum,JordanLeePreskill2014QuantumCompScatter}.
In such models, it is convenient to characterize the local state of the system in terms of a local quantum number, such as the occupation number of a bosonic mode at a particular site, or the electric field of a gauge variable at a particular link. When simulating a lattice model using a classical or quantum computer, it is typically necessary 
to truncate the local Hilbert space, replacing it by a finite-dimensional space in which the local quantum number has a maximum value. We call this maximum value the \emph{truncation threshold}, and denote it by $\Lambda$. 

Quantum states of the ideal untruncated model, if concentrated on relatively low values of the local quantum numbers, can be accurately approximated within the truncated model. However, in a dynamical simulation governed by a specified Hamiltonian, local quantum numbers may increase as the system evolves. Therefore, even if the initial state is well approximated within the truncated model, the approximation might no longer be accurate after evolution for a sufficiently long time. To ensure that the truncated model can accommodate the evolved state we need to bound the rate of growth of the local quantum numbers in the ideal model.

One way to obtain such a bound is to invoke conservation of the total energy. However, even though the total energy is conserved, the local quantum numbers are not, and we need to worry about whether energy which is initially distributed among many lattice sites might become focused on a much smaller number of sites, pushing the local quantum numbers at some sites beyond the capacity of the truncated local Hilbert space. Using conservation of energy, combined with the Chebyshev inequality to bound the probability of large deviations from mean values, one may infer that (for a fixed evolution time), 
quantum states can be truncated with an error at most $\epsilon$ using a threshold $\Lambda$ scaling polynomially with $\epsilon^{-1}$ \cite{JordanLeePreskill2012quantum,JordanLeePreskill2014QuantumCompScatter}. However, it is unclear whether this \emph{energy-based bound} can be used to truncate Hamiltonians with a provable accuracy guarantee when the local quantum numbers are not conserved under time evolution. We will further clarify this issue in Section~\ref{sec:framework}.

In this work, we develop a unified framework that shows, for a large class of models, this energy-based estimate of $\Lambda$ is far too pessimistic --- a truncation threshold scaling as $\polylog(\epsilon^{-1})$ actually suffices, as previously suggested in  \cite{MacridinEtAL2018digital,MacridinEtAl2018ElectronPhonon}. This model class includes systems involving bosons such as the Hubbard-Holstein model \cite{holstein1959studies}, the Fr{\"o}hlich model \cite{Frohlich1954electrons}, and the Dicke model \cite{Dicke1954coherence,Hepp1973superradiant}, as well as both U(1) and SU(2) LGTs (although our results do not apply to interacting scalar field theories such as $\phi^4$ theory).
For a system with many bosonic modes or gauge links, the truncation error scales with the total number of truncated local variables; therefore the exponentially improved dependence of $\Lambda$ on the precision also implies exponentially improved scaling of $\Lambda$ with the total system size. To illustrate the improvement, Figure~\ref{fig:m_compare_bounds} compares our truncation threshold with the energy-based estimate for the case of the Hubbard-Holstein model. See Section \ref{sec:comparison_with_the_energy_based_truncation_threshold} in the Appendix for a more detailed comparison.
We further establish a threshold for truncating the Hamiltonian such that the time evolution is provably accurate when the initial state is assumed to have low local quantum numbers.

Previous analytical studies of the truncation problem have been mostly restricted to simple models, while only limited small-scale numerical results are available for more complicated systems \cite{MacridinEtAL2018digital,klco2019digitization}.
For instance, Ref.\ \cite{Somma2015OneDimensional} proposed one such method for simulating a single quantum harmonic oscillator.
In Ref. \cite{MacridinEtAL2018digital} the authors argued via the Nyquist-Shannon sampling theorem that for a single bosonic mode with an occupation number cutoff, a grid discretization leads to exponentially small error. This argument was further extended in Ref.\ \cite{klco2019digitization} to the setting of scalar field theories.
The occupation number cutoff is justified by considering a forced Harmonic oscillator, for which analytic solution can be obtained. However, the model of a forced Harmonic oscillator does not cover all features of boson-fermion interaction, because by modeling the interaction between the bosonic mode and the rest of the system by a time-dependent force, this model ignores the entanglement between the two parts of the system.
To the best of our knowledge, the framework we develop provides the first exponential accuracy guarantee for truncating a wide range of unbounded quantum systems of physical interest.

The new truncation threshold enables us to more accurately analyze the computational cost of simulating dynamical evolution in the systems mentioned above.Although we will mainly consider applications of our result to quantum simulation, our techniques can be used to determine truncation threshold for classical simulation as well. Using standard estimates, the simulation cost typically depends on norms of local terms in the Hamiltonian, which are formally infinite in bosonic systems and LGTs. We can obtain a tighter estimate by considering evolution governed by a truncated Hamiltonian acting on the truncated Hilbert space. We focus specifically on digital quantum simulation of time evolution in the Hubbard-Holstein model and the U(1) and SU(2) LGTs. For the latter, by adapting the simulation algorithm of \cite{HaahHastingsKothariLow2021quantum} to our truncated Hamiltonian, we find a gate complexity that scales almost linearly with the spacetime volume. In doing so, we establish a constant Lieb-Robinson velocity for LGT models which is essential for the method of \cite{HaahHastingsKothariLow2021quantum} and may be of independent interest. We also observe that there are Hamiltonians in the class we consider such that the gate complexity of simulation for time $T$ is $\wt{\Omega}(T^2)$
\footnote{For functions of real variables $f,g$, we write \unexpanded{$f=\Or(g)$} if there exist \unexpanded{$c,t_0>0$} such that \unexpanded{$|f(T)|\leq c|g(T)|$} for all \unexpanded{$|T|\geq t_0$}. When there is no ambiguity, we will use \unexpanded{$f=\Or(g)$} to also represent that \unexpanded{$|f(\tau)|\leq c|g(\tau)|$} holds for all \unexpanded{$\tau\in\R$}. We then extend the definition of \unexpanded{$\mathcal{O}$} to functions of positive integers and multivariate functions. For example, we use \unexpanded{$f(N,T,1/\epsilon)=\Or((NT)^2/\epsilon)$} to mean that \unexpanded{$|f(N,T,1/\epsilon)|\leq c(N|T|)^2/\epsilon$} for some \unexpanded{$c,n_0,t_0,\epsilon_0>0$} and all \unexpanded{$|T|\geq t_0$}, \unexpanded{$0<\epsilon<\epsilon_0$}, and integers \unexpanded{$N\geq n_0$}. We write \unexpanded{$f=\Omega(g)$} if \unexpanded{$g=\Or(f)$} and we use \unexpanded{$\widetilde{\mathcal{O}}$} to suppress logarithmic factors in the asymptotic expression}, in stark contrast to the $\wt{\Or}(T)$ cost that applies when local Hilbert spaces are  finite dimensional \cite{BerryChildsCleveEtAl2014exponential,BerryChildsKothari2015,low2017optimal}. The cost can increase quadratically with $T$ in cases where local quantum numbers rise without bound as $T$ increases.

Although our main focus here is on the cost of dynamical simulation, our bounds on truncation error also have consequences for approximating eigenstates of the ideal untruncated Hamiltonian within the truncated Hilbert space.
For energy eigenvalues separated from the rest of the spectrum by a specified gap, we derive a ``tail bound'' showing that the corresponding eigenstates have very little support on large values of the local quantum numbers. It follows that, for the class of models we study, a truncation error less than $\epsilon$ can be achieved with truncation threshold $\Lambda= \polylog(\epsilon^{-1})$, in contrast with the more naive estimate $\Lambda=\poly(\epsilon^{-1})$ obtained using energy-based methods. 

In our analysis of the cost of simulating time evolution, we assume that in the initial state all local quantum numbers lie within a bounded range, and then derive bounds on how much the local quantum numbers can increase during time evolution. Our focus is somewhat related to previous work using conservation of energy or particle number to tighten the analysis of Trotter product formulas \cite{SahinogluSomma2020hamiltonian,SuHuangCampbell2021nearly}, but our techniques differ from previous works in that we need to deal with non-conserved quantities and unbounded local terms, the latter of which makes the main tools in \cite{SahinogluSomma2020hamiltonian}, namely Lemmas 1 and 2, no longer apply.
Our bounds also have potential applications to error mitigation in quantum simulations, as an unexpectedly large value of a local quantum number might flag an error that occurred during execution of the simulation algorithm. Similar proposals have been based on conserved quantities \cite{McArdle2019ErrorMitigated,BonetMonroigEtAl2018low,Sawaya2016Error,Huggins2021Virtual}, and here we note that the same idea can be applied to non-conserved quantities if we can rigorously bound the growth of those quantities during a specified time interval. 

Quantum simulations of non-abelian LGTs should eventually enable us to probe particle physics in regimes where classical simulations are intractable. Therefore the computational cost of such simulations is of fundamental interest. Though for the sake of concreteness we focus on SU(2) in this work, we anticipate that similar conclusions apply for other non-abelian gauge groups, including SU(3), the relevant case for quantum chromodynamics. We emphasize, though, that our results apply to LGTs where the lattice spacing is a fixed physical length; we have not studied the approach to the continuum limit or other formulations of quantum field theories without using lattices \cite{Liu2020}. We also emphasize that our analysis of the cost of simulating dynamics assumes that the initial state is well approximated by a state in which all local quantum numbers are less than the truncation threshold $\Lambda$; for appropriate initial states, for example when the initial state is a superposition of low-energy eigenstates, this assumption might be justified by our tail bounds. However, we do not consider the computational cost of the initial state preparation \cite{Moosavian2019Site}. Despite these important caveats, our findings strengthen the expectation that quantum computers will become powerful instruments for scientific discovery. 

\begin{figure}
    \centering
    \subfloat[]{
    \includegraphics[width=0.4\textwidth]{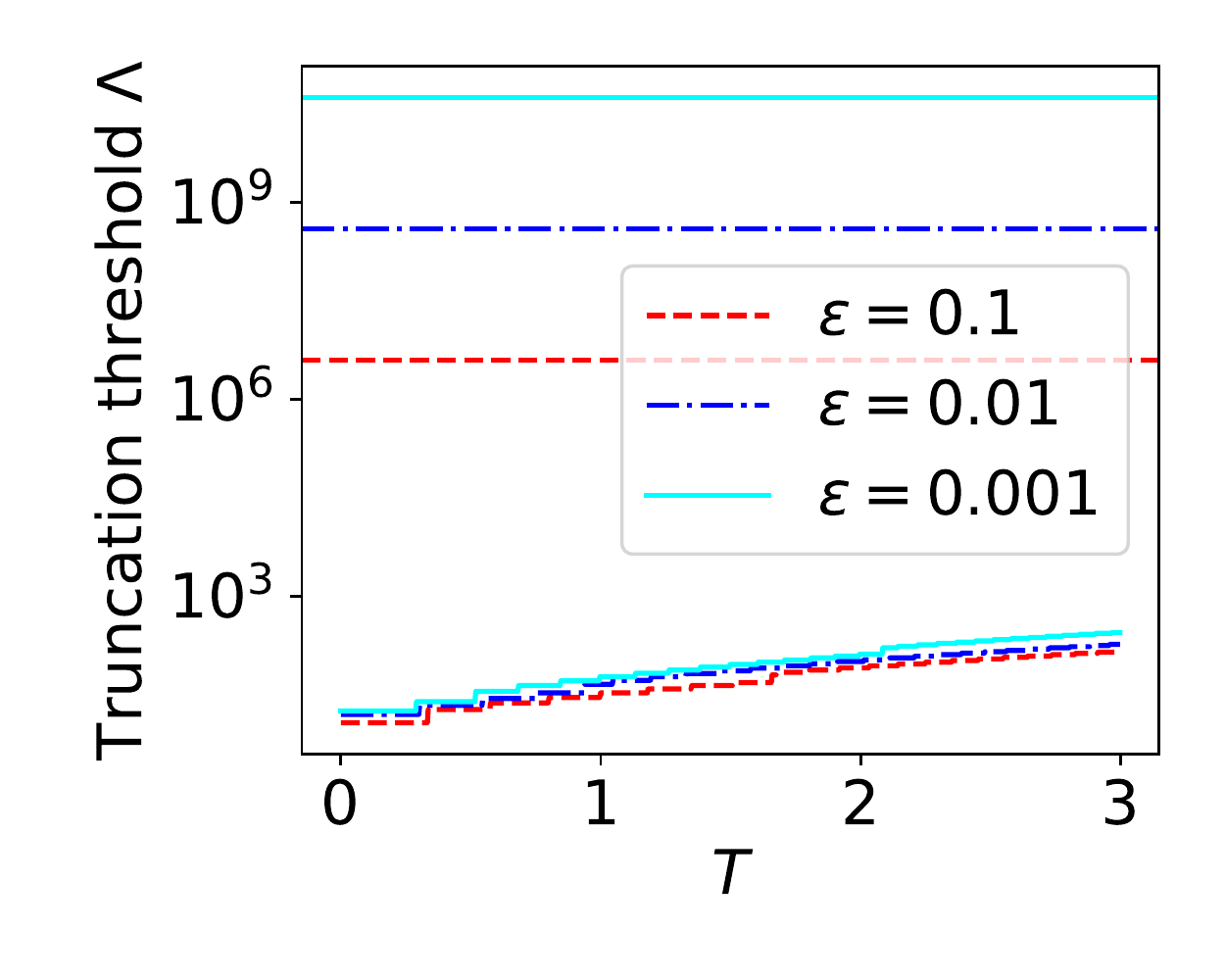}
    }
    \hfill
    \subfloat[]{
    \includegraphics[width=0.4\textwidth]{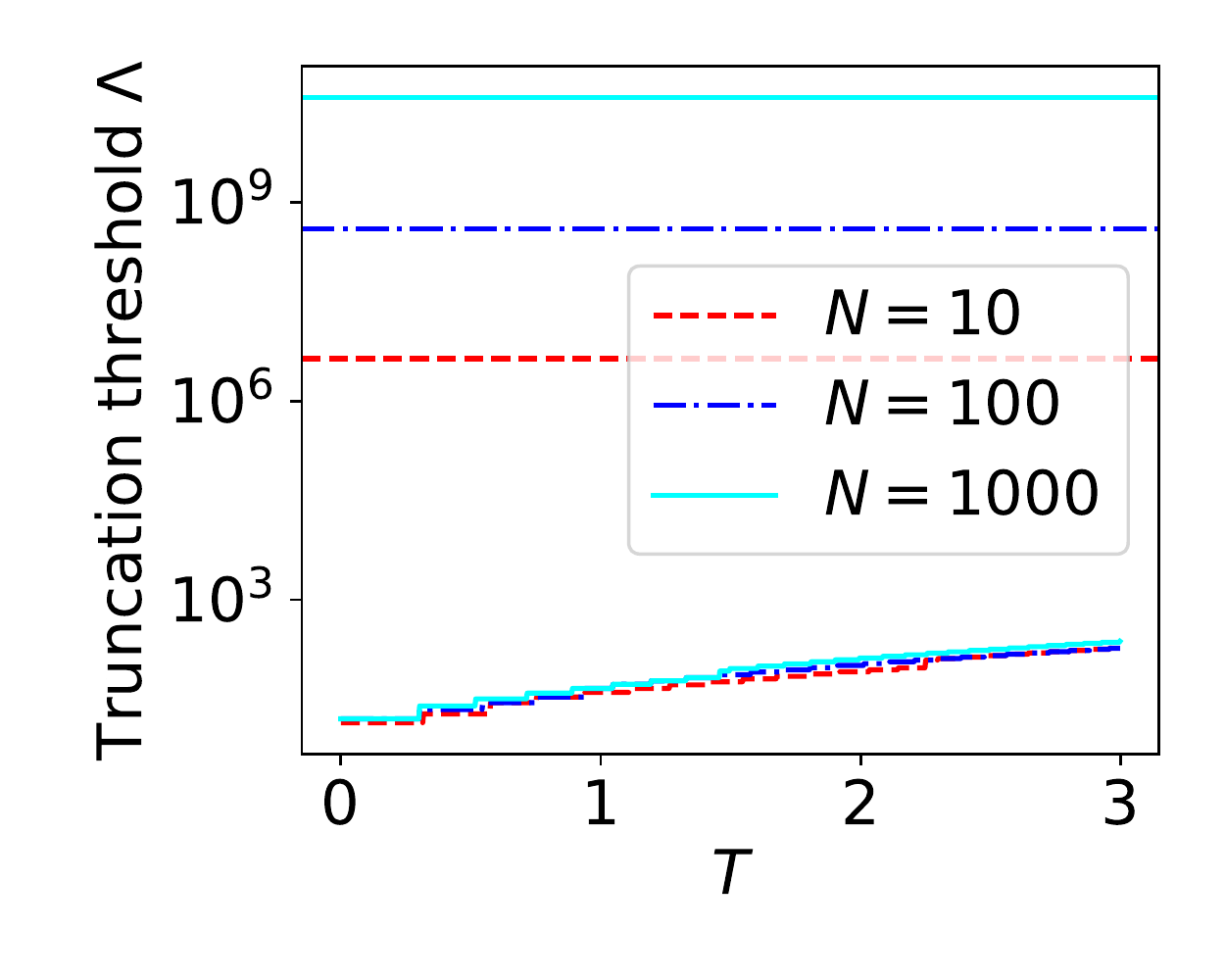}
    }
   \caption{Truncation threshold $\Lambda$ that achieves precision $\epsilon$ for the Hubbard-Holstein model with $N$ lattice sites and evolution time $T$.  (a) $\Lambda$ as a function of $T$ for $N=100$ and three different values of $\epsilon$. (b) $\Lambda$ as a function of $T$ for $\epsilon = .01$ and three different values of $N$. \changed{The Hamiltonian is given in \eqref{eq:m_ham_hubbard_holstein}.}
  Model parameters are from \cite{KlossReichmanTempelaar2019multiset}\changed{, with $\omega_0=1$ and $g=0.5$}. The horizonal lines are the time-independent truncation thresholds obtained using an energy-based method as in \cite{JordanLeePreskill2012quantum}; the other curves are values of $\Lambda$ obtained in this work. See Section \ref{sec:compare_thresholds_Hubbard_Holstein} in the Appendix for details.}
    \label{fig:m_compare_bounds}
\end{figure}

\section{Framework}
\label{sec:framework}
We begin by setting up our framework and concisely stating our results, to be proven in subsequent sections. \changed{For a more formal introduction of the framework see Sections \ref{sec:motivating_examples} and \ref{sec:common_structures}.} To illustrate our framework in a concrete setting, we first consider the Hubbard-Holstein model \cite{holstein1959studies}, 
a model of electron-phonon interactions. 
The model is defined on a $D$-dimensional lattice with linear size $L$ and $L^D=N$ sites.
Each site in the lattice, indexed by $x$, contains two fermionic modes (spin up and down) and a bosonic mode. 
The Hamiltonian is
\begin{equation}
    \label{eq:m_ham_hubbard_holstein}
    H = H_f + H_{fb} + H_b,
\end{equation}
where $H_{f}$ is the Hamiltonian of the Fermi-Hubbard model \cite{Hubbard1963electron} acting on only the fermionic modes,
and
\begin{equation}
    H_{fb} =  g\sum_{x=1}^N (b_x^{\dagger}+b_x)(n_{x,\uparrow}+n_{x,\downarrow}-1), \quad 
    H_b = \omega_0 \sum_{x=1}^N b_{x}^{\dagger}b_x,
\end{equation}
are the boson-fermion coupling and purely bosonic parts of the Hamiltonian respectively. Here, $b_x$ is the bosonic annihilation operator on site $x$, and $n_{x,\sigma}$ is the fermionic number operator for site $x$ and spin $\sigma$.

In this setting, the local Hilbert space of each bosonic mode is infinite-dimensional. In order to have a finite-dimensional local Hilbert space, a natural idea is to 
impose an upper limit $\Lambda$ on the occupation number $b_x^\dagger b_x$ (``number of particles'') in each bosonic mode. Then each bosonic local Hilbert space has dimension $\Lambda + 1$, and is spanned by the particle-number eigenstates $\{|\lambda\rangle: \lambda = 0, 1,2, \dots, \Lambda\}$. 
\changed{This imposed upper limit results in a truncation. For a fixed site $x$, we define the projection operator $\Pi^{(x)}_{[0,\Lambda]}=\sum_{\lambda=0}^{\Lambda}\ket{\lambda}\bra{\lambda}_x\otimes I$, where $I$ is the identity operator acting on the rest of the system. Imposing the upper limit truncates a quantum state $\ket{\phi}$ to be $\Pi^{(x)}_{[0,\Lambda]}\ket{\phi}$, and truncates the Hamiltonian $H$ to be $\Pi^{(x)}_{[0,\Lambda]}H\Pi^{(x)}_{[0,\Lambda]}$.}
One may ask how large $\Lambda$ should be for the resulting truncation error to be smaller than $\epsilon$.
There is some ambiguity regarding what ``truncation error'' means, and we will refine this question later. 

A similar situation is encountered in LGTs, where we consider the Hamiltonian formulation proposed in \cite{kogut1975hamiltonian}.
\changed{For a more detailed introduction to the LGTs we consider, see Section \ref{sec:motivating_examples}.}
We have a $D$-dimensional lattice 
consisting of $N$ total sites and $\Or(N)$ gauge links. 
Each gauge link has an infinite-dimensional local Hilbert space, and the Hamiltonian contains unbounded operators associated with each link. We no longer have a natural notion of particle number, but the truncation of the link Hilbert space can still be performed according to what we call the \textit{local quantum number}. We focus on two cases: the U(1) and SU(2) LGTs. For the U(1) case, we choose the local quantum number to be the integer-valued electric field. We retain only the part of local Hilbert space with electric field value in the interval $[-\Lambda,\Lambda]$; hence the truncated local Hilbert space at each link is $2\Lambda+1$ dimensional. \changed{More precisely, for a fixed gauge link $\nu$, we define the projection operator $\Pi^{(\nu)}_{[0,\Lambda]}=\sum_{\lambda=-\Lambda}^{\Lambda}\ket{\lambda}\bra{\lambda}_{\nu}\otimes I$, where $\ket{\lambda}$'s are elements of the electric eigenbasis. Similar to the previous example of the Hubbard-Holstein model, after truncation a quantum state $\ket{\phi}$ becomes $\Pi^{(\nu)}_{[-\Lambda,\Lambda]}\ket{\phi}$, and the Hamiltonian $H$ becomes $\Pi^{(\nu)}_{[-\Lambda,\Lambda]}H\Pi^{(\nu)}_{[-\Lambda,\Lambda]}$.} For the SU(2) case, we choose the local quantum number to be 2 times the total angular momentum (the multiplication by 2 makes the local quantum number an integer). If we retain only the part of the link Hilbert space with total angular momentum no larger than $\Lambda/2$, then the link Hilbert space has dimension $(\Lambda+1)(\Lambda+2)(2\Lambda + 3)/6$, and is spanned by the angular momentum eigenstates $\ket{jmm'}$ where $j$ is a half integer less than or equal to $\Lambda/2$ and $-j\leq m,m'\leq j$. Again one may ask how large $\Lambda$ should be for the resulting truncation error to be small.

When analyzing time evolution, this question can be refined into two different but related questions.

\vspace{1em}
\textbf{Question 1 (Truncating an evolved quantum state):} Consider an initial state such that at some particular site or link the local quantum number is no larger than $\Lambda_0$. After the state evolves forward for time $T$, how should the truncation threshold $\Lambda$ be chosen for that site or link so that the resulting error is at most $\epsilon$? 
We show that 
\begin{equation}
    \label{eq:m_truncation_threshold_scaling_state}
    \Lambda^{1-r} = \Lambda_0^{1-r} + \wt{\Or}((\chi T+1)\polylog(\epsilon^{-1}))
\end{equation}
suffices, where $r=1/2$ for bosons and $r=0$ for LGTs, and $\chi$ is a constant that only depends on the model parameters but not on the system size or on $T$. 
If we want to truncate every bosonic mode or gauge link in the model, then, to account for the accumulation of error, $\epsilon^{-1}$ in \eqref{eq:m_truncation_threshold_scaling_state} is replaced by $N\epsilon^{-1}$, where $N$ is the system size. 
\vspace{1em}

A simple example shows that this scaling of the truncation threshold is optimal in certain cases, such as quadratic scaling in time for bosons $(r=1/2)$. Suppose $H=b+b^{\dagger}$ where $b$ is the annihilation operator of a bosonic mode. Then $e^{-iTH}\ket{0}$, where $\ket{0}$ is the vacuum state, is a coherent state such that the  particle number distribution is Poissonian with mean $T^2$. Because the Poisson distribution concentrates around its mean, a truncation threshold that achieves constant precision must scale like $\Omega(T^2)$, which matches \eqref{eq:m_truncation_threshold_scaling_state} for $r=1/2$.

It is instructive to compare our approach with the method based on energy conservation described in \cite{JordanLeePreskill2012quantum,JordanLeePreskill2014QuantumCompScatter}. That method yields a truncation threshold for a single site with a polynomial dependence on the inverse accuracy $\epsilon^{-1}$. To truncate a system of $\wt{\mathcal{O}}(N)$ sites, we scale down $\epsilon$ by a factor of $N$, resulting in a threshold $\Lambda$ scaling polynomially with $N\epsilon^{-1}$. In contrast, our bound has only polylogarithmic dependence on $N\epsilon^{-1}$, an exponential improvement compared to the truncation threshold obtained using the energy-based method. 
Importantly, this advantage holds not only in the asymptotic regime, but also when the constant prefactors are incorporated. We numerically compare our bound with the energy-based bound for 
the Hubbard-Holstein model, observing a significantly better estimate in various parameter regimes. We illustrate this comparison in Figure~\ref{fig:m_compare_bounds} and discuss it in more detail in Section \ref{sec:comparison_with_the_energy_based_truncation_threshold} in the Appendix. 

\vspace{1em}
\textbf{Question 2 (Truncating the Hamiltonian):}  Consider an initial state such that the local quantum number is no larger than $\Lambda_0$ at all sites or at all links, and suppose the state evolves forward for time $T$ using a truncated Hamiltonian $\wt{H}$ rather than the ideal untruncated Hamiltonian $H$ (we will define $\wt{H}$ later).  How should the truncation threshold $\Lambda$ be chosen so that the truncated evolved state matches the ideal evolved state up to time $T$ with an error at most $\epsilon$?  
We show that
\begin{equation}
    \label{eq:m_truncation_threshold_scaling_hamiltonian}
    \Lambda^{1-r} = \Lambda_0^{1-r} + \wt{\Or}((\chi T+1)\polylog(N\epsilon^{-1}))
\end{equation}
suffices, where $r=1/2$ for bosons and $r=0$ for LGTs, $N$ is the system size, and $\chi$ is again a constant that does not depend on $N$ or on $T$.
\vspace{1em}

Our above two questions both concern the truncation of the local quantum number, albeit from different perspectives: the first focuses on evolved quantum states while the second focuses on Hamiltonians. In fact, a threshold for truncating the Hamiltonian can be directly used to truncate evolved quantum states, although some extra efforts are required to handle the converse. 
The truncation of an evolved quantum state in Question 1 is only for some fixed time $T$, but when we perform truncation on the Hamiltonian in Question 2, the evolved quantum state will never have a local quantum number beyond $[-\Lambda,\Lambda]$ throughout the evolution up to time $T$. In this sense, our second result is stronger than the first one.

\changed{It is worth noting that while the energy-based method in Refs.\cite{JordanLeePreskill2012quantum} is enough to establish a bound to address Question 1, it cannot be used to address Question 2.
This is because the state truncation error does not decay fast enough as we increase the truncation threshold, and as a result the energy-based bound is not enough for the derivation in Section \ref{sec:truncating_the_hamiltonian}.}

Before stating our results for a more general class of Hamiltonians, we first introduce some notation.
For a bosonic mode or gauge link which we denote by $\nu$, we denote by $\Pi^{(\nu)}_S$ the projection operator imposing the condition that the local quantum number takes values from the set $S$. We also denote $\Pi^{\mathrm{all}}_S=\prod_{\nu} \Pi^{(\nu)}_S$; this is the projection operator imposing the condition on all bosonic modes or gauge links. For any projection operator $\Pi$, we write its complement as $\overline{\Pi}=I-\Pi$. The truncated Hamiltonian mentioned in Question 2 is $\wt{H}=\Pi^{\mathrm{all}}_S H \Pi^{\mathrm{all}}_S$, where $H$ is the untruncated Hamiltonian, and $S$ is the set of local quantum numbers less than or equal to the truncation threshold.

Using this notation we can readily pinpoint the common structure of the Hamiltonians in the Hubbard-Holstein model and the U(1) and SU(2) LGTs. In all three examples, although the Hamiltonian contains local terms with unbounded norm, each of these terms changes the local quantum number at only a single site or a single link; there are no unbounded terms that allow the local quantum number to propagate from site to site or from link to link. For each site or link, denoted by $\nu$, we may write the full Hamiltonian $H$ of the model as
\begin{equation}
\label{eq:Ham-form-WR}
H=H_W^{(\nu)}+H_R^{(\nu)},
\end{equation}
where $H_W^{(\nu)}$ is the part of the Hamiltonian that can change the value of the local quantum number at $\nu$, and $H_R^{(\nu)}$ contains all the terms in the Hamiltonian that preserve the value of the local quantum number at $\nu$.
These two parts satisfy the conditions
\begin{subequations}
\label{eq:m_conditions_ham_general}
\begin{align}
    &\Pi^{(\nu)}_{\lambda}H_W^{(\nu)} \Pi^{(\nu)}_{\lambda'}=0,\ \text{if }|\lambda-\lambda'|>1, \label{eq:m_conditions_ham_general_1}\\
    &\|H_W^{(\nu)}\Pi^{(\nu)}_{[-\Lambda,\Lambda]}\|\leq \chi(\Lambda+1)^{r}, \label{eq:m_conditions_ham_general_2}\\
    &[H_R^{(\nu)},\Pi^{(\nu)}_{\lambda}]=0. \label{eq:m_conditions_ham_general_3}
\end{align}
\end{subequations}
Here $\Pi^{(\nu)}_{\lambda}$ projects onto the eigenspace with local quantum number $\lambda$, $\chi$ and $0\leq r<1$ are parameters that depend on the model, and $\|\cdot\|$ is the spectral norm.
(The notation $\Pi^{(\nu)}_{[-\Lambda,\Lambda]}$ is appropriate for the U(1) gauge theory, where the electric field can take either positive or negative integer values, but we will use this same notation for the other models as well, even though in those models the local quantum number takes only nonnegative values.)
These three conditions can be interpreted as follows: the first condition requires $H_W^{(\nu)}$ to change the local quantum number by at most $\pm 1$. The second condition requires that the rate at which the maximal local quantum number $\Lambda$ changes is sublinear in $\Lambda$.
The third condition requires $H_R^{(\nu)}$ to preserve the local quantum number.
See Section \ref{sec:common_structures} in the Appendix for a more detailed explanation of this framework.

Let us verify that the Hubbard-Holstein Hamiltonian in \eqref{eq:m_ham_hubbard_holstein} fits the general framework of \eqref{eq:Ham-form-WR} and \eqref{eq:m_conditions_ham_general}. The bosonic mode appears only in on-site terms. Choosing $H_W^{(x)} = g(b_x^{\dagger}+b_x)(n_{x,\uparrow}+n_{x,\downarrow}-1)$ and 
$H_R^{(x)} = \sum_{x'}\omega_0 b_{x'}^\dagger b_{x'} + \sum_{x'\neq x}g(b_{x'}^{\dagger}+b_{x'})(n_{x',\uparrow}+n_{x',\downarrow}-1)+H_f$, 
we see that $H_W^{(x)}$ changes the local bosonic particle number by at most $\pm 1$, and that $H_R^{(x)}$ preserves the local bosonic particle number. Moreover, using $\Pi^{(x)}_{\lambda}$ to denote the projector onto the subspace with $\lambda$ bosonic particles on site $x$, we see that $\|(b_x^{\dagger}+b_x)\Pi^{(x)}_{[0,\Lambda]}\|\leq 2\sqrt{\Lambda+1}$, which implies \eqref{eq:m_conditions_ham_general_2} is satisfied with $\chi=2g$ and $r=1/2$. In Sections \ref{sec:motivating_examples} and \ref{sec:other_examples} in the Appendix, we explain how other examples fit this framework, including U(1) and SU(2) LGTs,  the spin-fermion coupling in the Fr{\"o}hlich model \cite{Frohlich1954electrons}, and spin-boson coupling in the Dicke model. 
In Section~\ref{sec:m_hilbert_space_truncation_in_time_evolution} we show that for Hamiltonians with the structure indicated in \eqref{eq:Ham-form-WR}, \eqref{eq:m_conditions_ham_general}, local quantum numbers may be truncated as specified by the answer \eqref{eq:m_truncation_threshold_scaling_state} to Question 1 and the answer \eqref{eq:m_truncation_threshold_scaling_hamiltonian} to Question 2.

The linear dependence on the evolution time $T$ in \eqref{eq:m_truncation_threshold_scaling_state}, \eqref{eq:m_truncation_threshold_scaling_hamiltonian} has a simple interpretation. Specifically, for the case of a bosonic mode ($r=1/2$) where $H_W^{(\nu)}$ is linear in creation and annihilation operators, the conditions \eqref{eq:m_conditions_ham_general_1}, \eqref{eq:m_conditions_ham_general_2} impose that in time $T$ the position of the mode in phase space is translated by $\mathcal{O}(T)$. Since the particle number scales like the square of the displacement from the origin of phase space, a truncation threshold growing quadratically with $T$, as specified in \eqref{eq:m_truncation_threshold_scaling_state}, \eqref{eq:m_truncation_threshold_scaling_hamiltonian}, suffices to approximate the translated state accurately.

Given the scaling of the truncation threshold expressed in \eqref{eq:m_truncation_threshold_scaling_hamiltonian}, we can accurately approximate time evolution using the truncated Hamiltonian $\wt{H}$, in which all local terms in the Hamiltonian have bounded norm. In Section~\ref{sec:m_application_to_ham_sim}, we leverage this observation to analyze the cost of simulating time evolution on a digital quantum computer for the Hamiltonians characterized above. 
In particular, we develop algorithms for simulating the U(1) and SU(2) LGTs that achieve an almost linear dependence on the spacetime volume, a substantial improvement over previous estimates of the gate complexity \cite{ShawEtAl2020Schwinger,KanNam2021lattice,AbhishekRoggeroWiebe2013hybridized}. We also analyze the cost of simulating the Hubbard-Holstein model in Section~\ref{sec:m_application_to_ham_sim}. In Section~\ref{sec:m_eigenstate_tail_bound}, by applying these results on time evolution, we establish that spectrally isolated energy eigenstates can be approximated using a local quantum number truncation threshold scaling polylogrithmically with the allowed error.

\section{Hilbert space truncation in time evolution}
\label{sec:m_hilbert_space_truncation_in_time_evolution}

We now show how the truncation threshold scaling relations \eqref{eq:m_truncation_threshold_scaling_state} and \eqref{eq:m_truncation_threshold_scaling_hamiltonian} are obtained. Recall that in our two questions about the truncation threshold, Question~1 concerns truncating the quantum state obtained from exact time evolution for time $T$. Using the notations introduced earlier, we can clarify this question and our result. We define a quantity $\|\overline{\Pi}^{(\nu)}_{[-\Lambda,\Lambda]}e^{-iTH}\Pi^{(\nu)}_{[-\Lambda_0,\Lambda_0]}\|$ which we call \textit{leakage}. If we start with an initial state $\ket{\psi_0}$ with the local quantum number on $\nu$ taking a value in $[-\Lambda_0,\Lambda_0]$, and denote the state at time $t$ by $\ket{\psi(t)}$, then the truncation error $\|\overline{\Pi}^{(\nu)}_{[-\Lambda,\Lambda]}\ket{\psi(T)}\|$ is upper bounded by the leakage. Therefore, to ensure that the truncation error is at most $\epsilon$, we only need to keep the leakage below $\epsilon$.

As mentioned before, we assume $H$ has the structure \eqref{eq:Ham-form-WR} with $H^{(\nu)}_W$ and $H^{(\nu)}_R$ satisfying \eqref{eq:m_conditions_ham_general}.
First we prove a leakage bound that holds for relatively short evolution time governed by such $H$, and then establish \eqref{eq:m_truncation_threshold_scaling_state} by extending the short-time leakage bound to longer times. We view the time evolution in the interaction picture, and consider the evolution of $\ket{\psi_I(t)}=e^{itH^{(\nu)}_R}\ket{\psi(t)}$. Because we assume $H^{(\nu)}_R$ preserves the local quantum number, $\ket{\psi_I(t)}$ and $\ket{\psi(t)}$ induce the same local quantum number distribution. In the interaction picture, $\ket{\psi_I(t)}$ evolves with a time-dependent Hamiltonian $H^{(\nu)}_{W}(t)=e^{itH^{(\nu)}_R}H^{(\nu)}_{W}e^{-itH^{(\nu)}_R}$. We then apply the Dyson series expansion to the unitary operator generated by $H^{(\nu)}_{W}(t)$. In the proof of Lemma \ref{lem:short_time_bound} in the Appendix, we show that if $0\leq T\leq 1/(2\chi(\Lambda_0+1)^r)$, the truncated Dyson series with $\Delta$ terms approximates the exact evolution up to an error $e^{-\Omega(\Delta)}$. Moreover, such a truncated Dyson series can change the local quantum number by at most $\pm(\Delta-1)$ due to 
\eqref{eq:m_conditions_ham_general_1}. Therefore we have the short-time leakage bound
\begin{equation}
\label{eq:m_short_time_bound}
    \|\overline{\Pi}^{(\nu)}_{(-\Lambda_0-\Delta,\Lambda_0+\Delta)}e^{-iTH}\Pi^{(\nu)}_{[-\Lambda_0,\Lambda_0]}\|\leq e^{-\Omega(\Delta)}.
\end{equation}

Using this short-time leakage bound, we can derive the long-time bound in \eqref{eq:m_truncation_threshold_scaling_state}. Specifically, for any choice of $\Lambda_0<\Lambda_1<\cdots<\Lambda_J=\Lambda$, $0=T_0<T_1<\cdots<T_J=T$, the total leakage is at most the sum of $J$ short-time leakages (see Lemma \ref{lem:decompose_time_evolution} in the Appendix).
\begin{equation}
\label{eq:m_error_accumulation}
\begin{aligned}
    &\|\overline{\Pi}^{(\nu)}_{[-\Lambda,\Lambda]} e^{-iTH} \Pi^{(\nu)}_{[-\Lambda_0,\Lambda_0]}\| \\
    &\leq \sum_{j=1}^J \|\overline{\Pi}^{(\nu)}_{[-\Lambda_j,\Lambda_j]} e^{-i(T_j-T_{j-1})H}\Pi^{(\nu)}_{[-\Lambda_{j-1},\Lambda_{j-1}]}\|.
\end{aligned}
\end{equation}
We then carefully choose $T_j$'s and apply the short-time leakage bound to each segment $[T_{j-1},T_j]$, which gives an upper bound on the right-hand side of \eqref{eq:m_error_accumulation}.
Since the local quantum number can potentially change as the system evolves, we define the length of time steps adaptively based on the instantaneous quantum number to reach the same target accuracy.
Specifically, $T_j$ and $\Lambda_j$ are chosen to satisfy $0\leq T_j-T_{j-1}\leq 1/(2\chi(\Lambda_{j-1}+1)^r)$.
This establishes the scaling in \eqref{eq:m_truncation_threshold_scaling_state} and provides an answer to Question~1.
We summarize our result below and leave details of the proof to Section \ref{sec:leakage_time_evo} in the Appendix.

\vspace{1em}
\noindent\textbf{Theorem} (State truncation (Theorem \ref{thm:long_time_bound2} in the Appendix)).
Let $H$ be a Hamiltonian such that $H=H_W^{(\nu)}+H_R^{(\nu)}$ satisfies \eqref{eq:m_conditions_ham_general} with parameters $\chi$ and $r$ for a fixed mode or link $\nu$. 
For any $t\geq 0$ and integers $\Lambda\geq \Lambda_0\geq0$,
\begin{equation}
\begin{aligned}
    &\left\|\overline{\Pi}_{[-\Lambda,\Lambda]}^{(\nu)}e^{-itH}\Pi_{[-\Lambda_0,\Lambda_0]}^{(\nu)}\right\| \\
    &\leq \poly(\chi t,\Lambda_0,\Lambda)\exp\left(-\Omega\left(\frac{\Lambda^{1-r}-\Lambda_0^{1-r}}{\chi t+1}\right)\right).
\end{aligned}
\end{equation}
\vspace{1em}

We now set out to answer Question~2. First we clarify the question using our notation for projection operators. Here we consider replacing $H$ by a truncated Hamiltonian $\wt{H}=\Pi^{\mathrm{all}}_{[-\Lambda,\Lambda]}H\Pi^{\mathrm{all}}_{[-\Lambda,\Lambda]}$, where 
\changed{
\[
\Pi^{\mathrm{all}}_{[-\Lambda,\Lambda]} = \prod_{\nu}\Pi^{(\nu)}_{[-\Lambda,\Lambda]}
\]
applies truncation on all sites or links.}
the truncation threshold $\Lambda$ is chosen large enough so that evolution governed by $\wt{H}$ is a good approximation to the exact evolution.  
The approximation error is upper bounded by 
\begin{equation}
    \max_{0\leq t\leq T}\|(e^{-itH}-e^{-it\wt{H}})\Pi^{\mathrm{all}}_{[-\Lambda_0,\Lambda_0]}\|.
\end{equation}
Therefore our goal is to choose $\Lambda$ to ensure that this error is at most $\epsilon$. 
This is accomplished by the following theorem which we establish in Section \ref{sec:truncating_the_hamiltonian} in the Appendix and preview here.

\vspace{1em}
\noindent\textbf{Theorem} (Hamiltonian truncation (Theorem \ref{thm:truncate_ham} in the Appendix)).
Let $H$ be a Hamiltonian with $\Or(N)$ bosonic modes or gauge links, such that $H=H_W^{(\nu)}+H_R^{(\nu)}$ satisfies \eqref{eq:m_conditions_ham_general} with parameters $\chi$ and $r$ for every mode or link $\nu$.
Assume that all projection operators $\Pi^{(\nu)}_\lambda$ commute with each other.
For any integers $\Lambda\geq\Lambda_0\geq 0$, define
$\wt{H} = \Pi^{\mathrm{all}}_{[-\Lambda,\Lambda]}H \Pi^{\mathrm{all}}_{[-\Lambda,\Lambda]}$
and assume $\|[\wt{H},H]\|=\poly(N,\Lambda)$.
Then for any $t>0$,
\begin{equation}
\begin{aligned}
    &\left\|(e^{-it\wt{H}}-e^{-itH})\Pi^{\mathrm{all}}_{[-\Lambda_0,\Lambda_0]}\right\| \\
    &\leq \poly(\chi t,\Lambda_0,\Lambda,N)\exp\left(-\Omega\left(\frac{\Lambda^{1-r}-\Lambda_0^{1-r}}{\chi t+1}\right)\right).
\end{aligned}
\end{equation}
\vspace{1em}

We now briefly explain how we upper bound the Hamiltonian truncation error $\|(e^{-itH}-e^{-it\wt{H}})\Pi^{\mathrm{all}}_{[-\Lambda_0,\Lambda_0]}\|$. 
Since this is equal to $\|(I-e^{itH}e^{-it\wt{H}})\Pi^{\mathrm{all}}_{[-\Lambda_0,\Lambda_0]}\|$, we instead show that  $e^{itH}e^{-it\wt{H}}\Pi^{\mathrm{all}}_{[-\Lambda_0,\Lambda_0]}$ is close to $\Pi^{\mathrm{all}}_{[-\Lambda_0,\Lambda_0]}$ for sufficiently large $\Lambda$. We prove this by expanding the target quantity using the formula for Trotter error \cite[Eq.~(3.4)]{suzuki1985decomposition}. We then invoke the leakage bound Theorem \ref{thm:long_time_bound2} in the Appendix, as well as the fact that the truncated Hamiltonian and the original Hamiltonian act identically on a state with local quantum numbers all in $[-\Lambda+1,\Lambda-1]$, i.e. $(H-\wt{H})\Pi^{\mathrm{all}}_{[-\Lambda+1,\Lambda-1]}=0$. The resulting bound depends on the commutator norm $\mathcal{A}(\Lambda)=\|[H,\wt{H}]\|$. If this quantity scales polynomially with the system size $N$ and the truncation threshold $\Lambda$, we can then establish the desired scaling in \eqref{eq:m_truncation_threshold_scaling_hamiltonian} and answer our Question 2. This is indeed true for the Hubbard-Holstein model and the LGTs, which can be seen as follows. We have
$
    [H,\wt{H}] = (H\Pi^{\mathrm{all}}_{[-\Lambda,\Lambda]})^2-(\Pi^{\mathrm{all}}_{[-\Lambda,\Lambda]}H)^2,
$
where each local term in $H\Pi^{\mathrm{all}}_{[-\Lambda,\Lambda]}$ can be bounded by a polynomial of $\Lambda$ (linear for Hubbard-Holstein model and quadratic for LGTs) and there are $\Or(N)$ of such local terms. Thus we have $\|[H,\wt{H}]\|=\Or(N^2\poly(\Lambda))$ in all three examples. 
We discuss the Hamiltonian truncation in more detail in Section \ref{sec:truncating_the_hamiltonian} in the Appendix.

\section{Application to Hamiltonian simulation}
\label{sec:m_application_to_ham_sim}

Our main results on the truncation of unbounded Hamiltonians allow us to simulate such systems more efficiently with a provable accuracy guarantee. For concretenesss, we consider the problem of digital Hamiltonian simulation, wherein the dynamics of a quantum system are approximated on a quantum computer by elementary gates, and the cost of simulation is determined by the gate complexity. While the majority of the past work on Hamiltonian simulation has focused on quantum systems with finite-dimensional local Hilbert spaces, there are also systems of physical interest whose local Hilbert spaces are infinite dimensional. In such cases, it is typically necessary to perform truncation, so that quantum states can be represented and processed on a digital quantum computer.

In Section~\ref{sec:m_hilbert_space_truncation_in_time_evolution} we established that time evolution governed by the Hamiltonian $H$ can be accurately approximated by time evolution governed by the truncated Hamiltonian $\wt{H}$, if we choose the truncation threshold according to \eqref{eq:m_truncation_threshold_scaling_hamiltonian}. Here we propose Hamiltonian simulation approaches that take advantage of our results and discuss the implications for expected costs on a quantum computer.  Unlike $H$, $\wt{H}$ acts nontrivally on only a finite-dimensional subspace of the infinite-dimensional Hilbert space at each site or link. Furthermore, because each unbounded local term in $H$ is truncated separately, the truncation does not affect the geometric locality of the Hamiltonian. Therefore, 
to simulate $e^{-itH}$ we can instead simulate $e^{-it\wt{H}}$, which can be done on a quantum computer using existing simulation techniques for local Hamiltonians \cite{Lloyd1996universal,BerryAhokasCleveSanders2006,BerryChildsKothari2015,low2019hamiltonian,low2017optimal,LowWiebe2018interaction,HaahHastingsKothariLow2021quantum}. 
In what follows, we consider simulations of the U(1), SU(2) LGTs, as well as the Hubbard-Holstein model, although the quantum algorithms we present can in principle be extended to simulate other gauge theories and bosonic systems within our framework.

\vspace{1em}
\noindent\textbf{Simulating lattice gauge theories with near-linear spacetime volume scaling.}
We propose an algorithm to simulate the time evolution of the U(1) and SU(2) LGTs in $D$ spatial dimensions; Hamiltonians of these models are described in Eq.~\eqref{eq:LGT_ham} in the Appendix. The goal is to simulate a lattice with $N$ sites for time $T$ with total error at most $\epsilon$. 
Our algorithm combines the Haah-Hastings-Kothari-Low (HHKL) decomposition \cite{HaahHastingsKothariLow2021quantum}, which provides a nearly optimal approach for geometrically local Hamiltonians, with the interaction-picture simulation method \cite{LowWiebe2018interaction}, which gives further improved scaling with the truncation threshold. 
We show that the simulation can be done with gate complexity $\wt{\Or}(NT\polylog(\Lambda_0 \epsilon^{-1}))$, 
assuming that in the initial state the local quantum number (electric field value for U(1) or total angular momentum for SU(2)) on each gauge link is in the interval $[-\Lambda_0,\Lambda_0]$. Thus we achieve an almost linear dependence of the gate complexity on the spacetime volume $NT$. We briefly outline the algorithm here; further details are presented in Section \ref{sec:sim_LGT} in the Appendix. 

We first use \cite[Lemma~6]{HaahHastingsKothariLow2021quantum} to decompose the time evolution of the entire system due to $\wt{H}$ into time evolution of blocks. Each block, denoted by $\mathcal{B}$, has size $\ell^D=\Or(\polylog(NT\epsilon^{-1}))$ and we only need to implement its evolution for time $\tau=\Or(1)$. There are $\Or(N)$ such blocks and the entire time evolution is divided into $\Or(T)$ segments. We note that \cite[Lemma~6]{HaahHastingsKothariLow2021quantum} requires a constant Lieb-Robinson velocity, which was guaranteed by \cite[Lemma~5]{HaahHastingsKothariLow2021quantum} since all terms in their Hamiltonian were geometrically local with norm upper bounded by a constant. In our case, however, there are terms in the truncated Hamiltonian with norm $\textrm{poly}(\Lambda)$.
Fortunately, these terms with $\Lambda$-dependent norm act on either a single lattice site (in the models with bosonic modes)  or a single gauge link (in LGTs).
We show in Section \ref{sec:LR_veclocity_on_site} in the Appendix that for Hamiltonians of this form, the Lieb-Robinson velocity is indeed bounded above by a $\Lambda$-independent constant as \cite{HaahHastingsKothariLow2021quantum} requires.

When simulating each block $\mathcal{B}$ we use the interaction picture Hamiltonian simulation technique suggested in \cite{Somma2015OneDimensional} and developed in \cite{LowWiebe2018interaction}, and the gate complexity for simulation up to time $\tau=\mathcal{O}(1)$ is $\Or(\polylog(\Lambda NT\epsilon^{-1}))$. For $\Lambda$ we use the scaling \eqref{eq:m_truncation_threshold_scaling_hamiltonian}. There are in total $\Or(NT)$ such simulations that need to be performed, leading to a total gate complexity of $\wt{\Or}(NT\polylog(\Lambda_0 \epsilon^{-1}))$. The interaction picture is useful because it allows us to express the time evolution operator as a product of two operators. One factor in this product is the evolution arising from the terms in the truncated Hamiltonian $\wt{H}$ which have $\Lambda$-dependent norms, the terms involving the electric field at each link. This evolution can be ``fast-forwarded'' \cite{Atia2017FastForwarding,Gu2021FastForwarding} because the Hamiltonian is diagonal in a natural basis, and the evolution operator is just the tensor product of simple unitary operators, each acting on a single link. The other factor in the product is the interaction-picture evolution operator generated by the time-dependent interaction-picture Hamiltonian, in which each term has $\Lambda$-independent norm because the evolution induced by the electric field has been ``rotated away.'' As a result, the cost of simulating the evolution of a block $\mathcal{B}$ is polylogarithmic in $\Lambda$, and the cost of simulating evolution of $N$ sites for time $T$ is nearly linear in the spacetime volume $NT$. 

Previous work on the quantum simulation of LGTs such as \cite{ShawEtAl2020Schwinger,KanNam2021lattice,AbhishekRoggeroWiebe2013hybridized} does not explain how to choose the truncation threshold $\Lambda$ to perform simulation with a provable accuracy. 
While this issue can be remedied by using our Hamiltonian truncation threshold \eqref{eq:m_truncation_threshold_scaling_hamiltonian}, our result still substantially improves over the previous results $\wt{\Or}(N^{3/2}T^{5/2})$ from \cite{ShawEtAl2020Schwinger,KanNam2021lattice} and $\wt{\Or}(N^2T^2)$ from \cite{AbhishekRoggeroWiebe2013hybridized}.

\vspace{1em}
\noindent\textbf{Simulating bosonic systems and an $\wt{\Omega}(T^2)$ gate complexity lower bound.} Here we outline two methods for simulating bosonic systems, using the Hubbard-Holstein model as an example model. In the first method we again use the HHKL decomposition combined with the interaction-picture Hamiltonian simulation;  see Section \ref{sec:application_to_ham_sim_HHKL} in the Appendix for a detailed discussion. The important difference from the setting of LGTs is that, when simulating a block $\mathcal{B}$ of the Hubbard-Holstein model, we cannot get a polylogarithmic dependence on $\Lambda$. Rather, the gate complexity to simulate a block is $\Or(\sqrt{\Lambda}\polylog(\Lambda NT\epsilon^{-1}))$, because, as explained in Section \ref{sec:sim_boson_fermion} in the Appendix, the Hubbard-Holstein Hamiltonian has multiple unbounded terms and it is not known how to fast-forward them simultaneously.
Since there are $\Or(NT)$ blocks to be simulated, and the scaling of $\Lambda$ is given by \eqref{eq:m_truncation_threshold_scaling_hamiltonian}, the total gate complexity is $\wt{\Or}(NT(\sqrt{\Lambda_0}+T)\polylog(\epsilon^{-1}))$.

In the second method we use the $p$-th order Trotter product formula, which can be easier to implement in practice. To obtain a tight error bound in this case one may use the commutation relations among the Hamiltonian terms \cite{ChildsSuTranWiebeZhu2021commutator,ChildsSu2019nearly}. For the Hubbard-Holstein model we use the canonical commutation relation between the bosonic position and momentum operators $[X_{\alpha},P_{\alpha}]=i$, and also invoke geometric locality to tightly bound the error. 
A subtle issue with this naive analysis is that the canonical commutation relation no longer holds when acting on arbitrary states due to the truncation of the Hamiltonian terms. However, we recover the commutation relation by restricting to states with low particle numbers.
A detailed discussion of all the issues involved can be found in Section \ref{sec:boson_fermion_trotter} in the Appendix. In the end we obtain a gate complexity $\wt{\Or}\left(N^{1+1/p}(\sqrt{\Lambda_0}+T)^{1+2/p}T^{1+1/p}\epsilon^{-1/p}\right)$, nearly matching the complexity of the method based on HHKL decomposition for large values of $p$.

Notice that the gate complexity of simulating the Hubbard-Holstein model has an almost quadratic dependence on the time $T$, in stark contrast with the almost linear dependence that applies when all local terms in the Hamiltonian have bounded norm \cite{HaahHastingsKothariLow2021quantum,ChildsSu2019nearly}. 
In fact, there exist unbounded Hamiltonians which are impossible to simulate with an almost linear scaling in $T$.
In Section \ref{sec:lower_bound} in the Appendix we construct a class of Hamiltonians acting on one bosonic mode and $N$ qubits for which simulating the evolution of qubits for time $T$ requires $\wt{\Omega}(N T^2)$ gates in general, for $\sqrt{N}\leq T\leq 2^{N/2}$.

\section{The eigenstate tail bound}
\label{sec:m_eigenstate_tail_bound}

Aside from studies of dynamics, classical or quantum computers may be used to study the static properties of ground states or low-energy states in quantum systems involving bosons or gauge fields. As in simulations of dynamics, we must truncate the local quantum numbers to ensure that local Hilbert spaces at sites or links are finite dimensional. How well can we approximate energy eigenstates of the ideal untruncated Hamiltonian within the truncated Hilbert space?

Suppose that for each site or link, denoted by $\nu$, the Hamiltonian $H$ can be expressed as in \eqref{eq:Ham-form-WR} and satisfies \eqref{eq:m_conditions_ham_general}. Consider a nondegenerate eigenvalue $\varepsilon$ of $H$, with corresponding eigenstate $\ket{\Psi}$, where $\varepsilon$ is separated from the rest of the spectrum of $H$ by a gap $\delta$, and suppose that the expectation value of the absolute value of the local quantum number in the state $\ket{\Psi}$ is finite,
$\bar{\lambda}=\sum_{\lambda}\left|\lambda\right|\braket{\Psi|\Pi^{(\nu)}_{\lambda}|\Psi} < \infty$. 
Our goal is to find a truncation threshold $\Lambda$ such that $\|\overline{\Pi}^{(\nu)}_{[-\Lambda,\Lambda]}\ket{\Psi}\|\leq \epsilon$. We show that this truncation threshold can be chosen to scale with $\epsilon$, $\delta$, and $\bar \lambda$ according to
\begin{equation}
\label{eq:m_tail_bound}
    \Lambda^{1-r} = (2\bar{\lambda})^{1-r} + \Or(\chi \delta^{-1}\log^2(\epsilon^{-1})+\log(\epsilon^{-1})),
\end{equation}
where $\chi$ is a constant independent of system size.

A detailed proof of \eqref{eq:m_tail_bound} can be found in Section \ref{sec:tail} in the Appendix. The polylogarithmic dependence of the truncation threshold $\Lambda$ on the truncation error $\epsilon$ arises because the distribution of local quantum numbers in the eigenstate $|\Psi\rangle$ decays exponentially. 
This contrasts with the polynomial decay one can derive using Markov's or Chebyshev's inequality.

The main tool used in our proof is an approximate eigenstate projection operator \cite{Hastings2007area}
\begin{equation}
     \wt{P}_{\varepsilon} = \frac{\sigma}{\sqrt{2\pi}} \int_{-T}^T \dd t e^{-\frac{1}{2}\sigma^2 t^2} e^{-i\varepsilon t}e^{itH}.
\end{equation}
When $\sigma\ll \delta$ and $T\gg\sigma^{-1}$, this operator is close to the eigenstate projection operator $P_{\varepsilon}=\ket{\Psi}\bra{\Psi}$. We derive \eqref{eq:m_tail_bound} by applying the approximate projector $\wt{P}_{\varepsilon}$ to a suitable initial state and using properties of the time evolution operator $e^{-iHt}$, in particular the truncation threshold result \eqref{eq:m_truncation_threshold_scaling_state}. We may choose the initial state to be $\Pi^{(\nu)}_{[-2\bar{\lambda},2\bar{\lambda}]}\ket{\Psi}$, which by Markov's inequality has an $\mathcal{O}(1)$ overlap with $\ket{\Psi}$. By observing that $e^{-iHt}\Pi^{(\nu)}_{[-2\bar{\lambda},2\bar{\lambda}]}\ket{\Psi}$ can be well approximated by a state with an appropriately chosen truncation threshold, we obtain \eqref{eq:m_tail_bound}.

Note that \eqref{eq:m_tail_bound} does not apply to eigenstates that are degenerate due to symmetries of the Hamiltonian $H$. Nor is it particularly useful when applied to generic highly excited eigenstates, for which the gap $\delta$ may be exponentially small in the system size.

\section{Discussion}

We have studied the task of simulating Hamiltonian dynamics for quantum systems on a lattice, where local Hilbert spaces at lattice sites or links are infinite dimensional. In these systems, local quantum numbers on sites or links can be arbitrarily large in principle. For a large class of such models, we derived upper bounds on how rapidly these local quantum numbers can increase with time, hence showing that time evolved states can be well approximated in a truncated Hilbert space in which each local quantum number is no larger than a truncation threshold $\Lambda$. In particular, we showed that for a fixed evolution time $T$, a precision $\epsilon$ can be achieved by choosing $\Lambda$ scaling polylogarithmically with $\epsilon^{-1}$, as indicated in \eqref{eq:m_truncation_threshold_scaling_state} and \eqref{eq:m_truncation_threshold_scaling_hamiltonian}. Leveraging this finding, we 
established a threshold for truncating the Hamiltonian with a provable accuracy guarantee and
developed algorithms for quantum simulation of LGTs with gate complexity $\wt{\Or}(NT\polylog(\Lambda_0 \epsilon^{-1}))$, where $N$ is the system size, assuming that the initial state can be well approximated with truncation threshold $\Lambda_0$. For a bosonic system like the Hubbard-Holstein model, our algorithm has gate complexity $\wt{\Or}(NT(\sqrt{\Lambda_0}+T)\polylog(\epsilon^{-1}))$. By applying our bounds on the growth of local quantum numbers, we also showed that spectrally isolated energy eigenstates can be approximated with precision $\epsilon$ using a truncation threshold polylogartihmic in $\epsilon^{-1}$, as indicated in \eqref{eq:m_tail_bound}.

Although formally the local Hilbert spaces are infinite dimensional in the models we considered, our results show that at least for some purposes these models can be accurately approximated by models with finite-dimensional local Hilbert spaces of relatively modest size. Many fundamental results have been derived for quantum spin systems with finite-dimensional spins on each lattice site, such as the exponential clustering theorem \cite{Hastings2004locality,HastingsKoma2006spectral,NachtergaeleSims2006lieb}, 

the area law in one dimension \cite{Hastings2007area,AradKitaevLandauVazirani2013area}, and the connection between local and global eigenstates \cite{AradKuwaharaLandau2016connecting}. Perhaps the tools we have developed can be exploited to extend some of these results to systems with infinite-dimensional local degrees of freedom. 

There are certain models of physical interest that do not immediately fit in our framework. These include models that involve a quadratic coupling between bosonic modes, such as the Bose-Hubbard model (\changed{$r=1$} in \eqref{eq:m_conditions_ham_general_2}) and the discretized $\phi^4$ theory (\changed{$r=2$}); our analysis handles the case where $r<1$ in \eqref{eq:m_conditions_ham_general}. Our framework also does not apply to boson-fermion coupling models where anharmonicity is involved that leads to $r=2$. Nevertheless, the method we have developed already provides a unified treatment for a wide range of bosonic systems and lattice gauge theories, and we hope future work could study other physical systems that have not been considered in our work.

For $\phi^4$ theory on a lattice, truncation thresholds were previously analyzed using energy conservation and Chebyshev's inequality \cite{JordanLeePreskill2012quantum}, a method that can be extended to other models as well. Our results apply only to models that satisfy \eqref{eq:Ham-form-WR} and \eqref{eq:m_conditions_ham_general}. For models in this class, we compare our methods with energy-based methods in Section \ref{sec:comparison_with_the_energy_based_truncation_threshold} in the Appendix, finding that our methods yield a more favorable truncation threshold in the limit of short time, high precision, or large system size.

The energy-based truncation threshold in \cite{JordanLeePreskill2012quantum} has the advantage of being time independent, and it can also be applied to models that do not satisfy \eqref{eq:Ham-form-WR} and \eqref{eq:m_conditions_ham_general}, such as $\phi^4$ theory and other models involving bosons with anharmonic couplings. However, it has the disadvantage that the truncation threshold scales polynomially rather than polylogarithmically with $\epsilon^{-1}$. Under suitable conditions, can the truncation threshold scale as $\polylog(\epsilon^{-1})$ in a broader class of models than those satisfying \eqref{eq:Ham-form-WR} and \eqref{eq:m_conditions_ham_general}, and are there models in which $\polylog(\epsilon^{-1})$ scaling can be achieved by a time-independent truncation threshold? 
Moreover, the energy-based truncation threshold provides an answer to Question 1 in the context of truncating a quantum state, but it has not been shown, at the same level of rigor, that the energy-based method also provides an answer to Question 2 in the context of truncating the Hamiltonian. 
The latter is however necessary if we want to rigorously apply the energy-based truncation threshold to Hamiltonian simulation. These are open questions to be addressed in future work.

Another question that has yet to be answered is how to control the error for observables in boson and gauge theory simulations. For bounded observables, once we can control the error in the quantum state, we can automatically control the error of observables. However for unbounded observables, such as the boson occupation number and the electric field value, this simple approach is not suitable. For local observables in lattice models with a finite speed limit for information propagation, one intuitively expects the observable error to also have a bound that respects this locality. We hope our approach can be extended to address questions of this kind, ultimately leading to a theoretical foundation for studying quantum systems with infinite degrees of freedom.

\section*{Acknowledgement}
We thank Kunal Sharma, Mark Wilde, Minh Cong Tran, Junyu Liu, Chi-Fang (Anthony) Chen, Ryan Babbush, Joonho Lee, Di Luo, Nathan Wiebe, Dominic Berry, and Lin Lin for helpful discussions.
YT was partly supported by the NSF Quantum Leap Challenge Institute (QLCI) program through Grant No. OMA-2016245, and by the Department of Energy under Grant No. FWP-NQISCCAWL.
JP was partly supported by the U.S. Department of Energy Office of Advanced Scientific Computing Research (DE-NA0003525, DE-SC0020290) and Office of High Energy Physics (DE-ACO2-07CH11359, DE-SC0018407), the Simons Foundation It from Qubit Collaboration, the Air Force Office of Scientific Research (FA9550-19-1-0360), and the National Science Foundation (PHY-1733907). 
YS was partly supported by the National Science Foundation RAISE-TAQS 1839204. 
VVA acknowledges support from the NSF Quantum Leap Challenge Institute (QLCI) program through Grant No. OMA-2120757.
The Institute for Quantum Information and Matter is an NSF Physics Frontiers Center. 
Contributions to this work by NIST, an agency of the US government, are not subject to US copyright. Any mention of commercial products does not indicate endorsement by NIST.

\bibliographystyle{abbrvnat}
\bibliography{ref}

\onecolumn\newpage
\appendix

\section{Motivating examples}
\label{sec:motivating_examples}

We begin by introducing example quantum systems that we will analyze and simulate. These include a general model for boson-fermion coupling, U(1) lattice gauge theory, and SU(2) lattice gauge theory. We refer the reader to Section~\ref{sec:other_examples} for other common models that can be analyzed within our framework.

\vspace{1em}
\textit{Boson-fermion coupling.}
We assume that there are $N_f$ fermionic modes and $N_b$ bosonic modes in the system. 
We label the fermionic modes by $i,j$ and bosonic modes by $\alpha$. 
The $c_i$ and $b_{\alpha}$ denote the fermionic and bosonic annihilation operators respectively.
The Hamiltonian takes the form
\begin{equation}
    \label{eq:ham_fermion_boson_general}
    \begin{aligned}
    H &= H_f+H_{fb}+H_b, \\
    H_{f} &= \sum_{ij}t_{ij}c_{i}^{\dagger}c_j + \sum_{ij}V_{ijkl}c_i^{\dagger}c_j^{\dagger}c_k c_l, \\
    H_{fb} &= \sum_{\alpha ij} g_{ij}^{(\alpha)} c_i^{\dagger}c_j X_{\alpha}
    + \sum_{\alpha ij} h_{ij}^{(\alpha)} c_i^{\dagger}c_j P_{\alpha}, \\
    H_b &= \frac{1}{2}\sum_{\alpha} \omega_{\alpha} X_{\alpha}^2 
    + \frac{1}{2}\sum_{\alpha} \omega_{\alpha} P_{\alpha}^2,
    \end{aligned}
\end{equation}
where $X_{\alpha}=(b_{\alpha}+b_{\alpha}^{\dagger})/\sqrt{2}$ is the position operator corresponding to the bosonic mode $\alpha$, and $P_{\alpha}=i(b^{\dagger}_{\alpha}-b_{\alpha})/\sqrt{2}$ is the momentum operator. $t=(t_{ij})$, $g^{(\alpha)}=(g^{(\alpha)}_{ij})$, and $h^{(\alpha)}=(h_{ij}^{(\alpha)})$ are all Hermitian matrices, and $V=(V_{ijkl})$ is the electron repulsion integral tensor satisfying the usual symmetry. We remark that the commonly seen Hubbard-Holstein model \cite{holstein1959studies} and the Fr{\"o}hlich model \cite{Frohlich1954electrons} both take the above form. 

\vspace{1em}
\textit{U(1) lattice gauge theory.} For notation simplicity we consider only the $(2+1)$-dimensional theory. Extension to the $(3+1)$-dimensional case is straightforward. The system consists of a square lattice of $N$ sites. We denote each site by $x$, and the lattice vector in the horizontal and vertical directions are noted $n_1$ and $n_2$ respectively. We use $(x,n_i)$ to represent the link between sites $x$ and $x+n_i$, $i=1,2$. The links are sometimes called \textit{gauge links}.

On each site $x$ we have a fermionic mode whose annihilation operator is denoted by $\phi_x$. Each link consists of a planar rotor, whose configuration space of states $\ket{\theta}$, with $\theta\in[0,2\pi]$ being an angle, is equivalent to that of a particle on a ring. An orthonormal basis of the Hilbert space can be chosen to be
\begin{equation}
    \ket{k} = \frac{1}{\sqrt{2\pi}}\int_0^{2\pi} e^{ik\theta}\ket{\theta}\dd \theta,
\end{equation}
for $k\in\ZZ$.

In Hilbert space of link $(x,n_i)$ we define operators $E_{x,n_i}$ and $U_{x,n_i}$ through
\begin{equation}
\label{eq:operators_U(1)}
    E_{x,n_i}\ket{k} = k\ket{k},\quad U_{x,n_i}\ket{k} = \ket{k-1}.
\end{equation}
Then the Hamiltonian of the system is
\begin{equation}
\label{eq:LGT_ham}
\begin{aligned}
    H &= H_M +H_{GM} + H_E + H_B, \\
    H_M &= g_M\sum_{x}(-1)^x \phi_x^{\dagger}\phi_x, \\
    H_{GM} &= g_{GM}\sum_{x,i} (\phi_x^{\dagger}U_{x,n_i}\phi_{x+n_i}+\mathrm{h.c.}), \\
    H_E &= g_E\sum_{x,i} E_{x,n_i}^2, \\
    H_B &= g_B\sum_P(\Tr U_P+\mathrm{h.c.}),
\end{aligned}
\end{equation}
where $\sum_P$ denotes a summation over all plaquettes $P$. For $P$ whose lower-left site is $x$, $U_P$ is defined as
\begin{equation}
\label{eq:wilson_loop}
    U_P = U_{x,n_1}U_{x+n_1,n_2}U_{x+n_2,n_1}^{\dagger}U_{x,n_2}^{\dagger}.
\end{equation}
The trace $\Tr$ in \eqref{eq:LGT_ham} is not needed here but will be required in the setting of SU(2) lattice gauge theory.
The four terms $H_M, H_{GM}, H_E, H_B$ describe the fermionic mass (using staggered fermions \cite{kogut1975hamiltonian}), the gauge-matter interaction, the electric energy, and the magnetic energy respectively.

\vspace{1em}
\textit{SU(2) lattice gauge theory.} The setup of the SU(2) lattice gauge theory is very similar to the U(1) case. Here for simplicity we only consider the theory using the fundamental representation of SU(2). Compared to the U(1) theory, each site $x$ now contains two fermionic modes, whose annihilation operators are denoted by $\phi_{x}^l$, $l=1,2$. We write $\phi_x = (\phi_{x}^1,\phi_{x}^2)^{\top}$. 
Each link consists of a rigid rotator whose configuration is described by an element of the group SU(2) \cite{kogut1975hamiltonian}. An orthonormal basis of the link Hilbert space consists of the quantum states $\ket{jmm'}$, where $j, m, m'$ are simultaneously either integers or half-integers with $-j\leq m,m'\leq j$. Here $j$ is the rotator's total angular momentum, and $m$, $m'$ denote the components of angular momentum along the $z$-axis in the body-fixed and space-fixed coordinate systems.

The Hamiltonian takes the form \eqref{eq:LGT_ham}, and is invariant under $\mathrm{SU}(2)$ transformations acting either from the left or from the right, which may be interpreted as rotations of the rigid rotator with respect to space-fixed or body-fixed axes respectively. The operators $E^2_{x,n_i}$ and $U_{x,n_i}$ are different from the U(1) case.
The operator $E^2_{x,n_i}$ is defined through
\begin{equation}
\label{eq:operators_SU(2)}
    E^2_{x,n_i}\ket{jmm'} = j(j+1)\ket{jmm'}.
\end{equation}
Because $\phi_x$ has two components, where each component is a fermionic mode,  $U_{x,n_i}$ is a $2\times 2$ matrix, where each of the 4 matrix entries is an operator acting on the link Hilbert space
\begin{equation}
    U_{x,n_i} = \begin{pmatrix}
    U^{11}_{x,n_i} & U^{12}_{x,n_i}\\
    U^{21}_{x,n_i} & U^{22}_{x,n_i}
    \end{pmatrix}.
\end{equation}
An important property that we will use later is
\begin{equation}
    \label{eq:U_SU(2)_properties}
    \braket{j_1 m_1 m_1'|U^{ll'}_{x,n_i}|j_2 m_2 m_2'}=0,\ \text{if }|j_1-j_2|>1/2,\quad \|U^{ll'}_{x,n_i}\|\leq 1,
\end{equation}
which follows from rules for the addition of angular momentum, given that $U_{x,n_i}$ transforms as the $j=1/2$ representation of $\mathrm{SU}(2)$. Here $\|O\|$ denotes the spectral norm of an operator $O$.
We also note that relative to the basis $\{\ket{jmm'}\}$, $U^{ll'}_{x,n_i}$'s are sparse matrices because
\begin{equation}
    \label{eq:U_SU(2)_sparse}
    \braket{j_1 m_1 m_1'|U^{ll'}_{x,n_i}|j_2 m_2 m_2'}=0,\ \text{if } m_1-m_2\neq l-3/2\ \text{or } m_1'-m_2'\neq l'-3/2,
\end{equation}
due to the conservation of angular momentum along the $z$-axis in the body-fixed and space-fixed coordinate systems. Here $l-3/2$ and $l'-3/2$ are the change of angular momentum as a result of applying $U^{ll'}_{x,n_i}$. Eqs.\ \eqref{eq:U_SU(2)_properties} and \eqref{eq:U_SU(2)_sparse} imply that the matrix representing $U^{ll'}_{x,n_i}$ has at most three non-zero elements in each row and column.

\section{The common structure}
\label{sec:common_structures}

Here we identify a common structure in all the examples introduced in Section~\ref{sec:motivating_examples}. We first decompose the entire Hilbert space $\mathcal{H}$ into a direct sum of subspaces $\mathcal{V}_{\lambda}$ with quantum numbers $\lambda\in\ZZ$. 
The projection operator onto each subspace $\mathcal{V}_{\lambda}$ is denoted by $\Pi_{\lambda}$. Then $\sum_{\lambda\in\ZZ}\Pi_{\lambda}=I$. 
We consider a class of Hamiltonians of the form
\begin{equation}
\label{eq:ham_general}
    H = H_W + H_R,
\end{equation}
where
\begin{equation}
\label{eq:conditions_ham_general}
    \Pi_{\lambda}H_W \Pi_{\lambda'}=0,\ \text{if }|\lambda-\lambda'|>1,\quad \|H_W\Pi_{[-\Lambda,\Lambda]}\|\leq \chi(\Lambda+1)^{r},\quad [H_R,\Pi_{\lambda}]=0,
\end{equation}
for some $\chi>0$, $0\leq r<1$. Here $\Pi_{[-\Lambda,\Lambda]}=\sum_{|\lambda|\leq\Lambda}\Pi_{\lambda}$. In such a Hamiltonian, $H_W$ changes the quantum number $\lambda$ in the time evolution while $H_R$ preserves it. \eqref{eq:conditions_ham_general} ensures that $\lambda$ is not changed too quickly. 
The first part of \eqref{eq:conditions_ham_general} ensures that the local quantum number is changed by at most $\pm 1$ each time the Hamiltonian is applied, and the second part ensures that the rate of the change is sublinear in the current local quantum number. The meaning of these conditions will be made clearer when we discuss the leakage bound in Section~\ref{sec:leakage_time_evo}. 

\begin{table}[]
    \centering
    \begin{tabular}{|c|c|c|}
        Model         & $\lambda$                 & $r$ \\
        \hline
        Boson-fermion & Bosonic particle number   & $1/2$ \\
        U(1) LGT      & Electric field value      & $0$   \\
        SU(2) LGT     & Total angular momentum    & $0$ 
    \end{tabular}
    \caption{Local quantum number $\lambda$ and the exponent $r$ for models discussed in Section~\ref{sec:common_structures}.}
    \label{tab:local_quantum_numbers}
\end{table}

We check that all the models introduced in Section~\ref{sec:motivating_examples} satisfy \eqref{eq:ham_general} and \eqref{eq:conditions_ham_general}. 
For the boson-fermion coupling Hamiltonian defined in \eqref{eq:ham_fermion_boson_general}, fixing a bosonic mode $\alpha_0$, we can decompose the Hilbert space according to the number of particles in the bosonic mode $\alpha_0$, which we denote by $m$. This means we let $\lambda=m$ and
\begin{equation}
\label{eq:boson_fermion_projection_operator}
    \Pi_{m}=\ketbra{m}{m}_{\alpha_0}\otimes I,
\end{equation} 
where $\ket{m}_{\alpha_0}$ means the $\ket{m}_{\alpha_0}$-particle state of the mode, and $I$ is the identity operator acting on the rest of the system. We set $\Pi_{\lambda}=0$ for all $\lambda<0$. We define $H_W$ to be the sum of terms in \eqref{eq:conditions_ham_general} that change the particle number in mode $\alpha_0$:
\begin{equation}
    H_W = \sum_{ij} \left(g^{\alpha_0}_{ij}c_i^{\dagger}c_j X_{\alpha_0} + h^{\alpha_0}_{ij}c_i^{\dagger}c_j P_{\alpha_0}\right),
\end{equation}
whereas the rest of the terms in $H$ are collected into $H_R$. Because of the fact that
\begin{equation}
    \label{eq:bosonic_XP_norm}
    \|X_{\alpha_0}\Pi_{[0,M]}\|, \|P_{\alpha_0}\Pi_{[0,M]}\| < \sqrt{2(M+1)},
\end{equation}
where $\Pi_{[0,M]}=\sum_{m=0}^M\Pi_{m}$,
one can see that \eqref{eq:conditions_ham_general} is satisfied if we choose $r=1/2$ and 
\begin{equation}
\label{eq:chi_boson_fermion}
    \chi=\sqrt{2}(\max_{\alpha}\Tr|g^{(\alpha)}| + \max_{\alpha}\Tr|h^{(\alpha)}|),
\end{equation}
where $|A|=\sqrt{A^{\dagger}A}$ for any matrix $A$. 
\changed{Note that $r$ is not determined by the highest order terms in terms of the position and momentum operators, but rather the highest order terms that change the local quantum number. For example, $X_{\alpha_0}^2$ and $P_{\alpha_0}^2$ each on their own will result in $r=1$, but since in the Hamiltonian they appear together, $X_{\alpha_0}^2+P_{\alpha_0}^2$ preserves the local quantum number and therefore does not contribute to how large $r$ is.}
 
$\Tr|A|$ is the trace norm of $A$ and it is used here because $\|\sum_{ij}A_{ij}c_i^{\dagger}c_j\|\leq \Tr|A|$ for any Hermitian matrix $A=(A_{ij})$. This can be proved for any matrix $A$ using the singular value decomposition~\cite{Otte2010Boundedness}.

In the setting of U(1) lattice gauge theory, again we fix a given link indexed by $(x_0,n_0)$ where $n_0\in\{n_1,n_2\}$. Then we decompose the Hilbert space by the electric field value on this link, i.e. we let $\lambda=k$ and define
\begin{equation}
    \label{eq:U(1)_projection_operator}
    \Pi_{k} = \ketbra{k}{k}_{(x_0,n_0)}\otimes I.
\end{equation}
Then $H_W$ should be chosen as
\begin{equation}
\label{eq:walk_ham_LGT}
    H_W = g_{GM}(\phi_{x_0}^{\dagger}U_{x_0,n_0}\phi_{x_0+n_0}+\mathrm{h.c.}) + g_B\sum_{P\ni(x_0,n_0)}(\Tr U_P+\mathrm{h.c.}).
\end{equation}
Because of the fact that $\|H_W\|\leq 4|g_{B}|+2|g_{GM}|$, \eqref{eq:conditions_ham_general} is satisfied if we choose $\chi= 4|g_{B}|+2|g_{GM}|$ and $r=0$.
Here we have $r=0$ because, unlike the bosonic position and momentum operators, $U_{x_0,n_0}$ is a bounded operator.

In the setting of SU(2) lattice gauge theory, again we fix a given link indexed by $(x_0,n_0)$. We decompose the Hilbert space according to the total angular momentum on this link. This is to say, we let $\lambda=2j$ ($j$ takes half-integer value), and
\begin{equation}
    \label{eq:SU(2)_projection_operator}
    \Pi_{2j} = \sum_{-j\leq m,m'\leq j}\ketbra{jmm'}{jmm'}_{(x_0,n_0)}\otimes I.
\end{equation}
Here we require $m,m'$ to be integers when $j$ is an integer and half-integers when $j$ is a half integer.
Then $H_W$ takes the same form as in \eqref{eq:walk_ham_LGT}.
Eq.\ \eqref{eq:U_SU(2)_properties} ensures that \eqref{eq:conditions_ham_general} is satisfied if we choose $\chi= 16|\lambda_{B}|+8|g_{GM}|$ and $r=0$. There is an additional factor of $4$ in $\chi$ compared to the U(1) case because there are now four operators $U_{x_0,n_0}^{ll'}$ contributing to the growth of the quantum number instead of one.

More generally, we define $\Pi_{S}$, where $S$ is a set of integers, as
\begin{equation}
    \Pi_S = \sum_{\lambda\in S} \Pi_{\lambda}.
\end{equation}
In the examples introduced above, we have focused on the quantum numbers on a single fixed bosonic mode or gauge link, and decomposed the Hilbert space accordingly. In fact this procedure can be done for every mode and link. Therefore we sometimes need to designate projection operators for each mode or link. In the boson-fermion coupling situation, we denote by $\Pi^{(\alpha)}_{m}$ the projection operator into the subspace with $m$ particles in the bosonic mode $\alpha$. Similarly we define $\Pi^{(\alpha)}_{S}$ for any integer set $S$. When we need to constrain the particle number on all modes, we define
\begin{equation}
    \Pi^{\mathrm{all}}_{S} = \prod_{\alpha}\Pi^{(\alpha)}_{S}.
\end{equation}
As a general rule, if $\Pi$ is any projection operator, we define $\overline{\Pi}=I-\Pi$.
For lattice gauge theories we adopt similar notations. For example we use $\Pi^{(x,n)}_S$ to denote the projection operator into the subspace with the quantum number taking value in set $S$ on gauge link $(x,n)$. Moreover, we sometimes use $\nu$ to index both the bosonic mode and gauge links when we discuss the two scenarios together. Therefore $\Pi^{(\nu)}_S$ can mean either $\Pi^{(\alpha)}_S$ or $\Pi^{(x,n)}_S$ depending on the context.

\section{Truncating an evolved quantum state}
\label{sec:leakage_time_evo}

Our first goal is to answer the following question: suppose we start from an initial state with quantum number $\lambda$ between $\pm\Lambda_0$, what is the probability that $|\lambda|$ grows beyond some given $\Lambda$ as the state evolves for time $t$? To be more concrete, we want to bound
\begin{equation}
\label{eq:leakage_def}
    \left\|\overline{\Pi}_{[-\Lambda,\Lambda]}e^{-itH}\Pi_{[-\Lambda_0,\Lambda_0]}\right\|,
\end{equation}
when $H$ has the structure \eqref{eq:ham_general}. We call this quantity the \emph{leakage}. As discussed in the main article, this upper bounds the error of truncating the quantum state at time $t$ when the initial state has a quantum number between $-\Lambda_0$ and $\Lambda_0$.

\subsection{The short-time leakage bound}
\label{sec:short_time_leakage_bound}

We first establish a bound on the leakage defined in \eqref{eq:leakage_def} for a short time $t$:
\begin{lem}[Short-time leakage bound]
\label{lem:short_time_bound}
Given Hamiltonian $H=H_W+H_R$ satisfying \eqref{eq:conditions_ham_general} with parameters $\chi$ and $r$, we have
\[
\left\|\overline{\Pi}_{(-\Lambda_0-\Delta,\Lambda_0+\Delta)}e^{-itH}\Pi_{[-\Lambda_0,\Lambda_0]}\right\|\leq \frac{1}{2^{\Delta-1}(\Delta!)^{1-r}}
\]
for any $0\leq \changed{|t|}\leq 1/(2\chi(\Lambda_0+1)^{r})$ and integers $\Lambda_0,\Delta\geq 0$, where $\overline{\Pi}_{(-\Lambda_0-\Delta,\Lambda_0+\Delta)}=I-\Pi_{(-\Lambda_0-\Delta,\Lambda_0+\Delta)}$.
\end{lem}

\begin{proof}
This proof is based on rewriting the time evolution using the interaction picture, and truncating the Dyson series of the new time evolution. 
\changed{Below we only consider $t>0$. The proof can readily be extended to $t<0$ because we only need to replace $H$ by $-H$, and the structure in \eqref{eq:conditions_ham_general} is preserved by this transformation.}
First we define
\begin{equation}
    \label{eq:H0_int_pic}
    H_W(t) = e^{itH_R}H_W e^{-itH_R}.
\end{equation}
Then writing the time evolution $e^{-itH}$ in the interaction picture, we have
\begin{equation}
    e^{-itH} = e^{-itH_R}\mathcal{T}e^{-i\int_0^t H_W(s)\dd s},
\end{equation}
where $\mathcal{T}$ is the time-ordering operator.
Since $e^{-itH_R}$ commutes with $\overline{\Pi}_{(-\Lambda_0-\Delta,\Lambda_0+\Delta)}$, we only need to bound
\begin{equation}
\label{eq:new_leakage}
    \left\|\overline{\Pi}_{(-\Lambda_0-\Delta,\Lambda_0+\Delta)}\mathcal{T}e^{-i\int_0^t H_W(s)\dd s}\Pi_{[-\Lambda_0,\Lambda_0]}\right\|.
\end{equation}

To this end, we consider the partial sum of the Dyson series of $\mathcal{T}e^{-i\int_0^t H_W(s)\dd s}$:
\begin{equation}
    \label{eq:dyson_partial_sum}
    S(\Delta) = \sum_{k=0}^{\Delta-1} (-i)^k \int_0^{t} \dd t_k \cdots \int_0^{t_3} \dd t_2  \int_0^{t_2} \dd t_1  H_W(t_k)\cdots H_W(t_2)H_W(t_1).
\end{equation}
In order to bound the error from replacing the exact time evolution with this truncated Dyson series, we need to estimate the norms of terms of the form
\begin{equation}
    H_W(t_k)\cdots H_W(t_2)H_W(t_1)\Pi_{[-\Lambda,\Lambda]}.
\end{equation}
Noting that $H_W(t)$ can only change $\lambda$ by $\pm 1$, we have
\begin{equation}
\label{eq:H0_particle_number_change}
    \Pi_{[-\Lambda-1,\Lambda+1]}H_W(t)\Pi_{[-\Lambda,\Lambda]} = H_W(t)\Pi_{[-\Lambda,\Lambda]},
\end{equation}
which implies by \eqref{eq:conditions_ham_general} that
\begin{equation}
\label{eq:H0_norm_bound}
    \|H_W(t)\Pi_{[-\Lambda,\Lambda]}\| = \|H_W\Pi_{[-\Lambda,\Lambda]}\| \leq  \chi(\Lambda+1)^r.
\end{equation}
By repeatedly applying \eqref{eq:H0_particle_number_change}, we have
\begin{equation}
    H_W(t_k)\cdots H_W(t_2)H_W(t_1)\Pi_{[-\Lambda_0,\Lambda_0]} = \prod_{j=1}^{k} \left(\Pi_{[-\Lambda_0-j,\Lambda_0+j]} H_W(t_j) \Pi_{[-\Lambda_0-j+1,\Lambda_0+j-1]}\right).
\end{equation}
Then applying \eqref{eq:H0_norm_bound},
\begin{equation}
\begin{aligned}
    \left\|H_W(t_k)\cdots H_W(t_2)H_W(t_1)\Pi_{[-\Lambda_0,\Lambda_0]}\right\| &\leq \prod_{j=1}^{k} \left\|\Pi_{[-\Lambda_0-j,\Lambda_0+j]} H_W(t_j) \Pi_{[-\Lambda_0-j+1,\Lambda_0+j-1]}\right\| \\ 
    &\leq \chi^k\left(\frac{(\Lambda_0+k)!}{\Lambda_0!}\right)^{r}.
\end{aligned}
\end{equation}
From the above inequality, we have
\begin{equation}
    \label{eq:dyson_truncation_err}
    \begin{aligned}
    \left\|\left(\mathcal{T}e^{-i\int_0^t H_W(s)\dd s}-S(\Delta)\right)\Pi_{[-\Lambda_0,\Lambda_0]}\right\| &\leq \sum_{k=\Delta}^{\infty} \frac{(\chi t)^k}{k!}\left(\frac{(\Lambda_0+k)!}{\Lambda_0!}\right)^r \\
    &\leq \sum_{k=\Delta}^{\infty} \frac{(\chi t(\Lambda_0+1)^r)^k}{(k!)^{1-r}} \\
    &\leq \frac{1}{(\Delta!)^{1-r}} \sum_{k=\Delta}^{\infty} \frac{1}{2^k} \\
    &=\frac{1}{2^{\Delta-1}(\Delta!)^{1-r}},
    \end{aligned}
\end{equation}
where in the second inequality we have used the fact that 
\begin{equation}
    \frac{(\Lambda_0+k)!}{\Lambda_0 ! k!}= \frac{\Lambda_0+1}{1}\frac{\Lambda_0+2}{2}\cdots \frac{\Lambda_0+k}{k}\leq (\Lambda_0+1)^k,
\end{equation}
and in the third inequality we have used $t\leq 1/(2\chi (\Lambda_0+1)^r)$.

Note that
\begin{equation}
    \overline{\Pi}_{(-\Lambda_0-\Delta,\Lambda_0+\Delta)}S(\Delta)\Pi_{[-\Lambda_0,\Lambda_0]} = 0
\end{equation}
because of \eqref{eq:H0_particle_number_change}. Therefore,
\begin{equation}
\begin{aligned}
    \left\|\overline{\Pi}_{(-\Lambda_0-\Delta,\Lambda_0+\Delta)}\mathcal{T}e^{-i\int_0^t H_W(s)\dd s}\Pi_{[-\Lambda_0,\Lambda_0]}\right\| &= \left\|\overline{\Pi}_{(-\Lambda_0-\Delta,\Lambda_0+\Delta)}\left(\mathcal{T}e^{-i\int_0^t H_W(s)\dd s}-S(\Delta)\right)\Pi_{[-\Lambda_0,\Lambda_0]}\right\| \\
    &\leq \frac{1}{2^{\Delta-1}(\Delta!)^{1-r}}.
\end{aligned}
\end{equation}
This finishes the proof.
\end{proof}

\subsection{The long-time leakage bound}
\label{sec:long_time_leakage_bound}

The long time bound is based on the following decomposition. 
\begin{lem}
\label{lem:decompose_time_evolution}
Let $P_j$, $\overline{P}_j$ be projection operators such that $P_j + \overline{P}_j = I$, $j=0,1,\ldots,J$. Then for any $0=t_0<t_1<\cdots <t_J=t$,
\begin{equation}
\label{eq:decompose_time_evolution_eq}
    \overline{P}_J e^{-itH} P_0 = \overline{P}_J \sum_{j=1}^{J}  e^{-i(t-t_j)H}\overline{P}_j \prod_{j'=0}^{j-1} \left(e^{-i(t_{j'+1}-t_{j'})H}P_{j'}\right),
\end{equation}
which implies
\begin{equation}
    \label{eq:error_accumulation}
    \|\overline{P}_J e^{-itH} P_0\| \leq \sum_{j=1}^J \|\overline{P}_j e^{-i(t_j-t_{j-1})H}P_{j-1}\|.
\end{equation}
\end{lem}
At a high level, this lemma suggests that the total leakage (quantified by the spectral norm) is upper bounded by the sum of the leakage in each time step. 
This lemma can be easily proved by induction on $J$. A more intuitive way of proving it is to write
\begin{equation}
    e^{-itH} = \prod_{j'=0}^{J-1}\left((P_{j'+1}+\overline{P}_{j'+1})e^{-i(t_{j'+1}-t_{j'})H}\right),
\end{equation}
and expand the product on the right-hand side into a sum of terms, each of which is a string of $P_{j'}$ and $\overline{P}_{j'}$ interspersed with $e^{-i(t_{j'+1}-t_{j'})H}$. We then recombine these terms according to where the first $\overline{P}_{j'}$ appears, or if it does not show up at all. The sum of all terms for which the first $\overline{P}_{j'}$ appears in the $j$-th place is 
\begin{equation}
    e^{-i(t-t_j)H}\overline{P}_j \prod_{j'=1}^{j-1} \left(e^{-i(t_{j'+1}-t_{j'})H}P_{j'}\right)e^{-i(t_{1}-t_{0})H}.
\end{equation}
We then sum over $j$, and multiply $\overline{P}_J$ and $P_0$ to the left and right respectively, to get the right-hand side of \eqref{eq:decompose_time_evolution_eq}.

\begin{figure}
    \centering
    \includegraphics[width=0.4\textwidth]{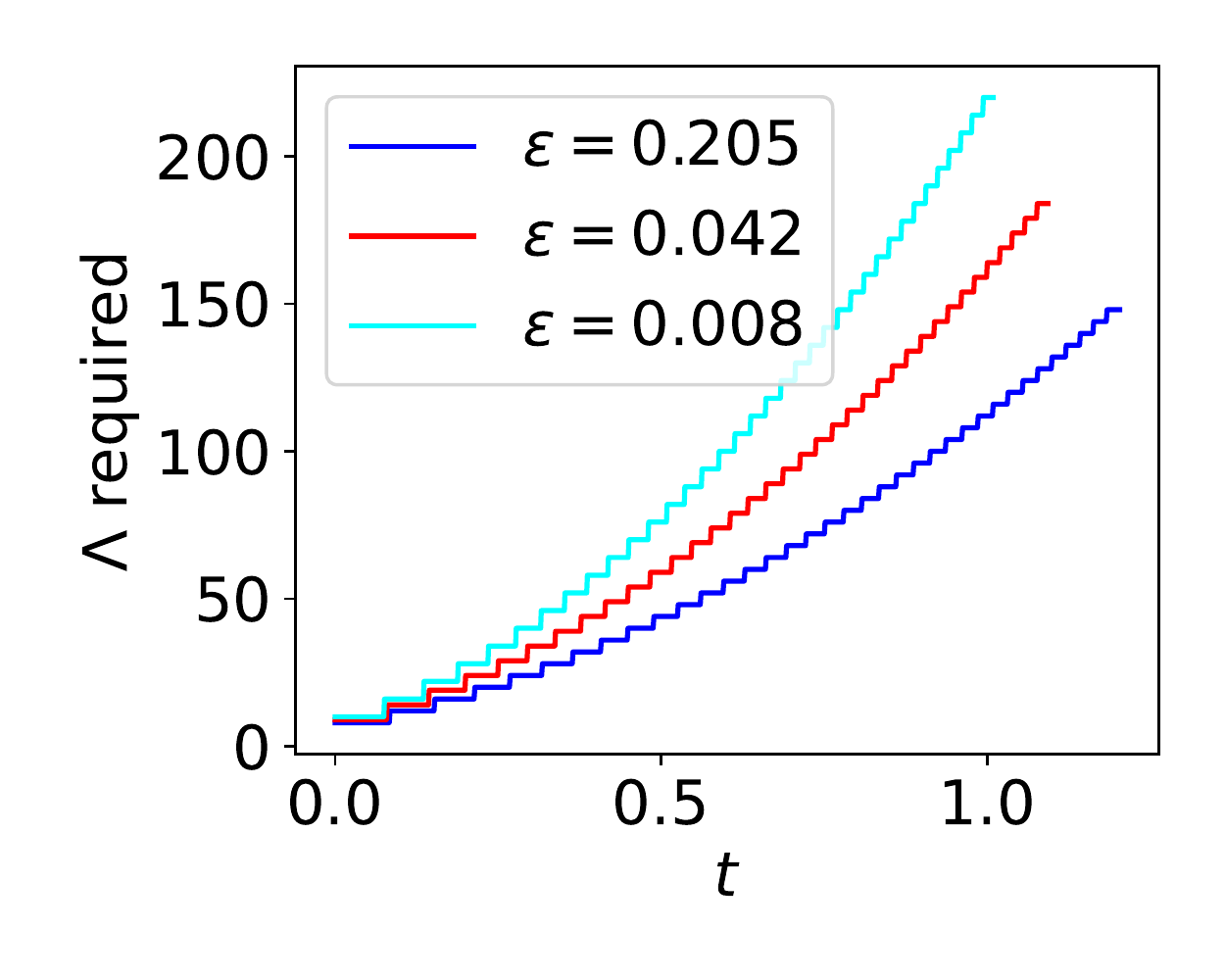}
    \includegraphics[width=0.4\textwidth]{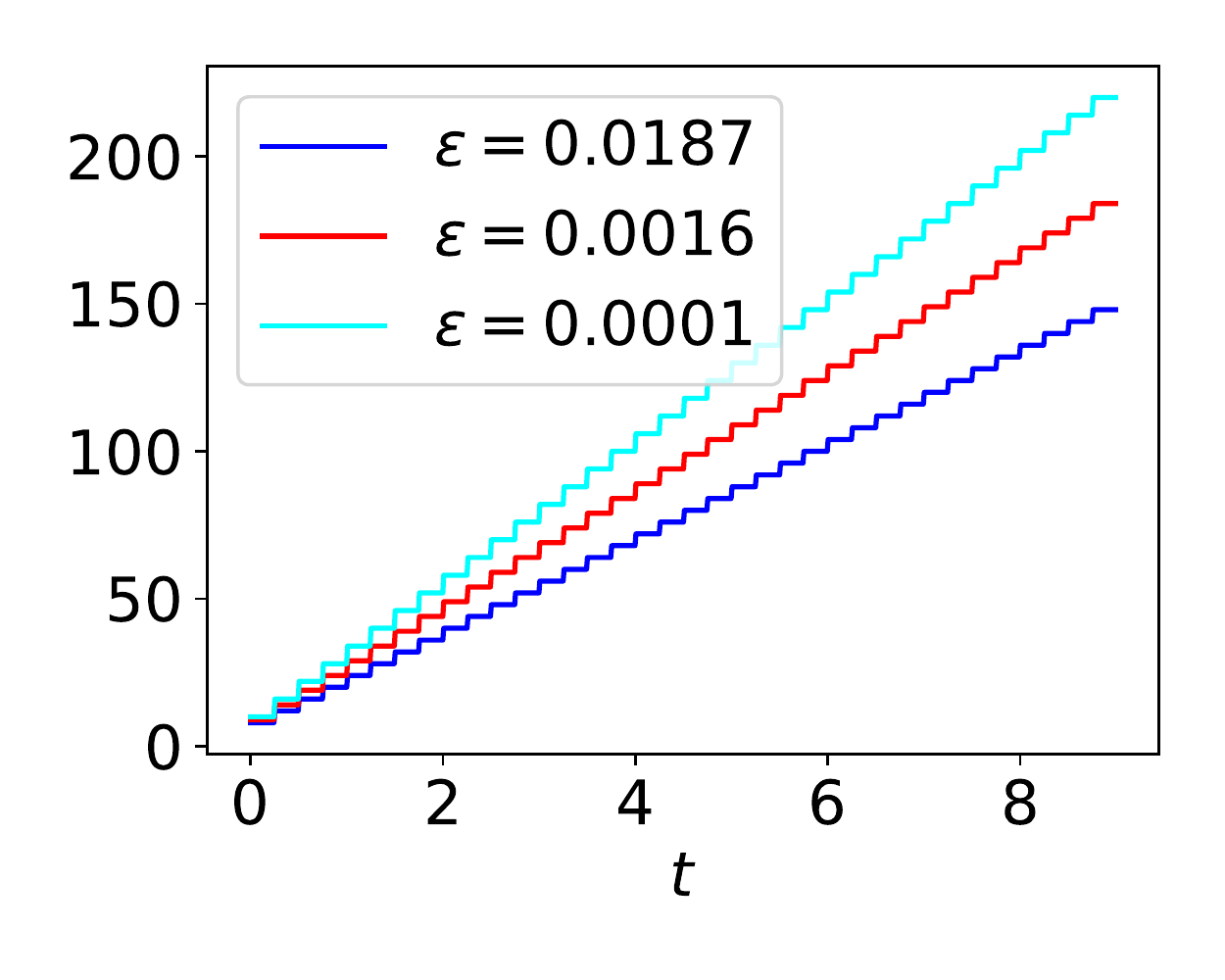}
    \caption{The truncation threshold $\Lambda(t)$ needed to keep the error below $\epsilon$ for boson-fermion coupling (left) and lattice gauge theories (right). Note that $\Lambda(t)$ grows quadratically for the former and linearly for the latter. The truncation threshold is obtained, according to \eqref{eq:choose_M(t)}, in the boson-fermion coupling setting for the Hubbard-Holstein model \eqref{eq:ham_fermi_hubbard} with $g=1$, and in the lattice gauge theory setting for the U(1) lattice gauge theory \eqref{eq:LGT_ham} with $g_B=g_{GM}=1/6$. 
    }
    \label{fig:quantum_number_growth}
\end{figure}

We now state our long-time leakage bound:
\begin{thm}[Long-time leakage bound]
\label{thm:long_time_bound}
Let $H=H_W+H_R$ be a Hamiltonian satisfying \eqref{eq:conditions_ham_general} with parameters $\chi$ and $r$. Defining
\begin{equation}
    \label{eq:choose_M(t)}
    \Lambda(t) = \Lambda_0 + \Big\lceil \frac{1}{\Delta-1}\left((\Lambda_0^{1-r}+2\chi \changed{|t|}(1-r)(\Delta-1))^{\frac{1}{1-r}}-\Lambda_0\right) \Big\rceil(\Delta-1)
\end{equation}
for any \changed{$t\in\mathbb{R}$} and integers $\Lambda_0\geq 0$ and $\Delta>1$, we have
\[
\left\|\overline{\Pi}_{[-\Lambda(t),\Lambda(t)]}e^{-itH}\Pi_{[-\Lambda_0,\Lambda_0]}\right\|\leq \frac{\lceil ((\Lambda_0^{1-r}+2\chi \changed{|t|}(1-r)(\Delta-1))^{\frac{1}{1-r}}-\Lambda_0)/(\Delta-1) \rceil}{2^{\Delta-1}(\Delta!)^{1-r}}.
\]
\end{thm}

\changed{Below we will focus on the case of $t>0$. The $t<0$ case can be dealt with in the same way because of the reason explained in the proof of Lemma \ref{lem:short_time_bound}.}

In Figure~\ref{fig:quantum_number_growth} we plot the truncation threshold $\Lambda(t)$ needed to ensure the leakage $\left\|\overline{\Pi}_{[-\Lambda(t),\Lambda(t)]}e^{-itH}\Pi_{[-\Lambda_0,\Lambda_0]}\right\|\leq \epsilon$, for the boson-fermion coupling setting and the lattice gauge theory setting. We can see that $\Lambda(t)$ grows quadratically with time for the former and linearly for the latter. Moreover, very small leakage can be achieved by only slightly increasing $\Lambda(t)$. This follows from the exponential suppression of leakage that we will describe in Theorem~\ref{thm:long_time_bound2}.

The basic idea of the proof is to partition the time evolution into small segments, and apply the short-time bound in Lemma~\ref{lem:short_time_bound} to each segment. We denote by $T_j$, for $j\geq 0$, the intermediate times where we make the partition. First we define the instantaneous quantum number
\begin{equation}
\label{eq:Lambda_j_def}
    \Lambda_j = \Lambda_0 + j(\Delta-1),
\end{equation}
and then choose $T_j$ to be
\begin{equation}
\label{eq:choose_Tj}
    T_j = T_{j-1} + \frac{1}{2\chi (\Lambda_{j-1}+1)^{r}}.
\end{equation}
From this definition we have
\begin{equation}
    \label{eq:Tj_lower_bound}
    T_j \geq \frac{1}{2\chi(\Delta-1)(1-r)}\left((\Lambda_0+j(\Delta-1))^{1-r}-\Lambda_0^{1-r}\right),
\end{equation}
which can be proved using the inequality
\begin{equation}
    \frac{1}{(\Lambda_{j-1}+1)^{r}} \geq \frac{1}{\Delta-1}\int_{\Lambda_0+j(\Delta-1)}^{\Lambda_0+(j+1)(\Delta-1)}\frac{\dd w}{w^r}.
\end{equation}

To establish the long-time leakage bound for arbitrary time $t$, we first prove the following lemma:
\begin{lem}
\label{lem:bound_at_Tj}
Let $H=H_W+H_R$ be a Hamiltonian satisfying \eqref{eq:conditions_ham_general} with parameters $\chi$ and $r$. Given integers $\Lambda_0\geq 0$ and $\Delta>0$, define $\Lambda_j$ and $T_j$ as in \eqref{eq:Lambda_j_def} and \eqref{eq:choose_Tj}. For any $T_{J-1}< t\leq T_J$ with integer $J\geq 0$,
\begin{equation}
\label{eq:bound_at_Tj}
    \begin{aligned}
    \left\|\overline{\Pi}_{[-\Lambda_J,\Lambda_J]}e^{-itH}\Pi_{[-\Lambda_0,\Lambda_0]}\right\| \leq \frac{J}{2^{\Delta-1}(\Delta!)^{1-r}}.
    \end{aligned}
\end{equation}
\end{lem}

\begin{proof}
We choose $t_j=T_j$ for $j=0,1,\ldots,J-1$, and $t_J=t$.
By Eq.\ \eqref{eq:error_accumulation} in Lemma~\ref{lem:decompose_time_evolution}, we have
\begin{equation}
    \begin{aligned}
    \left\|\overline{\Pi}_{[-\Lambda_J,\Lambda_J]}e^{-itH}\Pi_{[-\Lambda_0,\Lambda_0]}\right\| \leq \sum_{j=1}^J \left\|\overline{\Pi}_{[-\Lambda_j,\Lambda_j]}e^{-i(t_j-t_{j-1})H}\Pi_{[-\Lambda_{j-1},\Lambda_{j-1}]}\right\|.
    \end{aligned}
\end{equation}
Because of \eqref{eq:choose_Tj} each $t_j-t_{j-1}$ is short enough for us to apply our short-time bound Lemma~\ref{lem:short_time_bound}. This completes the proof.
\end{proof}

With this lemma we can prove the theorem by appropriately choosing $J$.
\begin{proof}[Proof of Theorem~\ref{thm:long_time_bound}]
We choose $J$ to be the first integer that makes $T_J\geq t$. By \eqref{eq:Tj_lower_bound}, we have \begin{equation}
    J=\Big\lceil \frac{1}{\Delta-1}\left((\Lambda_0^{1-r}+2\chi t(1-r)(\Delta-1))^{\frac{1}{1-r}}-\Lambda_0\right) \Big\rceil.
\end{equation}
The claimed bound then follows from Lemma~\ref{lem:bound_at_Tj}.
\end{proof}

Theorem~\ref{thm:long_time_bound} might not be very straightforward to apply in many situations because it involves an extra parameter $\Delta$. A more versatile version would be
\begin{thm}
\label{thm:long_time_bound2}
Let $H=H_W+H_R$ be a Hamiltonian satisfying \eqref{eq:conditions_ham_general} with parameters $\chi$ and $r$. For any $t\geq 0$ and integers $\Lambda\geq\Lambda_0\geq 0$,
\[
\left\|\overline{\Pi}_{[-\Lambda,\Lambda]}e^{-itH}\Pi_{[-\Lambda_0,\Lambda_0]}\right\|\leq \poly(\chi t,\Lambda_0,\Lambda)\exp\left(-\Omega\left(\frac{\Lambda^{1-r}-\Lambda_0^{1-r}}{2(1-r)\chi t+1}\right)\right).
\]
\end{thm}

\begin{proof}
In Theorem~\ref{thm:long_time_bound} we choose $\Delta$ so that 
\begin{equation}
    \Delta= \Big\lfloor\frac{\Lambda^{1-r}-\Lambda_0^{1-r}}{2(1-r)\chi t+1}\Big\rfloor + 1.
\end{equation}
If $\Delta=1$ then the claimed bound holds trivially since $\left\|\overline{\Pi}_{[-\Lambda,\Lambda]}e^{-itH}\Pi_{[-\Lambda_0,\Lambda_0]}\right\|\leq 1$. If $\Delta>1$ we use Theorem~\ref{thm:long_time_bound}.
For the $\Lambda(t)$ defined in Theorem~\ref{thm:long_time_bound} we have
\begin{equation}
    \begin{aligned}
    \Lambda(t) &= \Lambda_0 + \Big\lceil \frac{1}{\Delta-1}\left((\Lambda_0^{1-r}+2\chi t(1-r)(\Delta-1))^{\frac{1}{1-r}}-\Lambda_0\right) \Big\rceil(\Delta-1) \\
    &\leq (\Lambda_0^{1-r}+2\chi t(1-r)(\Delta-1))^{\frac{1}{1-r}} + \Delta-1 \\
    &\leq (\Lambda_0^{1-r}+2\chi t(1-r)(\Delta-1)+(\Delta-1))^{\frac{1}{1-r}} \\
    &\leq \Lambda,
    \end{aligned}
\end{equation}
where in the third line we have used the 
inequality that $a^p+b\leq (a+b)^p$ when $a\geq 0$, $p\geq 1$ and $b$ (to be chosen as $\Delta-1$) is a non-negative integer.
Using the fact that $2^{1-\Delta}\left(\Delta!\right)^{-1/2}=e^{-\Omega(\Delta)}$, the claim follows immediately from Theorem~\ref{thm:long_time_bound}.
\end{proof}

If we want to ensure that truncating at a threshold $\Lambda$ has an error $\left\|\overline{\Pi}_{[-\Lambda,\Lambda]}e^{-itH}\Pi_{[-\Lambda_0,\Lambda_0]}\right\|\leq \epsilon$, then by Theorem~\ref{thm:long_time_bound2} we can choose
\[
\Lambda^{1-r} = \Lambda_0^{1-r} + \wt{\Or}((\chi t+1)\polylog(\epsilon^{-1})).
\]
This is the scaling given in (3) of the main text.

\section{Truncating the Hamiltonian}
\label{sec:truncating_the_hamiltonian}

In this section we consider the problem of replacing an unbounded Hamiltonian $H$, such as one describing boson-fermion interactions or lattice gauge theories, with a bounded Hamiltonian, while keeping the error in time evolution small. More precisely, we want to construct some bounded $\wt{H}$ such that
\begin{equation}
\label{eq:ham_truncation_error}
    \left\|(e^{-it\wt{H}}-e^{-itH})\Pi^{\mathrm{all}}_{[-\Lambda_0,\Lambda_0]}\right\|
\end{equation}
is sufficiently small.

In the previous section we have focused on a single bosonic mode or gauge link, but here the truncation needs to be performed for every bosonic mode or gauge link, and we assume there are $\mathcal{N}$ of them in the system.
To simplify the discussion, we use $\nu$ to index either bosonic modes or gauge links, replacing the indices $\alpha$ and $(x,n)$. Therefore we have 
\begin{equation}
    \Pi_{S}^{\mathrm{all}} = \prod_{\nu=1}^{\mathcal{N}} \Pi_{S}^{(\nu)}.
\end{equation}
Note also that all projection operators $\Pi_S^{(\nu)}$ commute with each other.
Then for each $\nu$, there is a decomposition of the Hamiltonian $H=H_W^{(\nu)}+H_R^{(\nu)}$.

We will establish the following bound.
\begin{thm}[Hamiltonian truncation]
\label{thm:truncate_ham}
Let $H$ be a Hamiltonian such that $H=H_W^{(\nu)}+H_R^{(\nu)}$ satisfies \eqref{eq:conditions_ham_general} with parameters $\chi$ and $r$ for every $1\leq\nu\leq\mathcal{N}$. Assume that all projection operators $\Pi^{(\nu)}_\lambda$ commute with each other. For any integers $\wt{\Lambda}\geq\Lambda_0\geq 0$, define
\begin{equation}
\label{eq:truncate_ham_in_thm}
    \wt{H} = \Pi^{\mathrm{all}}_{[-\wt{\Lambda},\wt{\Lambda}]}H \Pi^{\mathrm{all}}_{[-\wt{\Lambda},\wt{\Lambda}]}.
\end{equation}
Then for any \changed{$t\in\mathbb{R}$},
\begin{equation}
    \left\|(e^{-it\wt{H}}-e^{-itH})\Pi^{\mathrm{all}}_{[-\Lambda_0,\Lambda_0]}\right\|\leq \sqrt{\mathcal{N}}\mathcal{A}(\wt{\Lambda})\poly(\chi \changed{|t|},\Lambda_0,\wt{\Lambda})\exp\left(-\Omega\left(\frac{\wt{\Lambda}^{1-r}-\Lambda_0^{1-r}}{4(1-r)\chi \changed{|t|}+1}\right)\right),
\end{equation}
where
\begin{equation}
    \mathcal{A}(\Lambda) = \left\|\left[H,\Pi^{\mathrm{all}}_{[-\Lambda,\Lambda]}H\Pi^{\mathrm{all}}_{[-\Lambda,\Lambda]}\right]\right\|.
\end{equation}
\end{thm}

We recall that $r=1/2$ for boson-fermion coupling and $r=0$ for lattice gauge theories.
We also note that for both boson-fermion coupling and the lattice gauge theories, $\mathcal{A}(\Lambda)$ can be bounded by a polynomial of the Hamiltonian coefficients and $\Lambda$. This is because
\begin{equation}
    \left[H,\Pi^{\mathrm{all}}_{[-\Lambda,\Lambda]}H\Pi^{\mathrm{all}}_{[-\Lambda,\Lambda]}\right] = (H\Pi^{\mathrm{all}}_{[-\Lambda,\Lambda]})^2-(\Pi^{\mathrm{all}}_{[-\Lambda,\Lambda]}H)^2,
\end{equation}
and in all examples we discussed in Section~\ref{sec:motivating_examples}, the norm of $H\Pi^{\mathrm{all}}_{[-\Lambda,\Lambda]}$ is bounded by a function that is linear in all coefficients, and linear or quadratic
in $\Lambda$ for boson-fermion coupling and lattice gauge theories respectively. 

\begin{rem}
\label{rem:ham_truncation_models}
For the boson-fermion coupling in \eqref{eq:ham_fermion_boson_general}, to ensure that the error in truncating the time evolution \eqref{eq:ham_truncation_error} is below $\epsilon$, Theorem~\ref{thm:truncate_ham} implies that a threshold of 
\begin{equation}
    \label{eq:truncation_threshold_boson_fermion}
    \wt{\Lambda} = \left(\sqrt{\Lambda_0}+\wt{\Or}((\chi \changed{|t|}+1)\polylog(\mathrm{coef},\Lambda_0,\changed{|t|},\epsilon^{-1}))\right)^2
\end{equation}
suffices, where $\mathrm{coef}$ includes all the coefficients in the model \eqref{eq:ham_fermion_boson_general} as well as the number of bosonic and fermionic modes. 

Similarly, for U(1) and SU(2) lattice gauge theories, it suffices to truncate the electric field and total angular momentum, for the two situations respectively, at
\begin{equation}
    \label{eq:truncation_threshold_LGT}
    \wt{\Lambda} = \Lambda_0+\wt{\Or}((\chi \changed{|t|}+1)\polylog(\mathrm{coef},\Lambda_0,\changed{|t|},\epsilon^{-1})),
\end{equation}
where $\mathrm{coef}$ includes all the coefficients in the models \eqref{eq:LGT_ham} as well as the number of lattice sites. The values of $\chi$ in all cases are discussed in Section~\ref{sec:common_structures}.
\end{rem}

To establish Theorem~\ref{thm:truncate_ham}, we use the following properties of the truncated Hamiltonian $\Pi^{\mathrm{all}}_{[-\Lambda,\Lambda]}H\Pi^{\mathrm{all}}_{[-\Lambda,\Lambda]}$:
\begin{align}
    & (H-\Pi^{\mathrm{all}}_{[-\Lambda,\Lambda]}H\Pi^{\mathrm{all}}_{[-\Lambda,\Lambda]})\Pi^{\mathrm{all}}_{[-\Lambda+1,\Lambda-1]} = 0, 
    \label{eq:HHwt_relation1} \\
    & \left[H,\Pi^{\mathrm{all}}_{[-\Lambda,\Lambda]}H\Pi^{\mathrm{all}}_{[-\Lambda,\Lambda]}\right]\Pi^{\mathrm{all}}_{[-\Lambda+2,\Lambda-2]} = 0,
    \label{eq:HHwt_relation2}
\end{align}
which follow immediately from the fact that for each $\nu$, $H$ can only change the quantum number $\lambda$ by $\pm 1$.

Because we are now studying the whole system rather than a single mode or link, we need to bound the total leakage from the leakage at each individual $\nu$. This is done through a union bound, as given in the following lemma:
\begin{lem}
\label{lem:union_bound}
Let $\Pi^{(\nu)}_{\lambda}$ be projections all commuting with each other. For any operator $A$ and set $S\subset \ZZ$, we have
\begin{equation}
    \left\|\overline{\Pi}^{\mathrm{all}}_{S}A\right\|^2\leq \sum_{\nu=1}^{\mathcal{N}} \left\|\overline{\Pi}^{(\nu)}_{S}A\right\|^2,
\end{equation}
where $\overline{\Pi}^{\mathrm{all}}_{S}=I-{\Pi}^{\mathrm{all}}_{S}$.
\end{lem}
\begin{proof}
Because $\Pi^{(\nu)}_{\lambda}$ commute with each other, they can be simultaneously diagonalized, and by the union bound we have
\begin{equation}
    \overline{\Pi}^{\mathrm{all}}_{S} \preceq \sum_{\nu} \overline{\Pi}^{(\nu)}_{S},
\end{equation}
which in turn leads to
\begin{equation}
    \left\|\overline{\Pi}^{\mathrm{all}}_{S}A\right\|^2=\left\|A^{\dagger}\overline{\Pi}^{\mathrm{all}}_{S}A\right\|\leq \left\|A^{\dagger}\sum_{\nu} \overline{\Pi}^{(\nu)}_{S}A\right\| \leq \sum_{\nu} \left\|\overline{\Pi}^{(\nu)}_{S}A\right\|^2,
\end{equation}
where the first step can be proven using the singular value decomposition, and we have used the fact that $A^{\dagger} \overline{\Pi}^{\mathrm{all}}_{S}A$ and $A^{\dagger}  \overline{\Pi}^{(\nu)}_{S}A$ are all positive semidefinite for the later steps. 
\end{proof}

Note that the above union bound actually holds even when the commutativity assumption about $\Pi^{(\nu)}_{\lambda}$ is dropped; see \cite[Lemma 2]{Sen2011Achieving} for details.
With these we are ready to prove the main result of this section.
\begin{proof}[Proof of Theorem~\ref{thm:truncate_ham}]
Using the formula for the error of the first-order Trotter decomposition \cite[Eq.~(3.4)]{suzuki1985decomposition}, we have
\begin{equation}
\label{eq:first_order_trotter_error_exact}
    e^{it\wt{H}}e^{-it H} = e^{-it(H-\wt{H})} + \int_0^{t} \dd s_1 \int_{0}^{s_1}\dd s_2\ e^{is_1 \wt{H}}e^{-i(s_1-s_2){H}}[\wt{H},H]e^{-is_2{H}}e^{-i(t-s_1)(H-\wt{H})}.
\end{equation}
We then use this, {along with invariance of the spectral norm under multiplication by a unitary operator,} to bound the truncation error as:
\begin{equation}
\label{eq:ham_truncation_error_first_decompose}
\begin{aligned}
    \left\|\left(e^{-it\wt{H}}-e^{-itH}\right)\Pi_{[-\Lambda_0,\Lambda_0]}^{\mathrm{all}}\right\| &= \left\|\left(I-e^{it\wt{H}}e^{-itH}\right)\Pi_{[-\Lambda_0,\Lambda_0]}^{\mathrm{all}}\right\| \\
    &\leq \left\|\left(I-e^{-it(H-\wt{H})}\right)\Pi_{[-\Lambda_0,\Lambda_0]}^{\mathrm{all}}\right\| \\
    &\quad+ \int_0^{t} \dd s_1 \int_{0}^{s_1}\dd s_2 \left\|[\wt{H},H]e^{-is_2{H}}e^{-i(t-s_1)(H-\wt{H})}\Pi_{[-\Lambda_0,\Lambda_0]}^{\mathrm{all}}\right\|.
\end{aligned}
\end{equation}
Now if we choose $\wt{\Lambda}\geq \Lambda_0+1$, then by \eqref{eq:HHwt_relation1},
\begin{equation}
    \left(I-e^{-it(H-\wt{H})}\right)\Pi_{[-\Lambda_0,\Lambda_0]}^{\mathrm{all}}=0.
\end{equation}
As a result the second line of \eqref{eq:ham_truncation_error_first_decompose} is $0$. We now only need to bound the integrand in the third line.

For this integrand we have 
\begin{equation}
\label{eq:ham_truncation_error_second_decompose}
    \begin{aligned}
    \left\|[\wt{H},H]e^{-is_2{H}}e^{-i(t-s_1)(H-\wt{H})}\Pi_{[-\Lambda_0,\Lambda_0]}^{\mathrm{all}}\right\| &= \left\|[\wt{H},H]e^{-is_2{H}}\Pi_{[-\Lambda_0,\Lambda_0]}^{\mathrm{all}}\right\| \\
    &\leq \left\|[\wt{H},H]\Pi_{[-\Lambda',\Lambda']}^{\mathrm{all}}e^{-is_2{H}}\Pi_{[-\Lambda_0,\Lambda_0]}^{\mathrm{all}}\right\| \\
    &\quad+ \left\|[\wt{H},H]\overline{\Pi}_{[-\Lambda',\Lambda']}^{\mathrm{all}}e^{-is_2{H}}\Pi_{[-\Lambda_0,\Lambda_0]}^{\mathrm{all}}\right\|,
    \end{aligned}
\end{equation}
for some $\Lambda'$ to be chosen. We choose $\Lambda'= \wt{\Lambda}-2$. With this choice and \eqref{eq:HHwt_relation2} we have
\begin{equation}
    [\wt{H},H]\Pi_{[-\Lambda',\Lambda']}^{\mathrm{all}}=0.
\end{equation}
This eliminates the right-hand side on the second line of \eqref{eq:ham_truncation_error_second_decompose}.
Therefore we are only left with the third line of \eqref{eq:ham_truncation_error_second_decompose} to deal with.

We apply Theorem~\ref{thm:long_time_bound2}, as well as Lemma~\ref{lem:union_bound}, to get
\begin{equation}
    \left\|\overline{\Pi}_{[-\Lambda',\Lambda']}^{\mathrm{all}}e^{-is_2{H}}\Pi_{[-\Lambda_0,\Lambda_0]}^{\mathrm{all}}\right\|\leq \sqrt{\mathcal{N}}\poly(\chi t,\Lambda_0,\Lambda')\exp\left(-\Omega\left(\frac{\Lambda'^{1-r}-\Lambda_0^{1-r}}{2(1-r)\chi t+1}\right)\right),
\end{equation}
where we have used the fact that $s_2\leq t$. Substituting this bound into \eqref{eq:ham_truncation_error_second_decompose} and then \eqref{eq:ham_truncation_error_first_decompose}, we have
\begin{equation}
    \left\|\left(e^{-it\wt{H}}-e^{-itH}\right)\Pi_{[-\Lambda_0,\Lambda_0]}^{\mathrm{all}}\right\| \leq \sqrt{\mathcal{N}}\mathcal{A}(\wt{\Lambda})\poly(\chi t,\Lambda_0,\wt{\Lambda})\exp\left(-\Omega\left(\frac{\wt{\Lambda}^{1-r}-\Lambda_0^{1-r}}{2(1-r)\chi t+1}\right)\right).
\end{equation}
In the above derivation we used the fact that $\Lambda'^{1-r}= (\wt{\Lambda}-2)^{1-r}=\wt{\Lambda}^{1-r}+o(1)$ as $\wt{\Lambda}\to\infty$. This completes the proof of the theorem.

\end{proof}

One can ask the following question about the proof above: can the energy-based truncation threshold proposed in \cite{JordanLeePreskill2012quantum}, and discussed in detail in Section~\ref{sec:comparison_with_the_energy_based_truncation_threshold}, be justified as a truncation threshold for Hamiltonian truncation, through a proof that is similar to the proof above? We remark that this may require different assumptions and the proof will need to be substantially modified. If one were to use the above proof strategy, together with the energy-based truncation threshold, to derive a truncation threshold for the Hamiltonian truncation, then an important obstacle is bounding the third line in \eqref{eq:ham_truncation_error_second_decompose}. This line is bounded, in the proof above, through
\[
 \left\|[\wt{H},H]\overline{\Pi}_{[-\Lambda',\Lambda']}^{\mathrm{all}}e^{-is_2{H}}\Pi_{[-\Lambda_0,\Lambda_0]}^{\mathrm{all}}\right\| \leq  \left\|[\wt{H},H]\right\|\left\|\overline{\Pi}_{[-\Lambda',\Lambda']}^{\mathrm{all}}e^{-is_2{H}}\Pi_{[-\Lambda_0,\Lambda_0]}^{\mathrm{all}}\right\|.
\]
On the right-hand side, $\left\|[\wt{H},H]\right\|$ grows polynomially with the truncation threshold $\wt{\Lambda}$, while $\left\|\overline{\Pi}_{[-\Lambda',\Lambda']}^{\mathrm{all}}e^{-is_2{H}}\Pi_{[-\Lambda_0,\Lambda_0]}^{\mathrm{all}}\right\|$ decays subexponentially with $\Lambda'$. Therefore asymptotically the latter decays faster than the former and consequently we can reach an arbitrarily high precision. If we could only use the energy-based truncation threshold, then the latter term only decays polynomially with $\Lambda'$, and as a result a careful comparison between the rates of growth and decay of the two terms would be needed, and we could only reach an arbitrarily high precision when the latter decays faster than the former. This would require further assumptions not included in our framework.

Moreover, the most appealing feature of the energy-based truncation threshold is that it does not depend on time. However, suppose one could overcome the above mentioned difficulty; then the energy-based quantum state truncation threshold would lead to a time-dependent Hamiltonian truncation threshold, because of the integration over time in \eqref{eq:first_order_trotter_error_exact}, and thus the above appealing feature no longer holds.

\section{Hamiltonian simulations using the HHKL decomposition}
\label{sec:application_to_ham_sim_HHKL}

In this section we consider performing Hamiltonian simulation for U(1) and SU(2) lattice gauge theories and boson-fermion coupling. The basic idea is to simulate the truncated Hamiltonian $\wt{H}$ defined in \eqref{eq:truncate_ham_in_thm} as opposed to the unbounded $H$, with the truncation threshold $\wt{\Lambda}$ chosen according to \eqref{eq:truncation_threshold_boson_fermion} and \eqref{eq:truncation_threshold_LGT} for boson-fermion coupling and lattice gauge theories respectively.

\subsection{Simulating lattice gauge theories}
\label{sec:sim_LGT}

In this section we propose an algorithm to simulate the time evolution of the $(D+1)$-dimensional U(1) and SU(2) lattice gauge theories whose Hamiltonians are of the form \eqref{eq:LGT_ham}. The goal is to perform simulation of a square lattice consisting of $N$ sites up to time $T$ with an error at most $\epsilon$. This algorithm is based on a combination of the HHKL decomposition \cite{HaahHastingsKothariLow2021quantum} and interaction picture Hamiltonian simulation \cite{LowWiebe2018interaction}. We will show that the simulation can be done with gate complexity $\Or(NT\polylog(\Lambda_0 NT\epsilon^{-1}))$, where $\epsilon$ is the allowed error, 
assuming that the initial state is in the span of states whose quantum number is in the range $[-\Lambda_0,\Lambda_0]$ for each gauge link.

As mentioned above, we will be simulating $\wt{H}$ instead of $H$, and the resulting error has been analyzed in Section~\ref{sec:truncating_the_hamiltonian}.
We use $\wt{H}_E$ to denote the truncated electric field part of the Hamiltonian, i.e. $\wt{H}_E=\Pi^{\mathrm{all}}_{[-\wt{\Lambda},\wt{\Lambda}]}H_E\Pi^{\mathrm{all}}_{[-\wt{\Lambda},\wt{\Lambda}]}$, and we adopt similar notation $\wt{H}_M$, $\wt{H}_{GM}$, and $\wt{H}_{B}$ for the other three parts. Moreover we denote
\begin{equation}
    \wt{E}^2_{x,n} = \Pi^{(x,n)}_{[-\wt{\Lambda},\wt{\Lambda}]}E^2_{x,n}\Pi^{(x,n)}_{[-\wt{\Lambda},\wt{\Lambda}]}.
\end{equation}

Note that the Hamiltonians for lattice gauge theories, both the original $H$ and $\wt{H}$, consist of geometrically local terms, and to achieve a linear scaling in both system size and time we consider using the HHKL decomposition developed in \cite{HaahHastingsKothariLow2021quantum}.

\subsubsection{The HHKL decomposition}
\label{sec:sim_LGT_HHKL}

We first use \cite[Lemma~6]{HaahHastingsKothariLow2021quantum} to decompose the time evolution of the entire system into evolution of blocks, each of which, denoted by $\mathcal{B}$, has size $\ell^D=\Or(\polylog(NT\epsilon^{-1}))$ and we only need to simulate its evolution for time $\tau=\Or(1)$.
The entire time evolution is divided into $\Or(T)$ segments and there are $\Or(N)$ such blocks within each segment.

The original \cite[Lemma~6]{HaahHastingsKothariLow2021quantum} requires that all the local terms in the target Hamiltonian have norm bounded by a constant. However, local terms in $\wt{H}$ have norm depending on the truncation threshold $\wt{\Lambda}$, which scales with the system size $N$, time $T$, and allowed error $\epsilon$. We address this issue as follows. 
The only local terms that are not bounded by a constant are the electric field terms in $\wt{H}_E$, i.e. $g_E\wt{E}^2_{x,n}$ for each $(x,n)$, and each of these terms only acts on a single gauge link. We call such terms, i.e. terms that act only on a single lattice site or gauge link (which can be seen as a lattice site as well for this purpose), \emph{on-site} interactions. In Lemma~\ref{lem:lieb_robinson_on_site} of Section~\ref{sec:LR_veclocity_on_site} we show that on-site interactions do not change the Lieb-Robinson velocity. Therefore, even with the terms in $\wt{H}_E$, the system still has a constant Lieb-Robinson velocity, and as a result we can invoke \cite[Lemma~6]{HaahHastingsKothariLow2021quantum} to decompose the time evolution.

\subsubsection{Simulating the blocks}

We now consider simulating the dynamics of each individual block $\mathcal{B}$ of size $\ell^D$ for short time $\tau$. The Hamiltonian for $\mathcal{B}$, which we denote by $\wt{H}^{\mathcal{B}}$, includes all the local terms in $\wt{H}$ that only act on sites and links within the block $\mathcal{B}$. As discussed in Section~\ref{sec:motivating_examples}, each local term can be represented by a sparse matrix with respect to the basis discussed in Section~\ref{sec:common_structures}, and can therefore be encoded by a quantum walk operator \cite{BerryChilds2012black,BerryChildsCleveEtAl2014exponential,ChildsKothariSomma2017}. A sum of these terms can be encoded in an unitary using the linear combination of unitaries (LCU) method \cite{BerryChildsCleveEtAl2014exponential,ChildsKothariSomma2017}. In this way we have an encoding (known as ``block-encoding'' or ``standard-form'' in \cite{low2019hamiltonian,gilyen2019quantum}) of the Hamiltonian $\wt{H}^{\mathcal{B}}$, i.e. a unitary with $\wt{H}^{\mathcal{B}}$ as a matrix block, with a subnormalization factor $\Or(\ell^D \wt{\Lambda}^2)$. Using the Hamiltonian simulation algorithm for encoded Hamiltonians \cite{low2019hamiltonian}, we can then simulate the time evolution of the block $\mathcal{B}$ with gate-complexity $\Or(\ell^{2D} \wt{\Lambda}^2\tau)=\Or(\wt{\Lambda}^2\polylog(NT\epsilon^{-1}))$, which results in a total gate complexity of $\wt{\Or}(NT(\Lambda_0+T)^2\polylog(\epsilon^{-1}))$. This is however not the best method in terms of asymptotic complexity.

The polynomial dependence on $\wt{\Lambda}$ can be improved to be poly-logarithmic using the interaction-picture simulation technique developed in \cite{LowWiebe2018interaction}. We group the local terms in $\mathcal{B}$ into $\wt{H}_{M}^{\mathcal{B}}$, $\wt{H}_{GM}^{\mathcal{B}}$, $\wt{H}_{E}^{\mathcal{B}}$, and $\wt{H}_{B}^{\mathcal{B}}$ depending on whether the term describes fermionic mass energy, gauge-matter interaction, electric field energy, or magnetic field energy. Then the polynomial dependence on $\wt{\Lambda}$ comes only from $\wt{H}_{E}^{\mathcal{B}}$. We note that the time evolution under $\wt{H}_{E}^{\mathcal{B}}$ can be fast-forwarded, i.e. the number of gates required to implement it has a poly-logarithmic dependence on the evolution time multiplied by the Hamiltonian norm. To be more specific, the time evolution due to each electric field term $g_E\wt{E}^2_{x,n}$ for time $t$ can be implemented with gate complexity $\Or(\polylog(\wt{\Lambda}t))$ because this term is represented by a diagonal matrix in both U(1) and SU(2) settings (see \eqref{eq:operators_U(1)} and \eqref{eq:operators_SU(2)} for the two settings respectively). And all of these terms act on different gauge links and therefore commute with each other. To implement $e^{-it\wt{H}_{E}^{\mathcal{B}}}$ we only need to evolve these terms separately, and thus pay a cost of $\Or(\ell^D \polylog(\wt{\Lambda}t))=\Or(\polylog(NT\wt{\Lambda}t\epsilon^{-1}))$.

Now instead of directly simulating the Hamiltonian $\wt{H}^{\mathcal{B}}$, we simulate
\begin{equation}
    \wt{H}^{\mathcal{B}}_I(t) = e^{it\wt{H}_E^{\mathcal{B}}}(\wt{H}_M^{\mathcal{B}}+\wt{H}_{GM}^{\mathcal{B}}+\wt{H}_{B}^{\mathcal{B}})e^{-it\wt{H}_E^{\mathcal{B}}}.
\end{equation}
The original time evolution $e^{-it\wt{H}^{\mathcal{B}}}$ and the interaction picture evolution are related through
\begin{equation}
    e^{-it\wt{H}^{\mathcal{B}}} = e^{-it\wt{H}_E^{\mathcal{B}}}\mathcal{T}e^{-i\int_0^t \dd s \wt{H}^{\mathcal{B}}_I(s)}.
\end{equation}
It then suffices to implement $\mathcal{T}e^{-i\int_0^t \dd s \wt{H}^{\mathcal{B}}_I(s)}$, namely the time evolution due to the time-dependent Hamiltonian $\wt{H}^{\mathcal{B}}_I(s)$. We accomplish this using the truncated Dyson series method in \cite[Corollary~4]{LowWiebe2018interaction}. The time-dependent matrix encoding in \cite[Definition~2]{LowWiebe2018interaction} can be constructed using the the encoding of the local Hamiltonian terms as well as the fast-forwarding of $\wt{H}_E^{\mathcal{B}}$ discussed above. This yields a gate complexity $\Or(\ell^{2D}\tau\polylog(NT\wt{\Lambda}\epsilon^{-1}))=\Or(\polylog(NT\wt{\Lambda}\epsilon^{-1}))$ for implementing the interaction picture time evolution and consequently $e^{-it\wt{H}^{\mathcal{B}}}$ can be implemented with the same gate complexity scaling. Note that here we want to keep the error for simulating this block to be at most $\Or(N^{-1}T^{-1}\epsilon)$ instead of $\epsilon$. This however does not significantly increase the asymptotic scaling because the scaling with respect to the allowed error is poly-logarithmic.

There are in total $\Or(NT)$ such simulations to perform for all the $\Or(N)$ blocks and $\Or(T)$ times steps. Therefore the total gate complexity for implementing the time evolution of the entire system is $\Or(NT\polylog(NT\wt{\Lambda}\epsilon^{-1}))$. Using the truncation threshold given in \eqref{eq:truncation_threshold_LGT}, we have that the total gate complexity for simulating the U(1) or SU(2) lattice gauge theory with $N$ sites up to time $T$ and allowed error $\epsilon$ is $\Or(NT\polylog(\Lambda_0 NT\epsilon^{-1}))$, provided that the initial state is in the support of $\Pi^{\mathrm{all}}_{[-\Lambda_0,\Lambda_0]}$, i.e. the quantum numbers on each gauge link are in the interval $[-\Lambda_0,\Lambda_0]$.

We remark that in using the HHKL decomposition, we need to preserve the locality of fermionic operators after the Jordan-Wigner transformation. This can be done by introducing auxiliary fermionic modes as discussed in \cite{HaahHastingsKothariLow2021quantum} using the method developed in \cite{VerstraeteCirac2005mapping}.

\subsection{Simulating boson-fermion coupling}
\label{sec:sim_boson_fermion}

In this section we consider simulating the Hubbard-Holstein model \cite{holstein1959studies}, which is the simplest model describing the electron-phonon interaction. 
This model is defined on a $D$-dimensional lattice, and each side of the lattice contains $L$ sites where $L^D=N$. 
Each site $x$ in the lattice contains two fermionic modes $c_{x,\sigma}$ (with $\sigma$ denoting either spin up and down) and a bosonic mode $b_x$.
We are interested in the case where $D$ is a constant. The Hamiltonian is
\begin{equation}
    \label{eq:ham_hubbard_holstein}
    H = H_f + H_{fb} + H_b,
\end{equation}
where $H_{f}$ is the Hamiltonian of the Fermi-Hubbard model: 
\begin{equation}
    \label{eq:ham_fermi_hubbard}
    H_{f} = -\sum_{\braket{x,x'},\sigma}(c_{x,\sigma}^{\dagger}c_{x',\sigma}+\mathrm{h.c.})+ U\sum_{x=1}^N (n_{x,\uparrow}-\frac{1}{2})(n_{x,\downarrow}-\frac{1}{2}) - \mu\sum_{r=1}^N n_x,
\end{equation}
and
\begin{equation}
    H_{fb} =  g\sum_{x=1}^N (b_x^{\dagger}+b_x)(n_{x,\uparrow}+n_{x,\downarrow}-1) \quad 
    H_b = \omega_0 \sum_{x=1}^N b_{x}^{\dagger}b_x,
\end{equation}
are the boson-fermion coupling part and bosonic part respectively. The lattice sites are indexed by $x$ and $x'$, and spins are indexed by $\sigma$. It is easy to verify that this model satisfies the general form of boson-fermion coupling in \eqref{eq:ham_fermion_boson_general}.

Here, we propose an algorithm that simulates the above model up to time $T$ and error $\epsilon$ with gate complexity $\wt{\Or}(NT(\sqrt{\Lambda_0}+T)\polylog(\epsilon^{-1}))$, assuming the initial state has no more than $\Lambda_0$ particles on each bosonic mode. Just like in the previous section the algorithm is based on HHKL decomposition \cite{HaahHastingsKothariLow2021quantum} and interaction picture Hamiltonian simulation \cite{LowWiebe2018interaction}.

First we replace the exact Hamiltonian $H$ with the truncated Hamiltonian $\wt{H}$ in \eqref{eq:truncate_ham_in_thm} and simulate the evolution of $\wt{H}$. The resulting error is analyzed in Section~\ref{sec:truncating_the_hamiltonian}. We also denote different parts of the Hamiltonian after truncation by $\wt{H}_f$, $\wt{H}_{fb}$, and $\wt{H}_b$. 

We apply the HHKL method to decompose the entire time evolution into evolution of blocks, each of which is denoted by $\mathcal{B}$, for short time $\tau$. Here again we encounter local terms whose norms are not bounded by a constant, and in this case these terms are contained in $\wt{H}_{fb}$ and $\wt{H}_b$. With the help of Lemma~\ref{lem:lieb_robinson_on_site} however, we can still apply the HHKL decomposition because $\wt{H}_{fb}$ and $\wt{H}_b$ are both on-site. So far the algorithm proceeds in a similar way as that for the lattice gauge theories.

Then we apply interaction picture Hamiltonian simulation to simulate the evolution in each block $\mathcal{B}$. 
We denote by $\wt{H}_f^{\mathcal{B}}$, $\wt{H}_{fb}^{\mathcal{B}}$, and $\wt{H}_b^{\mathcal{B}}$ the fermionic, coupling, and bosonic terms respectively. 
Here, the terms in $\wt{H}_b^{\mathcal{B}}$ can still be fast-forwarded the same way as the electric field terms in lattice gauge theories. 
However, it is not known whether the boson-fermion coupling terms in $\wt{H}_{fb}^{\mathcal{B}}$ can be fast-forwarded.
Therefore when we simulate the interaction picture Hamiltonian
\begin{equation}
    H_I^{\mathcal{B}} = e^{it\wt{H}_b^{\mathcal{B}}}(\wt{H}_f^{\mathcal{B}}+\wt{H}_{fb}^{\mathcal{B}})e^{-it\wt{H}_b^{\mathcal{B}}},
\end{equation}
the dependence on the truncation threshold is not poly-logarithmic. Rather a factor of $\sqrt{\wt{\Lambda}}$ shows up in the subnormalization factor of the encoding of the Hamiltonian because 
\begin{equation}
    \|\Pi^{\mathrm{all}}_{[0,\wt{\Lambda}]}(b_x^{\dagger}+b_x)(n_{x,\uparrow}+n_{x,\downarrow}-1)\Pi^{\mathrm{all}}_{[0,\wt{\Lambda}]}\| \leq 2\sqrt{\wt{\Lambda}+1}.
\end{equation}
There is some subtle difference between the subnormalization factor and the spectral norm, but in the present case they have the same asymptotic scaling.
The number of gates required to simulate the time evolution of a block $\mathcal{B}$ for time $\tau$ is then $\Or(\ell^{2D}\sqrt{\wt{\Lambda}}\tau\polylog(\wt{\Lambda}NT\epsilon^{-1}))=\Or(\sqrt{\wt{\Lambda}}\polylog(\wt{\Lambda}NT\epsilon^{-1}))$ with a target accuracy of $\Or(N^{-1}T^{-1}\epsilon)$ for each block.

We need to perform $\Or(NT)$ such simulations and the number of required gates is therefore $\wt{\Or}(NT\sqrt{\wt{\Lambda}}\polylog(\epsilon^{-1}))$. The truncation threshold $\wt{\Lambda}$ can be chosen according to \eqref{eq:truncation_threshold_boson_fermion}. As a result the total gate complexity for simulating an $N$-site Hubbard-Holstein model for time $T$ up to error $\epsilon$ is $\wt{\Or}(NT(\sqrt{\Lambda_0}+T)\polylog(\epsilon^{-1}))$, assuming the initial state has at most $\Lambda_0$ particles in each bosonic mode.

\section{Simulating boson-fermion coupling using Trotterization}
\label{sec:boson_fermion_trotter}

In this section, we consider a generic Hamiltonian of a fermion-boson coupled system.
We assume there are $N_f$ fermionic modes (labeled by $i,j$) and $N_b$ bosonic modes (labeled by $\alpha$). We have
\begin{equation}
    \label{eq:ham_fermion_boson_1}
    \begin{aligned}
    H &= \sum_{\gamma=1}^6 H_{\gamma}, \\
    H_1 &= \sum_{ij}t_{ij}c_{i}^{\dagger}c_j, \quad
    H_2 = \sum_{ij}V_{ij}n_i n_j, \\
    H_3 &= \sum_{\alpha ij} g_{ij}^{(\alpha)} c_i^{\dagger}c_j X_{\alpha}, \quad
    H_4 = \sum_{\alpha ij} h_{ij}^{(\alpha)} c_i^{\dagger}c_j P_{\alpha}, \\
    H_5 &= \frac{1}{2}\sum_{\alpha} \omega_{\alpha} X_{\alpha}^2, \quad
    H_6 = \frac{1}{2}\sum_{\alpha} \omega_{\alpha} P_{\alpha}^2,
    \end{aligned}
\end{equation}
where $X_{\alpha}=(b_{\alpha}+b_{\alpha}^{\dagger})/\sqrt{2}$ is the position operator corresponding to the bosonic mode $\alpha$, and $P_{\alpha}=i(b^{\dagger}_{\alpha}-b_{\alpha})/\sqrt{2}$ is the momentum operator. Matrices $t=(t_{ij})$, $V=(V_{ij})$, $g^{(\alpha)}=(g^{(\alpha)}_{ij})$, and $h^{(\alpha)}=(h_{ij}^{(\alpha)})$ are all Hermitian.
Such a Hamiltonian is in general not geometrically local, and therefore the HHKL decomposition we have introduced in the previous section is not directly applicable here. Instead, we consider performing quantum simulation using Trotterization, which can be very efficient when many terms in the Hamiltonian commute~\cite{ChildsSuTranWiebeZhu2021commutator} and can be easier to realize in practice.

The Hamiltonians introduced above belong to a subclass of \eqref{eq:ham_fermion_boson_general} as the fermion interaction term has two indices instead of four. We impose this restriction to ensure that this term can be fast-forwarded. 
We also require $\sum_{ij}g_{ij}^{(\alpha)} c_i^{\dagger}c_j$ to commute with each other for all $\alpha$, and require the same for $\sum_{ij}h_{ij}^{(\alpha)} c_i^{\dagger}c_j$ for all $\alpha$. These requirements are equivalent to 
\begin{equation}
\label{eq:trotter_coupling_commute}
    [g^{(\alpha)},g^{(\alpha')}]=[h^{(\alpha)},h^{(\alpha')}]=0,
\end{equation}
for all bosonic modes $\alpha,\alpha'$, where $g^{(\alpha)}=(g_{ij}^{(\alpha)})$ and $h^{(\alpha)}=(h_{ij}^{(\alpha)})$.
Note in particular that this holds for the Hubbard-Holstein model, since $(n_{x,\uparrow}+n_{x,\downarrow}-1)$ commutes with each other for all sites $x$.

This decomposition of the Hamiltonian into $H_{\gamma}$'s and the above requirements are to ensure that all the $H_{\gamma}$'s can be fast-forwarded. More precisely, the evolution $e^{-itH_{\gamma}}$ due to the truncated Hamiltonian terms can be implemented with gate complexity that depends poly-logarithmically on $t$. The fermionic part $e^{-itH_1}$ can be implemented using a series of Givens rotations \cite{KivlichanMcCleanWiebeEtAl2018quantum}. 
The circuit implementation becomes more challenging when a boson is involved. However, with a suitable particle number truncation, we can still implement the time evolution of these terms. We denote by $\wt{H}_{\gamma}$ the truncated Hamiltonian terms:
\begin{equation}
    \wt{H}_{\gamma} = \Pi_{[0,\wt{\Lambda}]}^{\mathrm{all}}H_{\gamma}\Pi_{[0,\wt{\Lambda}]}^{\mathrm{all}},
\end{equation}
for some $\wt{\Lambda}$ to be chosen.
The coupling parts $e^{-it\wt{H}_2}$ and $e^{-it\wt{H}_3}$ can be implemented using the technique described in \cite[Section IV]{MacridinEtAL2018digital}, where the boson is represented using the eigenbasis of $X_{\alpha}$, combined with the Givens rotation technique. We can implement the coupling term associated with each bosonic mode separately because of our requirement \eqref{eq:trotter_coupling_commute}.
The bosonic parts $e^{-it\wt{H}_4}$ and $e^{-it\wt{H}_5}$ are straightforward to implement using the same technique described in \cite[Section IV]{MacridinEtAL2018digital}.
See also \cite{Sawaya2020} for alternative implementations of the bosonic operators and their cost comparison.

The rest of this section is devoted to analyzing how many times we need to apply $e^{-it\wt{H}_{\gamma}}$ in order to simulate the time evolution $e^{-iTH}$ using Trotterization, starting from an initial state that lies in the support of $\Pi^{\mathrm{all}}_{[0,\Lambda_0]}$ with no more than $\Lambda_0$ bosonic particles on each mode $\alpha$. 
Up to a constant-factor difference, this is also the number of Trotter steps required to reach a target precision, which in turn determines the gate complexity of Hamiltonian simulation.

\subsection{Sources of error}

There are two sources of error that we need to deal with. The first source of error comes from the fact that we are evolving the system with $\wt{H}$ instead of $H$, and this is already analyzed in Theorem~\ref{thm:truncate_ham}. The second source of error is the Trotter error, which will be our focus here. A simple bound for the Trotter error is readily available if we ignore the commutation relation between pairs of the Hamiltonian terms. But here we aim for the commutator scaling described in \cite{ChildsSuTranWiebeZhu2021commutator}, which can be much tighter when many terms in the Hamiltonian commute. 

There is a technical issue that prevents us from directly applying the result of \cite{ChildsSuTranWiebeZhu2021commutator}. After truncation, the commutation relation between the projected position and momentum operators
\begin{equation}
    \wt{X}_{\alpha}=\Pi_{[0,\wt{\Lambda}]}^{\alpha} X_{\alpha} \Pi_{[0,\wt{\Lambda}]}^{\alpha},\quad \wt{P}_{\alpha}=\Pi_{[0,\wt{\Lambda}]}^{\alpha} P_{\alpha} \Pi_{[0,\wt{\Lambda}]}^{\alpha},
\end{equation}
is different from the canonical commutation relation between the exact position and momentum operators. To address this, we use the fact that the exact commutation relation is recovered when the particle number is some distance below the truncation threshold $\wt{\Lambda}$, and this in turn requires carefully tracking the particle number under the exact and truncated time evolution respectively.

Our proof uses the following telescoping lemma:
\begin{lem}
\label{lem:decompose_unitary}
Let $\Pi$ be a projection operator and $U_j$, $\wt{U}_j$ ($j=1,2,\ldots,J$) be unitary operators. We have
\[
\left\|\left(\prod_{j=1}^J \wt{U}_j - \prod_{j=1}^J {U}_j\right)\Pi\right\| \leq \sum_{j=1}^J \left\|\left(\wt{U}_j-U_j\right)\left(\prod_{j'=1}^{j-1} {U}_{j'}\right)\Pi\right\|.
\]
\end{lem}
\begin{proof}
This inequality follows immediately from the identity
\[
\prod_{j=1}^J \wt{U}_j - \prod_{j=1}^J {U}_j = \sum_{j=1}^J \left(\prod_{j'=j+1}^{J} \wt{U}_{j'}\right)\left(\wt{U}_j-U_j\right)\left(\prod_{j'=1}^{j-1} {U}_{j'}\right),
\]
which can be proved by induction on $J$.
\end{proof}

Our goal is to simulate the dynamics up to time $T$. We achieve this by dividing the entire time evolution into $R=T/\tau$ steps, each of which has duration $\tau$ and is simulated by a $p$-th order product formula $S(\tau)$. Then the Trotter error can be bounded as
\begin{equation}
    \label{eq:trotter_error_total}
    \begin{aligned}
        \|(e^{-iT\wt{H}}-S(\tau)^R)\Pi^{\mathrm{all}}_{[0,\Lambda_0]}\| &\leq \sum_{j=1}^R \|(S(\tau)-e^{-i\tau \wt{H}})e^{-i(j-1)\tau \wt{H}}\Pi^{\mathrm{all}}_{[0,\Lambda_0]}\| \\
        &\leq \sum_{j=1}^R \|(S(\tau)-e^{-i\tau \wt{H}})\Pi^{\mathrm{all}}_{[0,\Lambda'_0]}e^{-i(j-1)\tau \wt{H}}\Pi^{\mathrm{all}}_{[0,\Lambda_0]}\| \\
        &\quad+ \sum_{j=1}^R \|(S(\tau)-e^{-i\tau \wt{H}})\overline{\Pi}^{\mathrm{all}}_{[0,\Lambda'_0]}e^{-i(j-1)\tau \wt{H}}\Pi^{\mathrm{all}}_{[0,\Lambda_0]}\| \\
        &\leq R \|(S(\tau)-e^{-i\tau \wt{H}})\Pi^{\mathrm{all}}_{[0,\Lambda'_0]}\| \\
        &\quad+ 2\sum_{j=1}^R \|\overline{\Pi}^{\mathrm{all}}_{[0,\Lambda'_0]}e^{-i(j-1)\tau \wt{H}}\Pi^{\mathrm{all}}_{[0,\Lambda_0]}\|,
    \end{aligned}
\end{equation}
where $\overline{\Pi}^{\mathrm{all}}_{[0,\Lambda'_0]}=I-\Pi^{\mathrm{all}}_{[0,\Lambda'_0]}$ is the complementary projection with $\Lambda'_0$ to be chosen later. Here in the first line we have used Lemma~\ref{lem:decompose_unitary}, and in the second and third line we have used the decomposition
\begin{equation}
    e^{-i(j-1)\tau \wt{H}}\Pi^{\mathrm{all}}_{[0,\Lambda_0]} = \Pi^{\mathrm{all}}_{[0,\Lambda'_0]}e^{-i(j-1)\tau \wt{H}}\Pi^{\mathrm{all}}_{[0,\Lambda_0]} + \overline{\Pi}^{\mathrm{all}}_{[0,\Lambda'_0]}e^{-i(j-1)\tau \wt{H}}\Pi^{\mathrm{all}}_{[0,\Lambda_0]}.
\end{equation}

By Theorem~\ref{thm:long_time_bound2} and Lemma~\ref{lem:union_bound} we have
\begin{equation}
\label{eq:choose_Lambda'_0_ineq}
    \|\overline{\Pi}^{\mathrm{all}}_{[0,\Lambda'_0]}e^{-i(j-1)\tau \wt{H}}\Pi^{\mathrm{all}}_{[0,\Lambda_0]}\| \leq \sqrt{N_b} \poly(\chi T,\Lambda_0,\Lambda_0')e^{-\Omega((\sqrt{\Lambda'_0}-\sqrt{\Lambda_0})/(\chi T+1))},
\end{equation}
where $\chi$ is given in \eqref{eq:chi_boson_fermion}.
We now only need to bound $\|(S(\tau)-e^{-i\tau \wt{H}})\Pi^{\mathrm{all}}_{[0,\Lambda'_0]}\|$.

\subsection{Trotter error with bounded particle number}

The main result of \cite[Theorem~3]{ChildsSuTranWiebeZhu2021commutator} is a bound on the Trotter error in terms of the spectral norm of nested commutators of Hamiltonian terms. That bound does not take into account the fact that the initial state has a finite number of particles and is thus not suitable for our purpose. Instead, we use an exact representation of Trotter error, which is provided in Theorems~3 and 5 of \cite{ChildsSuTranWiebeZhu2021commutator}.

In \cite[Theorem~3]{ChildsSuTranWiebeZhu2021commutator} they derive the following expression for the Trotter error:
\begin{equation}
    S(\tau) - e^{-i\tau \wt{H}} = \int_{0}^{\tau} \dd \tau_1 e^{-i(\tau-\tau_1)\wt{H}}S(\tau_1)T(\tau_1),
\end{equation}
and by \cite[Theorem~5]{ChildsSuTranWiebeZhu2021commutator}, $T(\tau_1)$ can be written as
\begin{equation}
    \begin{aligned}
    T(\tau_1) &= \sum_{|\vec{\gamma}|=p+1}\sum_{\ell=1}^{L_{\vec{\gamma}}}\sum_{q=1}^p \tau_1^{p-q}\int_{0}^{\tau_1} \dd \tau_2 C_{\vec{\gamma}\ell q} (\tau_1-\tau_2)^{q-1}\\ 
    &\qquad\qquad\qquad\qquad\cdot U^{\dagger}_{\vec{\gamma}\ell}(\tau_1,\tau_2)[\wt{H}_{\gamma_{p+1}},\cdots[\wt{H}_{\gamma_2},\wt{H}_{\gamma_1}]\cdots]U_{\vec{\gamma}\ell}(\tau_1,\tau_2),
    \end{aligned}
\end{equation}
where
\begin{equation}
    U_{\vec{\gamma}\ell}(\tau_1,\tau_2)  = e^{-ic_{\vec{\gamma}\ell 0}\tau_2 \wt{H}_{\vec{\gamma}\ell 0}}\prod_{\nu=1}^{R_{\vec{\gamma}\ell}}e^{-ic_{\vec{\gamma}\ell \nu}\tau_1 \wt{H}_{\vec{\gamma}\ell \nu}}
\end{equation}
is a product of operator exponentials. In the above equations $\wt{H}_{\vec{\gamma}\ell \nu}\in\{\wt{H}_1,\wt{H}_2,\ldots,\wt{H}_5\}$, $\vec{\gamma}=(\gamma_1,\gamma_2,\ldots,\gamma_{p+1})$ is a string of indices for Hamiltonian terms, and $L_{\vec{\gamma}}$, $C_{\vec{\gamma}\ell q}$, $c_{\vec{\gamma}\ell \nu}$, and $R_{\vec{\gamma}\ell}$ are 
constants that only depend on the Trotter formula but not on the Hamiltonian or time variable $\tau$. Also $C_{\vec{\gamma}\ell q}$ is non-zero only when $\gamma_{p+1}=\gamma_{p}=\cdots=\gamma_{p-q+2}$, 
but this property does not affect the asymptotic gate complexity and will thus be ignored in the subsequent analysis.

With this exact representation of Trotter error, we have
\begin{equation}
\label{eq:trotter_error_decomp_particle_num}
    \begin{aligned}
        \|(S(\tau)-e^{-i\tau \wt{H}})\Pi^{\mathrm{all}}_{[0,\Lambda'_0]}\| &\leq \sum_{|\vec{\gamma}|=p+1}\sum_{\ell=1}^{L_{\vec{\gamma}}}\sum_{q=1}^p \int_0^{\tau}\dd \tau_1\tau_1^{p-q}\int_{0}^{\tau_1} \dd \tau_2 |C_{\vec{\gamma}\ell q}| (\tau_1-\tau_2)^{q-1} \\
        &\qquad\qquad\qquad\qquad\times \|[\wt{H}_{\gamma_{p+1}},\cdots[\wt{H}_{\gamma_2},\wt{H}_{\gamma_1}]\cdots]U_{\vec{\gamma}\ell}(\tau_1,\tau_2)\Pi^{\mathrm{all}}_{[0,\Lambda'_0]}\| \\
        &\leq \sum_{|\vec{\gamma}|=p+1}\sum_{\ell=1}^{L_{\vec{\gamma}}}\sum_{q=1}^p \int_0^{\tau}\dd \tau_1\tau_1^{p-q}\int_{0}^{\tau_1} \dd \tau_2 |C_{\vec{\gamma}\ell q}| (\tau_1-\tau_2)^{q-1} \\
        &\qquad\qquad\qquad\qquad\times \big(\|[\wt{H}_{\gamma_{p+1}},\cdots[\wt{H}_{\gamma_2},\wt{H}_{\gamma_1}]\cdots]\Pi^{\mathrm{all}}_{[0,\Lambda'_1]}\| \\
        &\qquad\qquad\qquad\qquad\qquad+ \|[\wt{H}_{\gamma_{p+1}},\cdots[\wt{H}_{\gamma_2},\wt{H}_{\gamma_1}]\cdots]\|\|\overline{\Pi}^{\mathrm{all}}_{[0,\Lambda'_1]}U_{\vec{\gamma}\ell}(\tau_1,\tau_2)\Pi^{\mathrm{all}}_{[0,\Lambda'_0]}\| \big),
    \end{aligned}
\end{equation}
where we have used the decomposition
\begin{equation}
    U_{\vec{\gamma}\ell}(\tau_1,\tau_2)\Pi^{\mathrm{all}}_{[0,\Lambda'_0]} = \Pi^{\mathrm{all}}_{[0,\Lambda'_1]}U_{\vec{\gamma}\ell}(\tau_1,\tau_2)\Pi^{\mathrm{all}}_{[0,\Lambda'_0]} + \overline{\Pi}^{\mathrm{all}}_{[0,\Lambda'_1]}U_{\vec{\gamma}\ell}(\tau_1,\tau_2)\Pi^{\mathrm{all}}_{[0,\Lambda'_0]},
\end{equation}
for some $\Lambda'_1$ to be chosen later.
Since $\tau=\Or(1)$ ($\tau$ should be chosen to be much smaller than constant to suppress Trotter error) and $\sum_{\nu=0}^{R_{\vec{\gamma}\ell}}|c_{\vec{\gamma}\ell \nu}|$ is a constant, we have from Theorem~\ref{thm:long_time_bound2} that
\begin{equation}
\label{eq:leakage_within_trotter}
    \|\overline{\Pi}^{\mathrm{all}}_{[0,\Lambda'_1]}U_{\vec{\gamma}\ell}(\tau_1,\tau_2)\Pi^{\mathrm{all}}_{[0,\Lambda'_0]}\| = \poly(\chi,\Lambda'_0,\Lambda'_1)e^{-\Omega((\sqrt{\Lambda'_1}-\sqrt{\Lambda'_0})/(\chi+1))}.
\end{equation}
Here $U_{\vec{\gamma}\ell}(\tau_1,\tau_2)$ involves a constant number of operator exponentials each of which is generated by a term $\wt{H}_{\gamma}$ from the Hamiltonian. We note that Theorem~\ref{thm:long_time_bound2} applies to each operator exponential because $\wt{H}_{\gamma}$ also has the structure described in Section~\ref{sec:common_structures}. Thus we can apply Theorem~\ref{thm:long_time_bound2} a constant number of times to arrive at inequality \eqref{eq:leakage_within_trotter}.

We now follow \cite{ChildsSuTranWiebeZhu2021commutator} to define
\begin{equation}
    \wt{\alpha}_{\mathrm{comm}}(\wt{\Lambda}) = \sum_{|\vec{\gamma}|=p+1}\left\|[\wt{H}_{\gamma_{p+1}},\cdots[\wt{H}_{\gamma_2},\wt{H}_{\gamma_1}]\cdots]\right\|.
\end{equation}
There is a $\wt{\Lambda}$ dependence because the truncated Hamiltonian depends on the threshold $\wt{\Lambda}$. 
Furthermore, we define 
\begin{equation}
    \wt{\beta}_{\mathrm{comm}}(\wt{\Lambda},\Lambda'_1) = \sum_{|\vec{\gamma}|=p+1}\left\|[\wt{H}_{\gamma_{p+1}},\cdots[\wt{H}_{\gamma_2},\wt{H}_{\gamma_1}]\cdots]\Pi^{\mathrm{all}}_{[0,\Lambda'_1]}\right\|.
\end{equation}
Then by \eqref{eq:trotter_error_decomp_particle_num} and \eqref{eq:leakage_within_trotter}
\begin{equation}
\label{eq:trotter_error_small_time_bound}
\begin{aligned}
    \|(S(\tau)-e^{-i\tau \wt{H}})\Pi^{\mathrm{all}}_{[0,\Lambda'_0]}\| &= \Or(\wt{\alpha}_{\mathrm{comm}}(\wt{\Lambda})\tau^{p+1}\poly(\chi,\Lambda'_0,\Lambda'_1)e^{-\Omega((\sqrt{\Lambda'_1}-\sqrt{\Lambda'_0})/(\chi+1))}) \\ 
    &\quad+ \Or(\wt{\beta}_{\mathrm{comm}}(\wt{\Lambda},\Lambda'_1)\tau^{p+1}). 
\end{aligned}
\end{equation}
Combining this with \eqref{eq:trotter_error_total} and \eqref{eq:choose_Lambda'_0_ineq}, we bound the total error from the Trotter decomposition as
\begin{equation}
    \label{eq:trotter_error_total_bound}
    \begin{aligned}
        \|(e^{-iT\wt{H}}-S(\tau)^R)\Pi^{\mathrm{all}}_{[0,\Lambda_0]}\| &\leq R\sqrt{N_b} \poly(\chi T,\Lambda_0,\Lambda_0')e^{-\Omega((\sqrt{\Lambda'_0}-\sqrt{\Lambda_0})/(\chi T+1))} \\
        &\quad+ \Or(\wt{\alpha}_{\mathrm{comm}}(\wt{\Lambda})T^{p+1}R^{-p}\poly(\chi,\Lambda'_0,\Lambda'_1)e^{-\Omega((\sqrt{\Lambda'_1}-\sqrt{\Lambda'_0})/(\chi+1))}) \\ 
        &\quad+ \Or(\wt{\beta}_{\mathrm{comm}}(\wt{\Lambda},\Lambda'_1)T^{p+1}R^{-p}),
    \end{aligned}
\end{equation}
where we have used the relation $T=R\tau$.

We now choose $\Lambda'_0$, $\Lambda'_1$, $\wt{\Lambda}$, and $R$ so that the right-hand side of the above inequality is at most $\epsilon$, while simultaneously keeping the truncation error from Theorem~\ref{thm:truncate_ham} below $\epsilon$.

There is one other constraint in our choice of parameters: we need to ensure the canonical commutation relation of $X_{\alpha}$ and $P_{\alpha}$, when replaced by $\wt{X}_{\alpha}$ and $\wt{P}_{\alpha}$, holds exactly when evaluating $\wt{\beta}_{\mathrm{comm}}(\wt{\Lambda},\Lambda'_1)$. 
Note that in $\wt{\beta}_{\mathrm{comm}}(\wt{\Lambda},\Lambda'_1)$ there are at most $2(p+1)$ truncated position and momentum operators multiplied together because each Hamiltonian term is at most quadratic in these operators. Thus, if
\begin{equation}
\label{eq:M_M_1'_relation}
    \wt{\Lambda} \geq \Lambda'_1+2(p+1),
\end{equation}
then we can simply treat $\wt{X}_{\alpha}$ and $\wt{P}_{\alpha}$ as if they satisfy the exact canonical commutation relation when evaluating $\wt{\beta}_{\mathrm{comm}}(\wt{\Lambda},\Lambda'_1)$. We recall that $\wt{\Lambda}$ is the particle number truncation threshold for the Hamiltonian. To be more specific, this means
\begin{equation}
    [\wt{H}_{\gamma_{p+1}},\cdots[\wt{H}_{\gamma_2},\wt{H}_{\gamma_1}]\cdots]\Pi^{\mathrm{all}}_{[0,\Lambda'_1]} = [{H}_{\gamma_{p+1}},\cdots[{H}_{\gamma_2},{H}_{\gamma_1}]\cdots]\Pi^{\mathrm{all}}_{[0,\Lambda'_1]}
\end{equation}
when \eqref{eq:M_M_1'_relation} is satisfied.

With this extra constraint and \eqref{eq:trotter_error_total_bound}, we choose
\begin{equation}
    \label{eq:trotter_choice_of_parameters}
    \begin{aligned}
        \sqrt{\Lambda_0'} &= \sqrt{\Lambda_0} + \wt{\Or}(\chi T\polylog(\mathrm{coef},\epsilon^{-1})), \\
        \sqrt{\Lambda_1'} &= \sqrt{\Lambda_0} + \wt{\Or}(\chi T\polylog(\mathrm{coef},\epsilon^{-1})), \\
        \sqrt{\wt{\Lambda}}    &= \sqrt{\Lambda_0} + \wt{\Or}(\chi T\polylog(\mathrm{coef},\epsilon^{-1})), \\
        R   &= \wt{\Or}(T^{1+1/p}(\wt{\beta}_{\mathrm{comm}}(\wt{\Lambda},\Lambda_1'))^{1/p}\epsilon^{-1/p}),
    \end{aligned}
\end{equation}
where $\mathrm{coef}$ denotes all the coefficients $t,V,g,h,\omega$ in the Hamiltonian $H$, and $\mathcal{C}$ is defined in \eqref{eq:chi_boson_fermion}. This choice of $\wt{\Lambda}$ will also ensure that the Hamiltonian truncation error is upper bounded by $\epsilon$ by Theorem~\ref{thm:truncate_ham}. 
In choosing these parameters, we have omitted the scaling with $\wt{\alpha}_{\mathrm{comm}}(\wt{\Lambda})$.
This is because $\wt{\alpha}_{\mathrm{comm}}(\wt{\Lambda})$ is upper bounded by a polynomial of $\wt{\Lambda}$ and the Hamiltonian coefficients, and it gets absorbed into the poly-logarithmc factors.

\subsection{Bounding the nested commutators}
\label{sec:bounding_the_nested_commutators}

It now remains to bound $\wt{\beta}_{\mathrm{comm}}(\wt{\Lambda},\Lambda_1')$. 
Suppose we are given a series of indices of Hamiltonian terms $\gamma_1,\gamma_2,\ldots$. We will show that 
\begin{equation}
\label{eq:nested_commutator_bound}
    \|[\wt{H}_{\gamma_{p+1}},\cdots[\wt{H}_{\gamma_2},\wt{H}_{\gamma_1}]\cdots]\Pi^{\mathrm{all}}_{[0,\Lambda'_1]}\| \leq A_{\gamma_{p+1}}^{(p)}A_{\gamma_{p-1}}^{(p-2)}\cdots A_{\gamma_2}^{(1)}B_{\gamma_1},
\end{equation}
where
\begin{equation}
    \begin{aligned}
    &A_{1}^{(q)} = 2q\max_i\sum_j |t_{ij}|,\quad B_1 = \sum_{ij}|t_{ij}|, \quad
    A_{2}^{(q)} = 4q\max_i\sum_j |V_{ij}|,\quad B_2 = \sum_{ij}|V_{ij}|, \\
    &A_3^{(q)} = 2q\max_{j}\sum_{\alpha i}|g^{(\alpha)}_{ij}|\sqrt{2(\Lambda'_1+1)}+q\max_{\alpha}\sum_{ ij}|g^{(\alpha)}_{ij}|(2(\Lambda'_1+1))^{-1/2}, \\
    &A_4^{(q)} = 2q\max_{j}\sum_{\alpha i}|h_{ij}^{(\alpha)}|\sqrt{2(\Lambda'_1+1)}+q\max_{\alpha}\sum_{ ij}|h_{ij}^{(\alpha)}|(2(\Lambda'_1+1))^{-1/2}, \\
    &B_3 = \sum_{\alpha ij} |g_{ij}^{(\alpha)}|\sqrt{2(\Lambda'_1+1)},\quad B_4 = \sum_{\alpha ij} |h_{ij}^{(\alpha)}|\sqrt{2(\Lambda'_1+1)}, \\
    &A_5^{(q)}=A_6^{(q)} = q\max_{\alpha}|\omega_{\alpha}|,\quad B_5=B_6=\sum_{\alpha}|\omega_{\alpha}|(\Lambda'_1+1).
    \end{aligned}
\end{equation}
At a high level, $A_{\gamma}^{(q)}$'s quantify the growth of the nested commutator when the nesting layer increases by one, while $B_{\gamma}$'s are chosen to handle the base case when there is only one operator.

Once we have established \eqref{eq:nested_commutator_bound}, we define
\begin{equation}
\label{eq:quantities_AB_trotter}
    A = A_1^{(p)}+\cdots+A_6^{(p)},\quad B=B_1+\cdots+B_6,
\end{equation}
then
\begin{equation}
    \wt{\beta}_{\mathrm{comm}}(\wt{\Lambda},\Lambda_1') =\Or(A^p B).
\end{equation}
This implies that we need 
\begin{equation}
\label{eq:number_trotter_steps}
    R = \wt{\Or}\left(AB^{1/p}T^{1+1/p}\epsilon^{-1/p}\right)
\end{equation}
Trotter steps by \eqref{eq:trotter_choice_of_parameters}. In the above analysis we treat the order $p$ as a constant. The gate complexity depends on how we implement each $e^{-it\wt{H}_{\gamma}}$. 
For concreteness, we analyze the gate complexity of simulating the Hubbard-Holstein model in the next section, although the approach may be extended to simulate other quantum systems within our framework.

We now derive the bound \eqref{eq:nested_commutator_bound} for an arbitrary nested commutator. 
We first note that a nested commutator multiplied to a projection operator  $[\wt{H}_{\gamma_q},\cdots[\wt{H}_{\gamma_2},\wt{H}_{\gamma_1}]\cdots]\Pi^{\mathrm{all}}_{[0,\Lambda'_1]}$ can be written as a linear combination of products of at most $q$ fermionic operators $c_i^{\dagger}c_j$, and at most $q$ projected bosonic position or momentum operators, multiplied to the projection operator at the end. This can be proved inductively. 
We introduce some notations to formalize this observation.
For convenience we denote
\begin{equation}
    Q^{\varsigma}_\alpha = 
    \begin{cases}
    \wt{X}_{\alpha},\ \varsigma=0,\\
    \wt{P}_{\alpha},\ \varsigma=1.
    \end{cases}
\end{equation}
We first define a set of index strings:
\begin{equation}
\label{eq:allowd_indices}
    \Xi_q = \{(\vec{i},\vec{j},\vec{\alpha},\vec{\varsigma}):|\vec{i}|=|\vec{j}|\leq q,|\vec{\alpha}|=|\vec{\varsigma}|\leq q\},
\end{equation}
where $\vec{i}$ and $\vec{j}$ are strings of fermionic mode indices, $\vec{\alpha}$ is a string of bosonic mode indices, and $\vec{\varsigma}$ is a string of $0$'s and $1$'s.
Then the claimed expansion is formally given by
\begin{equation}
    [\wt{H}_{\gamma_q},\cdots[\wt{H}_{\gamma_2},\wt{H}_{\gamma_1}]\cdots]\Pi^{\mathrm{all}}_{[0,\Lambda'_1]} = \sum_{(\vec{i},\vec{j},\vec{\alpha},\vec{\varsigma})\in\Xi_q} S_{(\vec{i},\vec{j},\vec{\alpha},\vec{\varsigma})}^{(q)} \prod_{k=1}^{|\vec{i}|}c_{i_k}^{\dagger}c_{j_k}\prod_{k=1}^{|\vec{\alpha}|}Q_{\alpha_k}^{\varsigma_k}\Pi^{\mathrm{all}}_{[0,\Lambda'_1]},
\end{equation}
for some coefficients $S_{(\vec{i},\vec{j},\vec{\alpha},\vec{\varsigma})}^{(q)}$.
This can be readily proved by induction on $q$.

Now we define
\begin{equation}
    \mathcal{S}_q = \sum_{(\vec{i},\vec{j},\vec{\alpha},\vec{\varsigma})\in\Xi_q} |S_{(\vec{i},\vec{j},\vec{\alpha},\vec{\varsigma})}^{(q)}|(2(\Lambda'_1+1))^{|\vec{\alpha}|/2}.
\end{equation}
Then by the triangle inequality and the fact that $\|\wt{X}_{\alpha}\|,\|\wt{P}_{\alpha}\|\leq \sqrt{2(\Lambda'_1+1)}$, we have 
\begin{equation}
    \|[\wt{H}_{\gamma_q},\cdots[\wt{H}_{\gamma_2},\wt{H}_{\gamma_1}]\cdots]\Pi^{\mathrm{all}}_{[0,\Lambda'_1]}\|\leq \mathcal{S}_q.
\end{equation}
Therefore we only need to show
\begin{equation}
\label{eq:one_norm_coef_commutators}
    \mathcal{S}_q \leq A_{\gamma_q}^{q-1}\cdots A_{\gamma_2}^{1} B_{\gamma_1}.
\end{equation}
This is done inductively by examining the commutator $\left[\wt{H}_{\gamma},\prod_{k=1}^{|\vec{i}|}c_{i_k}^{\dagger}c_{j_k}\prod_{k=1}^{|\vec{\alpha}|}Q_{\alpha_k}^{\varsigma_k}\right]\Pi^{\mathrm{all}}_{[0,\Lambda'_1]}$ for each $\gamma$, which gives an upper bound for $\mathcal{S}_{q+1}$ that depends on $\mathcal{S}_q$.
Combined with the fact that $\|\wt{H}_{\gamma}\|\leq B_{\gamma}$ we will establish \eqref{eq:one_norm_coef_commutators}.

For simplicity, we only provide the proof for the induction step when $\gamma_{q+1}=3$ corresponds to one of the boson-fermion coupling terms, which together with $\gamma_{q+1}=4$, are the most complicated ones to analyze among all the possible choices of $\gamma_{q+1}$. We have
\begin{equation}
\label{eq:commutator_H2}
    \begin{aligned}
    \left[\wt{H}_3,\prod_{k=1}^{|\vec{i}|}c_{i_k}^{\dagger}c_{j_k}\prod_{k=1}^{|\vec{\alpha}|}Q_{\alpha_k}^{\varsigma_k}\right]\Pi^{\mathrm{all}}_{[0,\Lambda'_1]}
    &= \sum_{\alpha ij} \left[g_{ij}^{(\alpha)}c_i^{\dagger}c_j X_{\alpha},\prod_{k=1}^{|\vec{i}|}c_{i_k}^{\dagger}c_{j_k}\prod_{k=1}^{|\vec{\alpha}|}Q_{\alpha_k}^{\varsigma_k}\right]\Pi^{\mathrm{all}}_{[0,\Lambda'_1]} \\
    &= \sum_{\alpha ij} g_{ij}^{(\alpha)}\left[c_i^{\dagger}c_j,\prod_{k=1}^{|\vec{i}|}c_{i_k}^{\dagger}c_{j_k}\right]X_{\alpha}\prod_{k=1}^{|\vec{\alpha}|}Q_{\alpha_k}^{\varsigma_k}\Pi^{\mathrm{all}}_{[0,\Lambda'_1]} \\
    &\quad + \sum_{\alpha ij} g_{ij}^{(\alpha)}\left(\prod_{k=1}^{|\vec{i}|}c_{i_k}^{\dagger}c_{j_k}\right)c_i^{\dagger}c_j\left[X_{\alpha},\prod_{k=1}^{|\vec{\alpha}|}Q_{\alpha_k}^{\varsigma_k}\right]\Pi^{\mathrm{all}}_{[0,\Lambda'_1]},
    \end{aligned}
\end{equation}
where we have used the identity that for any operators $O_1,O_2,O_3,O_4$, 
\begin{equation}
    [O_1\otimes O_2,O_3\otimes O_4] = [O_1,O_3]\otimes O_2O_4 + O_1 O_3\otimes [O_2,O_4].
\end{equation}
We then apply the commutation rule
\begin{equation}
    [c_i^{\dagger}c_j,c_k^{\dagger}c_l] = c_i^{\dagger}c_l\delta_{jk} - c_k^{\dagger}c_j\delta_{il},
\end{equation}
so the second line of \eqref{eq:commutator_H2} becomes
\begin{equation}
    \sum_{k'=1}^{|\vec{i}|}c_{i_{|\vec{i}|}}^{\dagger}c_{j_{|\vec{i}|}}\cdots\left(\sum_{\alpha i} g^{(\alpha)}_{ii_{k'}}c_i^{\dagger}c_{j_{k'}}-\sum_{\alpha j}g^{(\alpha)}_{j_{k'}j}c_{i_{k'}}^{\dagger}c_j\right)\cdots c_{i_1}^{\dagger}c_{j_1}X_{\alpha}\prod_{k=1}^{|\vec{\alpha}|}Q_{\alpha_k}^{\varsigma_k}\Pi^{\mathrm{all}}_{[0,\Lambda'_1]}.
\end{equation}
In the above sum of products of fermionic and bosonic operators, the number of bosonic operators in the product is increased by $1$, and the absolute value of the coefficients $g^{(\alpha)}_{ii_{k'}}$ and $g^{(\alpha)}_{j_{k'}j}$ add up to at most $2q\max_{j}\sum_{\alpha i}|g^{(\alpha)}_{ij}|$. Therefore the contribution to $\mathcal{S}_{q+1}$ is at most $2q\max_{j}\sum_{\alpha i}|g^{(\alpha)}_{ij}|\sqrt{2(\Lambda'_1+1)}\mathcal{S}_{q}$.

For the third line in \eqref{eq:commutator_H2}, we have
\begin{equation}
\begin{aligned}
    &\sum_{\alpha ij} g_{ij}^{(\alpha)}\left(\prod_{k=1}^{|\vec{i}|}c_{i_k}^{\dagger}c_{j_k}\right)c_i^{\dagger}c_j\sum_{k'=1}^{|\vec{\alpha}|}\delta_{\alpha,\alpha_{k'}}\delta_{\varsigma_{k'},1}Q^{\varsigma_{|\vec{\alpha}|}}_{\alpha_{|\vec{\alpha}|}}\cdots Q^{\varsigma_{k'+1}}_{\alpha_{k'+1}}Q^{\varsigma_{k'-1}}_{\alpha_{k'-1}}\cdots Q^{\varsigma_{1}}_{\alpha_{1}}\Pi^{\mathrm{all}}_{[0,\Lambda'_1]} \\
    &= \sum_{ij} \left(\prod_{k=1}^{|\vec{i}|}c_{i_k}^{\dagger}c_{j_k}\right)c_i^{\dagger}c_j\sum_{k'=1}^{|\vec{\alpha}|}g_{ij}^{\alpha_{k'}}\delta_{\varsigma_{k'},1}Q^{\varsigma_{|\vec{\alpha}|}}_{\alpha_{|\vec{\alpha}|}}\cdots Q^{\varsigma_{k'+1}}_{\alpha_{k'+1}}Q^{\varsigma_{k'-1}}_{\alpha_{k'-1}}\cdots Q^{\varsigma_{1}}_{\alpha_{1}}\Pi^{\mathrm{all}}_{[0,\Lambda'_1]},
\end{aligned}
\end{equation}
where we have used the canonical commutation relation between $X_{\alpha}$ and $P_{\alpha}$ on $\wt{X}_{\alpha}$ and $\wt{P}_{\alpha}$. This is justified because the nested commutator is multiplied to the projection operator $\Pi^{\mathrm{all}}_{[0,\Lambda'_1]}$ and we have imposed the constraint \eqref{eq:M_M_1'_relation}.
The sum of the absolute value of the coefficients is at most $q\max_{\alpha}\sum_{ij}|g^{(\alpha)}_{ij}|$ and the the number of bosonic operators in the product is reduced by $1$. Therefore the contribution to $S_{q+1}$ is at most $q\max_{\alpha}\sum_{ij}|g^{(\alpha)}_{ij}|(2(\Lambda'_1+1))^{-1/2}S_q$.

Combining our analysis for the second and third lines of \eqref{eq:commutator_H2} we have
\begin{equation}
    S_{q+1}\leq 2q\max_{j}\sum_{\alpha i}|g^{(\alpha)}_{ij}|\sqrt{2(\Lambda'_1+1)}\mathcal{S}_{q} + q\max_{\alpha}\sum_{ij}|g^{(\alpha)}_{ij}|(2(\Lambda'_1+1))^{-1/2}S_q = A_2^{(q-1)}\mathcal{S}_{q},
\end{equation}
if $\gamma_{q+1}=3$. The commutators with the other $\wt{H}_{\gamma}$'s can be analyzed in a similar way.
The proof of \eqref{eq:nested_commutator_bound} is now completed.

\subsection{Simulating the Hubbard-Holstein model with Trotterization}
\label{sec:hubbard_holstein_trotter}

We recall the definition of the Hubbard Holstein model given in Section~\ref{sec:sim_boson_fermion}: The Hamiltonian is
\begin{equation}
    H = H_f + H_{fb} + H_b,
\end{equation}
where $H_{f}$ is the Hamiltonian of the Fermi-Hubbard model:
\begin{equation}
    H_{f} = -\sum_{\braket{x,x'},\sigma}(c_{x,\sigma}^{\dagger}c_{x',\sigma}+\mathrm{h.c.})+ U\sum_{x=1}^N (n_{x,\uparrow}-\frac{1}{2})(n_{x,\downarrow}-\frac{1}{2}) - \mu\sum_{r=1}^N n_x,
\end{equation}
and
\begin{equation}
    H_{fb} =  g \sum_{x=1}^N (b_x^{\dagger}+b_x)(n_{x,\uparrow}+n_{x,\downarrow}-1) \quad 
    H_b = \omega_0 \sum_{x=1}^N b_{x}^{\dagger}b_x,
\end{equation}
are the boson-fermion coupling part and bosonic part respectively. The lattice sites are indexed by $x$ and $x'$, and spins are indexed by $\sigma$.
As in Section~\ref{sec:sim_boson_fermion}, we assume for simplicity that all model parameters except for the system size $N$, i.e. $g$, $\omega_0$, $U$, $\mu$, are all constants. We consider the case where the time evolution starts with an initial state that has at most $\Lambda_0$ bosonic particles at each site.

We note that this Hamiltonian satisfies the general form of boson-fermion coupling Hamiltonians given in \eqref{eq:ham_fermion_boson_1}. 
Therefore we can directly apply our above analysis to analyze the number of required Trotter steps. 
First we note that all the quantities involved in $A$ given in \eqref{eq:quantities_AB_trotter}, i.e.
\[
\|t\|_1,\ \|V\|_1,\ \max_j\sum_{\alpha i}|g^{(\alpha)}_{ij}|,\ \max_j\sum_{\alpha i}|h_{ij}^{(\alpha)}|,\ \max_{\alpha}\sum_{ij}|g^{(\alpha)}_{ij}|,\ \max_{\alpha}\sum_{ij}|h_{ij}^{(\alpha)}|,\ \max_{\alpha}|\omega_{\alpha}|,
\]
are upper bounded by some constants. This follows from the sparsity of the coefficient matrices $t$, $V$, $g^{(\alpha)}$, $h^{(\alpha)}$ (in fact $h^{(\alpha)}=0$ in this model). Similarly, all the quantities involved in $B$ given in \eqref{eq:quantities_AB_trotter}, i.e.
\[
\sum_{ij}|t_{ij}|,\ \sum_{ij}|V_{ij}|,\ \sum_{\alpha ij} |g_{ij}^{(\alpha)}|,\ \sum_{\alpha ij} |h_{ij}^{(\alpha)}|,\ \sum_{\alpha}|\omega_{\alpha}|,
\]
are all $\Or(N)$. Therefore we have
\begin{equation}
    A = \Or(\sqrt{\Lambda'_1}),\quad B=\Or(N\Lambda'_1).
\end{equation}
Then by \eqref{eq:number_trotter_steps} the number of Trotter steps required to simulate the Hubbard-Holstein model is
\begin{equation}
    R = \Or\left(\sqrt{\Lambda'_1}(N\Lambda'_1)^{1/p}T^{1+1/p}\epsilon^{-1/p}\right).
\end{equation}
Note that $\Lambda'_1$ has the asymptotic scaling described in \eqref{eq:trotter_choice_of_parameters}.
Taking into account the fact that for the Hubbard-Holstein model $\Tr|g^{(\alpha)}|,\Tr|h^{(\alpha)}|=\Or(1)$, which implies further that $\chi=\Or(1)$, we have
\begin{equation}
    \Lambda'_1 = \left(\sqrt{\Lambda_0}+\wt{\Or}(T\polylog(N\epsilon^{-1}))\right)^2,
\end{equation}
which gives
\begin{equation}
    R = \wt{\Or}\left(N^{1/p}(\sqrt{\Lambda_0}+T)^{1+2/p}T^{1+1/p}\epsilon^{-1/p}\right).
\end{equation}
Each Trotter step can be implemented with $\wt{\Or}(N\polylog(\Lambda_0 T\epsilon^{-1}))$ gates, and therefore the total gate complexity is
\begin{equation}
    \wt{\Or}\left(N^{1+1/p}(\sqrt{\Lambda_0}+T)^{1+2/p}T^{1+1/p}\epsilon^{-1/p}\right).
\end{equation}
For large $p$, this almost matches the gate complexity derived in Section~\ref{sec:sim_boson_fermion} based on the HHKL decomposition.

\section{A gate complexity lower bound for simulating bosons}
\label{sec:lower_bound}
 
In Sections~\ref{sec:sim_boson_fermion} and \ref{sec:boson_fermion_trotter}, we have discussed the gate complexity of simulating the Hubbard-Holstein model. 
One distinctive feature is that the scaling with respect to time is almost quadratic, instead of being almost linear for simulating bounded Hamiltonians. 
In this section we construct a class of Hamiltonians acting on a single bosonic mode and a register of qubits, for which performing simulation up to time $T$ will require at least $\wt{\Omega}(T^2)$ gates. This shows that simulation involving bosons cannot in general be expected to have a linear dependence on time. Note that here by simulation we mean simulating only the qubit part of the boson-qubit coupled system, and as a result we only need to deal with a finite-dimensional Hilbert space.

Specifically we consider the time evolution of a bosonic mode coupled to a register containing $N$ qubits. 
We will label the bosonic mode by a subscript $\beta$ and the qubit register by a subscript $q$. 
A product state is written as $\ket{\psi}_{\beta}\ket{\phi}_q$, where $\ket{\psi}_{\beta}$ is the state of the bosonic mode, and $\ket{\phi}_q$ is the state of the qubit register. We call this qubit register the \emph{$q$-register} because later we need an additional qubit register. When simulating the time evolution of this system, we consider a unitary circuit $W$ acting jointly on an ancilla register, which we label as $\mathrm{anc}$, and the $q$-register. We will also denote by $\ket{\lambda}_{\beta}$ the $\varepsilon$-particle state of the bosonic mode.
\begin{thm}
\label{thm:lower_bound}
    For any integers $N$ and $T$ such that $1\leq \sqrt{N}\leq T\leq 2^{N/2}$,
    there exists a boson-qubit coupled Hamiltonian $H = Ub + b^{\dagger} U^{\dagger}$,
    where $b$ and $b^{\dagger}$ are the bosonic annihilation and creation operators respectively,
    and $U$ is a unitary acting on the bosonic mode and $N$ qubits (the $q$-register) that preserves the bosonic number.
    If a quantum circuit $W$ satisfies
    \begin{equation}
    \label{eq:approx_expectation}
    \big|\bra{0}_{\beta}\bra{\phi}_q e^{iTH}(I_{\beta}\otimes O) e^{-iTH}\ket{0}_{\beta}\ket{\phi}_q-\bra{0}_{\mathrm{anc}}\bra{\phi}_q W^{\dagger}(I_{\mathrm{anc}}\otimes O)W\ket{0}_{\mathrm{anc}}\ket{\phi}_q\big|\leq 0.1
    \end{equation}
    for all $\ket{\phi}_q$,
    then $W$ must use at least $\wt{\Omega}(NT^2)$ 2-qubit gates.
    Here, $I_{\beta}$ and $I_{\mathrm{anc}}$ are the identity operator on the bosonic mode and the ancilla register respectively, and $O=\ket{0}\bra{0}\otimes I$ is the projection onto the $\ket{0}$ state of the first qubit of the $q$-register.
    
    The quantum circuit $W$ may use an arbitrarily large number of ancilla qubits, and gates in $W$ may be non-local and come from a possibly infinite gate set.
\end{thm}

In essence, this theorem asserts the existence of boson-qubit coupled systems whose single-qubit measurement statistics after evolving for time $T$ require $\wt{\Omega}(NT^2)$ gates to approximate to constant precision.

To prove Theorem~\ref{thm:lower_bound} we need to use the following lemma:
\begin{lem}
\label{lem:generalized_position}
Let $H=Ub + b^{\dagger} U^{\dagger}$, 
where $U = \sum_{\lambda=0}^{\infty}\ket{\lambda}_{\beta}\bra{\lambda}_{\beta}\otimes U_{\lambda}$.
Then
\[
e^{-itH}\ket{0}_{\beta}\ket{\phi} = e^{-t^2/2}\sum_{\lambda=0}^{\infty} \frac{(-it)^{\lambda}}{\sqrt{\lambda !}}\ket{\lambda}_{\beta}\left(\prod_{k=0}^{\lambda-1} U^{\dagger}_k \ket{\phi}\right)
\]
\end{lem}

\begin{proof}
Denoting $\wt{b}=Ub$, we have
\begin{equation}
    [\wt{b},\wt{b}^{\dagger}] = [b,b^{\dagger}]=1.
\end{equation}
Therefore $\wt{b}$ and $\wt{b}^{\dagger}$ can be treated as a new pair of annihilation and creation operators. By the Kermac-McCrae identity we have
\begin{equation}
    e^{-itH} = e^{-\frac{1}{2}t^2} e^{-it\wt{b}^{\dagger}}e^{-it\wt{b}}.
\end{equation}
This can then be used to prove the lemma by using the Taylor expansion and the fact that
\begin{equation}
    e^{-it\wt{b}}\ket{0}_{\beta}\ket{\phi} = \ket{0}_{\beta}\ket{\phi}.
\end{equation}
\end{proof}

\begin{proof}[Proof of Theorem~\ref{thm:lower_bound}]
First we consider a quantum circuit $U_{\mathrm{circ}}$ that acts on $N$ qubits and has depth $T^2$. It can then be written as
\begin{equation}
    U_{\mathrm{circ}} = \prod_{m=0}^{T^2-1} U^{\dagger}_{\lambda},
\end{equation}
where each $U_\lambda$ acts on $N$ qubits and has depth one. We also define $U_{\lambda} = I$ for all $\lambda\geq T^2$. Then we let the unitary $U$ in the theorem be 
\begin{equation}
    U=\sum_{\lambda=0}^{\infty} \ket{\lambda}_{\beta}\bra{\lambda}_{\beta}\otimes U_{\lambda}.
\end{equation}
Note that by construction we have $[U,\ket{\lambda}_{\beta}\bra{\lambda}_{\beta}]=0$ and therefore $U$ preserves the particle number in the bosonic mode.

We will show that by running time evolution $e^{-iT'H}$ for $T'=\Theta(T)$ starting from $\ket{0}_{\beta}\ket{\phi}_q$, and performing measurement on the first qubit in the $q$-register, we will be able to approximately sample from the distribution generated by running $U_{\mathrm{circ}}$ and then measuring the first qubit (note that $U_{\mathrm{circ}}$ acts only on register $q$). In this procedure we trace out the bosonic mode and focus only on the qubits.

By Lemma~\ref{lem:generalized_position}, we have
\begin{equation}
\label{eq:from_lem_generalized_position}
    e^{-iT'H}\ket{0}_{\beta}\ket{\phi}_q = e^{-T'^2/2}\sum_{\lambda=0}^{\infty} \frac{(-iT')^{\lambda}}{\sqrt{\lambda !}}\ket{\lambda}_{\beta}\left(\prod_{k=0}^{\lambda-1} U^{\dagger}_k \ket{\phi}_q\right).
\end{equation}
Now note that for any summand on the right-hand side with $\lambda\geq T^2$, we have
\begin{equation}
    \ket{\lambda}_{\beta}\left(\prod_{k=0}^{\lambda-1} U^{\dagger}_k \ket{\phi}_q\right) = \ket{\lambda}_{\beta} U_{\mathrm{circ}} \ket{\phi}_q.
\end{equation}
As a result we can write
\begin{equation}
    e^{-iT'H}\ket{0}_{\beta}\ket{\phi}_q = \ket{\Psi^{\perp}}_{\beta q} + \mathcal{A}\ket{\psi}_{\beta} U_{\mathrm{circ}} \ket{\phi}_q,
\end{equation}
where $\ket{\Psi^{\perp}}_{\beta q}$ is the sum of the first $T^2$ terms on the right-hand side of \eqref{eq:from_lem_generalized_position}, and 
\begin{equation}
    \ket{\psi}_{\beta} = \frac{e^{-T'^2/2}}{\mathcal{A}}\sum_{\lambda=T^2}^{\infty} \frac{(-iT')^{\lambda}}{\sqrt{\lambda !}}\ket{\lambda}_{\beta},\quad \mathcal{A} = \left(e^{-T'^2}\sum_{\lambda=T^2}^{\infty}\frac{T'^{2\lambda}}{\lambda !}\right)^{1/2}.
\end{equation}
Note that the normalization factor $\mathcal{A}$ is chosen so that $\|\ket{\psi}_{\beta}\|=1$.
In the above quantum state $e^{-iT'H}\ket{0}_{\beta}\ket{\phi}_q$, the bosonic particle number satisfies the Poisson distribution with mean $T'^2$. Because the Poisson distribution decays rapidly away from the mean \cite{Canonne2017PoissonTail}, we can choose $T'=\Theta(T)$ so that 
\begin{equation}
    \|\ket{\Psi^{\perp}}_{\beta q}\|^2 = e^{-T'^2}\sum_{\lambda=0}^{T^2-1}\frac{T'^{2\lambda}}{\lambda !}\leq 0.0025,
\end{equation}
and consequently
\begin{equation}
\begin{aligned}
    &\|e^{-iT'H}\ket{0}_{\beta}\ket{\phi}_q - \ket{\psi}_{\beta} U_{\mathrm{circ}}\ket{\phi}_q\| \\
    &\leq \|e^{-iT'H}\ket{0}_{\beta}\ket{\phi}_q - \mathcal{A}\ket{\psi}_{\beta} U_{\mathrm{circ}}\ket{\phi}_q\| + (1-\mathcal{A}) \\
    &\leq  2 \|\ket{\Psi^{\perp}}_{\beta q}\| \leq  0.1,
\end{aligned}
\end{equation}
where in going from the second line to the third line we have used the fact that $\|\ket{\Psi^{\perp}}_{\beta q}\|+\mathcal{A}\leq 1$.
Therefore
\begin{equation}
    \big|\bra{0}_{\beta}\bra{\phi}_qe^{iT'H}(I_{\beta}\otimes O)e^{-iT'H}\ket{0}_{\beta}\ket{\phi}_q - \bra{\phi}_q U_{\mathrm{circ}}^{\dagger} O U_{\mathrm{circ}}\ket{\phi}_q\big|\leq 0.2,
\end{equation}
where $O=\ket{0}\bra{0}\otimes I$. If a circuit $W$ as described in the theorem satisfies the inequality \eqref{eq:approx_expectation}, then by the triangle inequality
\begin{equation}
\label{eq:approx_sampling}
    \big|\bra{0}_{\mathrm{anc}}\bra{\phi}_q W^{\dagger}(I_{\mathrm{anc}}\otimes O)W\ket{0}_{\mathrm{anc}}\ket{\phi}_q - \bra{\phi}_q U_{\mathrm{circ}}^{\dagger} O U_{\mathrm{circ}}\ket{\phi}_q\big|\leq 0.3.
\end{equation}
This means the measurement outcome generated by running the circuit $U_{\mathrm{circ}}$ can be simulated by running the circuit $W$.

With the above setup, we then use $U_{\mathrm{circ}}$ to compute Boolean functions in the sense defined in \cite{HaahHastingsKothariLow2021quantum}: for a Boolean function $f:\{0,1\}^N\to\{0,1\}$, we say $U$ computes the Boolean function \emph{with high probability} if measuring the first qubit of $U\ket{x_1x_2\cdots x_N0\cdots 0}$ yields $f(x)$ with probability at least $2/3$. We also say 
$U$ computes the Boolean function \emph{exactly} if measuring the first qubit of $U\ket{x_1x_2\cdots x_N0\cdots 0}$ yields $f(x)$ with probability $1$.

By \eqref{eq:approx_sampling}, we know that if $U_{\mathrm{circ}}$ computes a Boolean function $f$ exactly, then $W$ computes the same Boolean function with high probability. By \cite[Lemma 8]{HaahHastingsKothariLow2021quantum}, we can choose $U_{\mathrm{circ}}$ acting on $N$ qubits and with depth $T^2$ to compute $2^{\wt{\Omega}(T^2 N)}$ different Boolean functions exactly. If $W$ uses $G$ 2-qubit gates, then by \cite[Lemma 8]{HaahHastingsKothariLow2021quantum} different $W$ can compute at most $2^{\wt{\Or}(G\log(N))}$ different Boolean functions with high probability. Therefore $G=\wt{\Omega}(T^2N)$, which completes the proof.
\end{proof}

\section{Quantum number distribution tail bound in eigenstates}
\label{sec:tail}

If we would like to prepare an eigenstate of a Hamiltonian of the form \eqref{eq:ham_general} on a quantum computer, then we need to be able to store this eigenstate using a finite number of qubits. 
This reaffirms the need to truncate infinite dimensional Hilbert spaces.
A natural approach is to truncate the local quantum number $\lambda$, which, as discussed in Section~\ref{sec:common_structures}, is the local bosonic particle number in the setting of boson-fermion coupling, the electric field value in the setting of U(1) lattice gauge theory, and the total angular momentum in the setting of SU(2) lattice gauge theory. 

In this section, we will show that the probability of a spectrally isolated eigenstate having a local quantum number beyond a certain threshold can be bounded, and we call this the \emph{tail bound}. This tail bound justifies cutting off the part of the Hilbert space with local quantum number beyond the threshold, thus enabling us to store eigenstates using a finite number of qubits. We describe the result in the following theorem:

\begin{thm}[Quantum number distribution tail bound]
    Let $H=H_W+H_R$ be a Hamiltonian satisfying \eqref{eq:conditions_ham_general} with parameters $\chi$ and $r$.
    Assume that $\ket{\Psi}$ is an eigenstate of $H$ corresponding to an eigenvalue $\varepsilon$ with multiplicity $1$, and that $\varepsilon$ is separated from the rest of the spectrum of $H$ by a spectral gap $\delta$. Moreover, we assume the absolute value of the quantum number distribution has a finite expectation
    \[
    \sum_{\lambda}|\lambda|\braket{\Psi|\Pi_{\lambda}|\Psi}=\bar{\lambda}<\infty.
    \]
    Then for any $\epsilon>0$, we can choose $\Lambda$ satisfying
    \[
    \Lambda^{1-r} = (2\bar{\lambda})^{1-r} + \Or(\chi \delta^{-1}\log^2(\epsilon^{-1})+\log(\epsilon^{-1})),
    \]
    such that
    $
    \|\overline{\Pi}_{[-\Lambda,\Lambda]}\ket{\Psi}\|\leq \epsilon.
    $
\end{thm}
\begin{proof}
We define the projection operator into the $\varepsilon$-eigensubspace by
\begin{equation}
    P_{\varepsilon} = \ket{\Psi}\bra{\Psi}.
\end{equation}
This projection operator, and its approximator to be introduced later, will be the main technical tool in this proof. We first apply a projection operator to truncate the eigenstate $\ket{\Psi}$:
\begin{equation}
    \alpha\ket{\zeta} = \Pi_{[-2\bar{\lambda},2\bar{\lambda}]}\ket{\Psi},
\end{equation}
where $\ket{\zeta}$ is a normalized quantum state and $\alpha= \| \Pi_{[-2\bar{\lambda},2\bar{\lambda}]}\ket{\Psi}\|>0$. Because of the assumption $\sum_{\lambda}\lambda\braket{\Psi|\Pi_{\lambda}|\Psi}=\bar{\lambda}$, we have by Markov's inequality that
\begin{equation}
    \braket{\Psi|\Pi_{[-2\bar{\lambda},2\bar{\lambda}]}|\Psi} \geq 1/2,
\end{equation}
and therefore 
\begin{equation}
    \alpha = \sqrt{\braket{\Psi|\Pi_{[-2\bar{\lambda},2\bar{\lambda}]}|\Psi}} \geq 1/\sqrt{2}.
\end{equation}
This further implies
\begin{equation}
    |\braket{\Psi|\zeta}| = \braket{\Psi|\Pi_{[-2\bar{\lambda},2\bar{\lambda}]}|\Psi}/\alpha=\alpha\geq 1/\sqrt{2}.
\end{equation}

We then apply an approximation of $P_{\varepsilon}$ to $\ket{\zeta}$. Note that $P_{\varepsilon}\ket{\zeta}$ is exactly the eigenstate $\ket{\Psi}$ up to a constant factor. Therefore applying an approximation of $P_{\varepsilon}$ will yield a quantum state that is close to the eigenstate. The approximation of $P_{\varepsilon}$ is constructed as
\begin{equation}
    \wt{P}_{\varepsilon} = \frac{\sigma}{\sqrt{2\pi}} \int_{-T}^T \dd t\ e^{-\frac{1}{2}\sigma^2 t^2} e^{-i\varepsilon t}e^{itH}.
\end{equation}
We will show that $\wt{P}_{\varepsilon}$ is close to $P_{\varepsilon}$ when $\sigma$ is small and $T$ is large. 
First we have
\begin{equation}
\label{eq:decompose_approx_error_eigenprojection}
\begin{aligned}
    P_{\varepsilon} - \wt{P}_{\varepsilon} &= \left(P_{\varepsilon} - e^{-\frac{(H-\varepsilon)^2}{2\sigma^2}}\right) + \left(e^{-\frac{(H-\varepsilon)^2}{2\sigma^2}} - \wt{P}_{\varepsilon}\right) \\
    &= \left(P_{\varepsilon} - e^{-\frac{(H-\varepsilon)^2}{2\sigma^2}}\right) + \frac{\sigma}{\sqrt{2\pi}} \int_{|t|\geq T} \dd t\ e^{-\frac{1}{2}\sigma^2 t^2} e^{-i\varepsilon t}e^{itH},
\end{aligned}
\end{equation}
where we have used the identity
\begin{equation}
    e^{-\frac{(H-\varepsilon)^2}{2\sigma^2}} = \frac{\sigma}{\sqrt{2\pi}} \int_{-\infty}^{\infty} \dd t\ e^{-\frac{1}{2}\sigma^2 t^2} e^{-i\varepsilon t}e^{itH}.
\end{equation}
For the first term on the second line of \eqref{eq:decompose_approx_error_eigenprojection}, we have
\begin{equation}
    \|P_{\varepsilon} - e^{-\frac{(H-\varepsilon)^2}{2\sigma^2}}\|\leq e^{-\frac{\delta^2}{2\sigma^2}}, 
\end{equation}
and for the second term we have
\begin{equation}
    \Big\|\frac{\sigma}{\sqrt{2\pi}} \int_{|t|\geq T} \dd t\ e^{-\frac{1}{2}\sigma^2 t^2} e^{-i\varepsilon t}e^{itH}\Big\| \leq \frac{\sigma}{\sqrt{2\pi}}\int_{|t|\geq T} \dd t\ e^{-\frac{1}{2}\sigma^2 t^2} \leq \sqrt{\frac{2}{\pi}}e^{-\frac{\sigma^2 T^2}{2}},
\end{equation}
where we have used \cite[Theorem~1]{Chang2011Chernoff} for the second inequality. Denoting the sum of these two bounds by $\epsilon_1$, we have
\begin{equation}
    \|P_{\varepsilon} - \wt{P}_{\varepsilon}\|\leq e^{-\frac{\delta^2}{2\sigma^2}} + \sqrt{\frac{2}{\pi}}e^{-\frac{\sigma^2 T^2}{2}} = \epsilon_1.
\end{equation}
We choose $\sigma$ and $T$ so that $\epsilon_1\leq 1/2\sqrt{2}$.

By applying the approximate projection operator we obtain a quantum state $\ket{\wt{\Psi}}$:
\begin{equation}
    \beta\ket{\wt{\Psi}} = \wt{P}_{\varepsilon}\ket{\zeta},
\end{equation}
where $\ket{\wt{\Psi}}$ is a normalized quantum state and $\beta=\|\wt{P}_{\varepsilon}\ket{\zeta}\|>0$. We have
\begin{equation}
    \big|\beta-\|P_{\varepsilon}\ket{\zeta}\|\big| =  \big|\|\wt{P}_{\varepsilon}\ket{\zeta}\|-\|P_{\varepsilon}\ket{\zeta}\|\big| \leq \|P_{\varepsilon} - \wt{P}_{\varepsilon}\|\leq \epsilon_1,
\end{equation}
and as a result
\begin{equation}
    \beta \geq \|P_{\varepsilon}\ket{\zeta}\| - \epsilon_1 = |\braket{\Psi|\zeta}| - \epsilon_1 \geq 1/2\sqrt{2}.
\end{equation}
We then demonstrate that $\ket{\wt{\Psi}}$ is close to $\ket{\Psi}$, which follows from
\begin{equation}
    \label{eq:eigenstate_and_approx_eigenstate_err}
    \begin{aligned}
    \|\ket{\Psi}-\ket{\wt{\Psi}}\| &= \|\beta^{-1}\wt{P}_{\varepsilon}\ket{\zeta}-\ket{\Psi}\| \\
    &\leq \beta^{-1}\|\wt{P}_{\varepsilon}-P_{\varepsilon}\| + \|\beta^{-1}P_{\varepsilon}\ket{\zeta}-\ket{\Psi}\| \\
    &= \beta^{-1}\|\wt{P}_{\varepsilon}-P_{\varepsilon}\| + \big|\beta^{-1}-\|P_{\varepsilon}\ket{\zeta}\|^{-1} \big| \|P_{\varepsilon}\ket{\zeta}\| \\
    &\leq \beta^{-1}\epsilon_1 + \beta^{-1}\big|\beta-\|P_{\varepsilon}\ket{\zeta}\|\big| \\
    &\leq 2\beta^{-1}\epsilon_1 \leq 4\sqrt{2}\epsilon_1.
    \end{aligned}
\end{equation}

We now consider the local quantum number tail bound for $\ket{\wt{\Psi}}$. For any $\Lambda\geq 0$, we have
\begin{equation}
    \label{eq:tail_bound_approx_eigenstate}
    \begin{aligned}
    \|\overline{\Pi}_{[-\Lambda,\Lambda]}\ket{\wt{\Psi}}\| &= \beta^{-1}\|\overline{\Pi}_{[-\Lambda,\Lambda]}\wt{P}_{\varepsilon}\ket{\zeta}\| \\
    &= \beta^{-1}\|\overline{\Pi}_{[-\Lambda,\Lambda]}\wt{P}_{\varepsilon}\Pi_{[-2\bar{\lambda},2\bar{\lambda}]}\ket{\zeta}\| \\
    &\leq \beta^{-1}\|\overline{\Pi}_{[-\Lambda,\Lambda]}\wt{P}_{\varepsilon}\Pi_{[-2\bar{\lambda},2\bar{\lambda}]}\| \\
    &\leq \frac{\sigma}{\beta\sqrt{2\pi}}\int_{-T}^{T}\dd t\ e^{-\frac{1}{2}\sigma^{2}t^2}\|\overline{\Pi}_{[-\Lambda,\Lambda]}e^{itH}\Pi_{[-2\bar{\lambda},2\bar{\lambda}]}\| \\
    &\leq \poly(\bar{\lambda},\Lambda,\chi T)\exp\left(-\Omega\left(\frac{\Lambda^{1-r}-(2\bar{\lambda})^{1-r}}{2(1-r)\chi T+1}\right)\right)
    \end{aligned}
\end{equation}
where we used the fact that $\Pi_{[-2\bar{\lambda},2\bar{\lambda}]}\ket{\zeta}=\ket{\zeta}$ going from the first line to the second line, and Theorem~\ref{thm:long_time_bound2} in deriving the last line.

Now combining \eqref{eq:eigenstate_and_approx_eigenstate_err} and \eqref{eq:tail_bound_approx_eigenstate} we have
\begin{equation}
\begin{aligned}
    \|\overline{\Pi}_{[-\Lambda,\Lambda]}\ket{\Psi}\| &\leq 4\sqrt{2}\left(e^{-\frac{\delta^2}{2\sigma^2}} + \sqrt{\frac{2}{\pi}}e^{-\frac{\sigma^2 T^2}{2}}\right) \\
    &\quad+  \poly(\bar{\lambda},\Lambda,\chi T)\exp\left(-\Omega\left(\frac{\Lambda^{1-r}-(2\bar{\lambda})^{1-r}}{2(1-r)\chi T+1}\right)\right).
\end{aligned}
\end{equation}
Therefore to ensure $\|\overline{\Pi}_{[-\Lambda,\Lambda]}\ket{\Psi}\|\leq \epsilon$ we can choose $\sigma$ and $T$ to scale like
\begin{equation}
    \sigma = \Theta(\delta/\sqrt{\log(\epsilon^{-1})}),\quad T = \Theta(\sigma^{-1}\sqrt{\log(\epsilon^{-1})})=\Theta(\delta^{-1}\log(\epsilon^{-1})),
\end{equation}
and choose $\Lambda$ to be
\begin{equation}
\begin{aligned}
    \Lambda^{1-r} &= (2\bar{\lambda})^{1-r} + \Theta\left(\log(\epsilon^{-1})(2(1-r)\chi T+1)\right) \\
    &=(2\bar{\lambda})^{1-r} + \Or\left(\chi \delta^{-1}\log^2(\epsilon^{-1})+\log(\epsilon^{-1})\right).
\end{aligned}
\end{equation}

\end{proof}

\section{Lieb-Robinson velocity with on-site interaction}
\label{sec:LR_veclocity_on_site}

In this section we show that the Lieb-Robinson velocity is unaffected by any on-site interaction. 
This fact has been proved in \cite[Section~2]{NachtergaeleRazSchleinSims2009lieb}, although their result is not completely in line with what is required in this work. Therefore we provide our own theorem and proof in this section.
We use the notation in \cite[Lemma~5]{HaahHastingsKothariLow2021quantum}. We consider a lattice $\Lambda$, with $\operatorname{dist}$ denoting the lattice distance. A Hamiltonian is \emph{geometrically local} if and only if it is a sum of terms supported on lattice sites that are within a ball of constant radius in terms of the lattice distance, and the norm of each term is also bounded by a constant. By an \emph{on-site interaction}, we mean an operator that can be written as a sum of terms each of which is supported on a single site. More precisely, an operator $O$ is an on-site interaction if and only if
\begin{equation}
    O = \sum_{p\in\Lambda} O_p,
\end{equation}
where each $O_p$ is supported on site $p$.

For any time-dependent Hamiltonian $A(t)$, we use $U^A_t$ to denote the evolution under this Hamiltonian for time $t$, i.e.,
\begin{equation}
    U^A_t = \mathcal{T}e^{-i\int_0^t \dd s A(s)}.
\end{equation}

\begin{lem}
\label{lem:lieb_robinson_on_site}
Let $H(t)=h(t)+B(t)$ be a Hamiltonian on lattice $\Lambda$, where 
\[
h(t) = \sum_{X\subset \Lambda} h_X(t),
\]
is geometrically local,
and $B(t)$ is an on-site interaction, i.e.
$
B(t) = \sum_{p\in\Lambda} B_p(t).
$
Assume that
$
\zeta_0 = \max_{p,t} \sum_{Z\ni p} \|h_Z(t)\|
$
is a constant.
Then there exists constants $v_{LR}>0$ and $\mu>0$ such that for any $X\subset\Omega\subset\Lambda$, we have
\[
\|U^{H}_t O_X (U^{H}_t)^{\dagger}-U^{H_{\Omega}}_t O_X (U^{H_{\Omega}}_t)^{\dagger}\|\leq \Or(|X|\|O_X\|e^{-\mu \ell}),
\]
provided $t$ satisfies $v_{LR}|t|<\ell$. Here, $\ell=\operatorname{dist}(\Lambda\setminus \Omega,X)$ and
\[
H_{\Omega}(t) = \sum_{X\subset \Omega} h_X(t) + \sum_{p\in\Omega} B_p(t).
\]
Moreover $v_{LR}$ and $\mu$ do not depend on $B(t)$.
\end{lem}

\begin{proof}
Define
\begin{equation}
    h_I(t) = (U^B_t)^{\dagger} h(t) U^B_t.
\end{equation}
Then by switching to the interaction picture, we have
\begin{equation}
\label{eq:interaction_pic_whole_system}
    U^H_t = U^B_t U^{h_I}_t.
\end{equation}
Similarly,
\begin{equation}
\label{eq:interaction_pic_omega}
    U^{H_{\Omega}}_t = U^{B_{\Omega}}_t U^{h_{I,\Omega}}_t,
\end{equation}
where
\begin{equation}
    B_{\Omega}(t) = \sum_{p\in\Omega} B_p(t),\quad h_{I,\Omega}(t)=(U^{B_{\Omega}}_t)^{\dagger} \sum_{X\subset\Omega} h_X(t) U^{B_{\Omega}}_t.
\end{equation}
Note that we have the identity 
\begin{equation}
    h_{I,\Omega}(t)=(U^{B_{\Omega}}_t)^{\dagger} \sum_{X\subset\Omega} h_X(t) U^{B_{\Omega}}_t = (U^{B}_t)^{\dagger} \sum_{X\subset\Omega} h_X(t) U^{B}_t
\end{equation}
because all the terms in $B$ are on-site. This would not hold if there were terms of $B$ acting across the boundary of $\Omega$. The above equation tells us that the local terms of $h_{I,\Omega}(t)$ agree completely with the local terms of $h_{I}(t)$ in the region $\Omega$.

We now consider the time evolution with $h_I(t)$. Note that $h_I(t)$ is a geometrically local Hamiltonian, and is the sum of terms $ (U^{B}_t)^{\dagger} h_X(t) U^{B}_t$. For these terms we have
\begin{equation}
    \max_{p,t} \sum_{Z\ni p} \|(U^{B}_t)^{\dagger} h_Z(t) U^{B}_t\| = \max_{p,t} \sum_{Z\ni p} \|h_Z(t) \| = \zeta_0.
\end{equation}
Therefore by \cite[Lemma 5]{HaahHastingsKothariLow2021quantum}, for any operator $\wt{O}_X$ supported on $X\subset \Lambda$, we have
\begin{equation}
\label{eq:from_lemma5_hhkl}
    \|(U^{h_I}_t)^{\dagger} \wt{O}_X U^{h_I}_t - (U^{h_{I,\Omega}}_t)^{\dagger} \wt{O}_X U^{h_{I,\Omega}}_t \| \leq |X|\|\wt{O}_X\|\frac{(2\zeta_0 |t|)^{\ell}}{\ell !}
\end{equation}
if $\Omega$ is such that $X\subset\Omega$ and $\ell=\operatorname{dist}(\Lambda\setminus \Omega,X)$.

Choosing
\begin{equation}
    \wt{O}_X = (U^B_t)^{\dagger} O_X U^B_t = (U^{B_{\Omega}}_t)^{\dagger} O_X U^{B_{\Omega}}_t,
\end{equation}
we obtain
\begin{equation}
\label{eq:LR_bound_piece_together_1}
    (U^{h_I}_t)^{\dagger} \wt{O}_X U^{h_I}_t = (U^{h_I}_t)^{\dagger} (U^B_t)^{\dagger} O_X U^B_t U^{h_I}_t = (U^H_t)^{\dagger}O_X U^H_t,
\end{equation}
where we have used \eqref{eq:interaction_pic_whole_system}. Similarly,
\begin{equation}
\label{eq:LR_bound_piece_together_2}
    (U^{h_{I,\Omega}}_t)^{\dagger} \wt{O}_X U^{h_{I,\Omega}}_t = (U^{h_{I,\Omega}}_t)^{\dagger} (U^{B_{\Omega}}_t)^{\dagger} O_X U^{B_{\Omega}}_t U^{h_{I,\Omega}}_t = (U^{H_{\Omega}}_t)^{\dagger}O_X U^{H_{\Omega}}_t,
\end{equation}
where we have used \eqref{eq:interaction_pic_omega}. Substituting \eqref{eq:LR_bound_piece_together_1} and \eqref{eq:LR_bound_piece_together_2} into \eqref{eq:from_lemma5_hhkl} gives
\begin{equation}
    \|U^{H}_t O_X (U^{H}_t)^{\dagger}-U^{H_{\Omega}}_t O_X (U^{H_{\Omega}}_t)^{\dagger}\|\leq \frac{(2\zeta_0 |t|)^{\ell}}{\ell !},
\end{equation}
which is the same as Eq.~(4) in \cite[Lemma 5]{HaahHastingsKothariLow2021quantum}. The existence of constants $v_{LR}>0$ and $\mu>0$ follows immediately as in \cite[Lemma 5]{HaahHastingsKothariLow2021quantum}. Note that $v_{LR}$ and $\mu$ depend entirely on $\zeta_0$, and since $\zeta_0$ does not depend on $B(t)$, $v_{LR}$ and $\mu$ are independent of $B(t)$ either.
\end{proof}

Here we compare this result with that of \cite[Lemma 5]{HaahHastingsKothariLow2021quantum}. Note that in our setup, $h(t)$ is a geometrically local Hamiltonian satisfying the requirements in \cite[Lemma 5]{HaahHastingsKothariLow2021quantum}, but this is not true for $H(t)$, since with $B(t)$ its local terms can be unbounded. The above lemma means that adding an on-site interaction does not change the Lieb-Robinson velocity. Heuristically, on-site interactions can cause variables to vary rapidly on site, but they don't affect how rapidly information propagates from one site to another.

In Section~\ref{sec:sim_LGT} we considered simulating time evolution due to the truncated version of the Hamiltonian in \eqref{eq:LGT_ham}, which we denoted by $\wt{H}$.
There we perform simulation using the HHKL decomposition \cite{HaahHastingsKothariLow2021quantum}, and the cost of simulation depends on the Lieb-Robinson velocity.
In the context of Lemma~\ref{lem:lieb_robinson_on_site} we set
\begin{equation}
    h(t) = \wt{H}_M + \wt{H}_{GM} + \wt{H}_B,
\end{equation}
which is time-independent in this case, but we nevertheless keep the time dependence in order to be consistent with Lemma~\ref{lem:lieb_robinson_on_site}.
For $H(t)$ in Lemma~\ref{lem:lieb_robinson_on_site} we have $H(t)=h(t)+\wt{H}_E$.
We note that $h(t)$ is a geometrically-local Hamiltonian, and the corresponding $\zeta_0$ is a constant. Also we note that $\wt{H}_E$ is on-site and thus does not contribute to the propagation of operators. Therefore they meet the conditions in Lemma~\ref{lem:lieb_robinson_on_site}. By Lemma~\ref{lem:lieb_robinson_on_site}, the Lieb-Robinson velocity of evolving with $\wt{H}$ is the same as the Lieb-Robinson velocity $v_{LR}$ of evolving with $h(t)$ and is therefore a constant.
The same reasoning applies to the Hubbard-Holstein model in Section~\ref{sec:sim_boson_fermion}.

\section{Other examples of applicable models}
\label{sec:other_examples}

In this section we briefly discuss some other example models that can be analyzed within our framework. We have introduced a generic Hamiltonian \eqref{eq:ham_fermion_boson_general} describing boson-fermion coupling, and analyzed in detail how to quantumly simulate a specific model of the form  \eqref{eq:ham_fermion_boson_general}, namely the Hubbard-Holstein model, which describes electron-phonon interaction. Another common model describing electron-phonon interaction is the Fr{\"o}hlich model \cite{Frohlich1954electrons}. 

\vspace{1em}
\textit{The Fr{\"o}hlich model.}
The Hamiltonian of this model can be written as
\begin{equation}
\label{eq:Frohlich_ham}
    H = \frac{1}{2}\hat{P}^2_{\mathrm{el}} + \sum_q \omega_q (b_q^{\dagger}b_q + \frac{1}{2}) + \sum_q (V_q b_q^{\dagger} e^{-iq\cdot \hat{r}_{\mathrm{el}}}+V_q^* b_q e^{iq\cdot \hat{r}_{\mathrm{el}}}),
\end{equation}
where $q$ denotes the bosonic momentum, $b_q$ is the bosonic annihilation operator, $\hat{r}_{\mathrm{el}}$ and $P_{\mathrm{el}}$ are electron position and momentum operators respectively. In \eqref{eq:ham_fermion_boson_general} we wrote the fermionic part in second quantization and the bosonic part in first quantization. We therefore rewrite \eqref{eq:Frohlich_ham} accordingly. We use $c_{k\sigma}$ to denote the annihilation operator for an electron with momentum $k$ and spin $\sigma$. Then,
\begin{equation}
    \frac{1}{2}\hat{P}^2_{\mathrm{el}} = \sum_{k\sigma}\varepsilon_k c_{k\sigma}^{\dagger}c_{k\sigma},\quad e^{-iq\cdot \hat{r}_{\mathrm{el}}} = \sum_{k\sigma} c_{k-q\sigma}^{\dagger}c_{k\sigma}.
\end{equation}
For the bosonic part we have $b_q = \frac{1}{\sqrt{2}}(X_q+iP_q)$, and therefore we can rewrite the Hamiltonian as
\begin{equation}
\label{eq:Frohlich_ham_rewrite}
    H = \sum_{k\sigma}\varepsilon_k c_{k\sigma}^{\dagger}c_{k\sigma} + \sum_q \omega_q (X_q^2 + P_q^2) + \sum_q\sum_{k\sigma} (V_q b_q^{\dagger} c_{k-q\sigma}^{\dagger}c_{k\sigma}+V_q^* b_q c_{k+q\sigma}^{\dagger}c_{k\sigma}).
\end{equation}
We thus see that the Hamiltonian is of the form \eqref{eq:ham_fermion_boson_general}, and therefore has the structure described in Section~\ref{sec:common_structures}. The number of particles in each bosonic mode under time evolution and energy eigenstates can be analyzed using the results in Sections~\ref{sec:leakage_time_evo}, \ref{sec:truncating_the_hamiltonian}, and \ref{sec:tail}. 

We also observe that the ab initio Hamiltonians describing electron-phonon coupling \cite{Giustino2017electron}, if no anharmonic terms are included, can also be analyzd within our framework due to its similarity to the Fr{\"o}hlich Hamiltonian.
\vspace{1em}

Besides boson-fermion coupling, spin-boson coupling can also be analyzed within the framework of this work. As an example, we consider the Dicke model which describes light-matter interaction \cite{Dicke1954coherence,Hepp1973superradiant}. 

\vspace{1em}
\textit{The Dicke model.}
The model Hamiltonian can be written as 
\begin{equation}
\label{eq:Dicke_ham}
    H = \omega_c b^{\dagger} b + \omega_z \sum_{i=1}^N \sigma_i^z + \frac{g}{\sqrt{N}}(b+b^{\dagger})\sum_{i=1}^N\sigma_i^x,
\end{equation}
where $\sigma^x_i$ and $\sigma^z_i$ are the Pauli-X and Z matrices respectively acting on site $i$, and $b$ is the annihilation operator for a bosonic mode corresponding to photons. We note that this Hamiltonian has the structure described in Section~\ref{sec:common_structures}. We choose
\begin{equation}
    H_W = \frac{g}{\sqrt{N}}(b+b^{\dagger})\sum_{i=1}^N\sigma_i^x,\quad H_R=\omega_c b^{\dagger} b + \omega_z \sum_{i=1}^N \sigma_i^z.
\end{equation}
Then $H_R$ preserves the bosonic particle number, $H_W$ changes the bosonic particle number by $\pm 1$, and 
\begin{equation}
    \|H_W\Pi_{[0,\Lambda]}\| \leq 2g\sqrt{N}\sqrt{\Lambda+1},
\end{equation}
where $\Pi_{[0,\Lambda]}$ is the projection operator into the subspace with at most $\Lambda$ bosonic particles. Therefore \eqref{eq:conditions_ham_general} is satisfied if we choose $\chi=2g\sqrt{N}$ and $r=1/2$. We thus see that this model can also be analyzed within our framework.

\section{Comparison with the energy-based truncation threshold}
\label{sec:comparison_with_the_energy_based_truncation_threshold}

In Ref.~\cite{JordanLeePreskill2012quantum}, to simulate the $\phi^4$ theory, a truncation threshold is chosen for the field value at each lattice site based on energy conservation and Chebyshev's inequality. This is a very general method and can be applied to the systems studied in this work. Here we compare the truncation threshold obtained using that method with the one in this work in two settings. In the first setting we consider a single bosonic mode, and in the second we consider the Hubbard-Holstein model consisting of $N$ sites. We find that the truncation threshold in this work tends to be lower than the energy-based one 
if the truncation is made for short-time evolution of large systems with high precision.

\subsection{A single bosonic mode}

We consider a system of a single bosonic mode with the Hamiltonian
\begin{equation}
\label{eq:single_bosonic_mode_ham}
    H = b + b^{\dagger} + \omega_0 b^{\dagger}b.
\end{equation}
In this setting, the particle number expectation value can be bounded by the energy as
\begin{equation}
    E = \braket{H}\geq \omega_0\braket{b^{\dagger}b} - |\braket{b+b^{\dagger}}|\geq \omega_0\braket{b^{\dagger}b}-2\sqrt{\braket{b^{\dagger}b}+1},
\end{equation}
where the second inequality follows from the Cauchy–Schwarz inequality.
As a result,
\begin{equation}
    \braket{b^{\dagger}b} \leq \left(\omega_0^{-1}+\sqrt{1+\omega_0^{-1}E+\omega_0^{-2}}\right)^2-1.
\end{equation}

Denoting the quantum state at time $t$ by $\ket{\psi(t)}$, 
we have by Markov's inequality that
\begin{equation}
    (\Lambda+1)\braket{\psi(t)|\overline{\Pi}_{[0,\Lambda]}|\psi(t)} \leq \braket{b^{\dagger}b}.
\end{equation}
Note that $\|\overline{\Pi}_{[0,\Lambda]}\ket{\psi(t)}\|^2=\braket{\psi(t)|\overline{\Pi}_{[0,\Lambda]}|\psi(t)}$. Consequently if we want to keep $\|\overline{\Pi}_{[0,\Lambda]}\ket{\psi(t)}\|\leq \epsilon$, it suffices to choose
\begin{equation}
\label{eq:energy_based_threshold}
    \Lambda + 1 \geq \frac{\left(\omega_0^{-1}+\sqrt{1+\omega_0^{-1}E+\omega_0^{-2}}\right)^2-1}{\epsilon^2}.
\end{equation}
This is the energy-based truncation threshold.

We further consider the case where we start with at most $\Lambda_0$ particles. Then we have
\begin{equation}
    E \leq \omega_0\braket{b^{\dagger}b}+2\sqrt{\braket{b^{\dagger}b}+1} \leq \omega_0\Lambda_0+2\sqrt{\Lambda_0+1}.
\end{equation}
Substituting this into \eqref{eq:energy_based_threshold} we have
\begin{equation}
    \Lambda +1 \geq \frac{(2\omega_0^{-1}+\sqrt{\Lambda_0+1})^2-1}{\epsilon^2}.
\end{equation}
This can be directly compared with the truncation threshold \eqref{eq:truncation_threshold_boson_fermion} in this work
\begin{equation}
    \Lambda = \left(\sqrt{\Lambda_0}+\Or(T\polylog(\epsilon^{-1}))\right)^2.
\end{equation}
We can see that although the energy-based threshold has the nice property that it is independent of $T$, if $T$ is not too large and high precision is required, the truncation threshold in this work is lower, which corresponds to fewer qubits needed and smaller simulation error. 

\subsection{The Hubbard-Holstein model}
\label{sec:compare_thresholds_Hubbard_Holstein}

We then consider the Hubbard Holstein model introduced in \eqref{eq:ham_hubbard_holstein}. We rewrite the Hamiltonian as
\begin{equation}
    H = H_f + \sum_x \left( g(b_x+b_x^{\dagger})(n_{x,\uparrow}+n_{x,\downarrow}-1) + \omega_0 b_x^{\dagger}b_x \right).
\end{equation}
Note that for any $x$,
\begin{equation}
    \braket{g(b_x+b_x^{\dagger})(n_{x,\uparrow}+n_{x,\downarrow}-1) + \omega_0 b_x^{\dagger}b_x} \geq \omega_0 \braket{b_x^{\dagger}b_x} - 2|g|\sqrt{\braket{b_x^{\dagger}b_x}+1}\geq -\frac{g^2}{\omega_0}-\omega_0.
\end{equation}
As a result
\begin{equation}
\begin{aligned}
    E &= \braket{H} \geq  g(b_{x'}+b_{x'}^{\dagger})(n_{x',\uparrow}+n_{x',\downarrow}-1) + \omega_0 b_{x'}^{\dagger}b_{x'}  \\
    &+E_{f,0} + \sum_{x\neq x'} \left( g(b_x+b_x^{\dagger})(n_{x,\uparrow}+n_{x,\downarrow}-1) + \omega_0 b_x^{\dagger}b_x \right) \\
    &\geq  \omega_0 \braket{b_{x'}^{\dagger}b_{x'}} - 2|g|\sqrt{\braket{b_{x'}^{\dagger}b_{x'}}+1}  +E_{f,0} - (N-1)(\frac{g^2}{\omega_0}+\omega_0),
\end{aligned}
\end{equation}
where $E_{f,0}$ is the ground state energy of $H_f$.
Therefore
\begin{equation}
    \omega_0 \braket{b_{x'}^{\dagger}b_{x'}} - 2|g|\sqrt{\braket{b_{x'}^{\dagger}b_{x'}}+1}  \leq E - E_{f,0} + (N-1)(\frac{g^2}{\omega_0}+\omega_0).
\end{equation}
Now we assume all parameters in the model, except for $\omega_0$, are constants, and we only consider the scaling with respect to the system size. Consequently, $|E_{f,0}|=\Or(N)$, which implies
\begin{equation}
    \omega_0 \braket{b_{x'}^{\dagger}b_{x'}} - 2|g|\sqrt{\braket{b_{x'}^{\dagger}b_{x'}}+1} \leq E + \Or(\omega_0^{-1}N).
\end{equation}
Therefore we can bound the particle number expectation value on site $x'$:
\begin{equation}
\label{eq:expectation_value_bound_energy_based_Hubbard_Holstein}
    \braket{b_{x'}^{\dagger}b_{x'}}\leq \Or(\omega_0^{-2} N+\omega_0^{-1}E).
\end{equation}

Again we denote the quantum state at time $t$ by $\ket{\psi(t)}$. The projection operator into the subspace with at most $\Lambda$ particles in bosonic mode $x$ is denoted by $\Pi^{(x)}_{[0,\Lambda]}$ and we denote $\Pi^{\mathrm{all}}_{[0,\Lambda]}=\prod_x \Pi^{(x)}_{[0,\Lambda]}$. In order to ensure
\begin{equation}
    \|\overline{\Pi}^{\mathrm{all}}_{[0,\Lambda]}\ket{\psi(t)}\|\leq \epsilon,
\end{equation}
it suffices to require that
\begin{equation}
    \|\overline{\Pi}^{(x')}_{[0,\Lambda]}\ket{\psi(t)}\|\leq N^{-1/2}\epsilon,
\end{equation}
for all $x'$. Using \eqref{eq:expectation_value_bound_energy_based_Hubbard_Holstein} and Markov's inequality, we thus need to choose the truncation threshold $\Lambda$ to scale as
\begin{equation}
\Lambda = \Or\left(\frac{\omega_0^{-2} N^2+\omega_0^{-1} N E}{\epsilon^2}\right).
\end{equation}
We compare this energy-based truncation threshold with the one derived in this work in \eqref{eq:truncation_threshold_boson_fermion}, which for the Hubbard-Holstein model is
\begin{equation}
    \Lambda = \left(\sqrt{\Lambda_0}+\Or(T\polylog(N\epsilon^{-1}))\right)^2.
\end{equation}
We see that besides the advantage mentioned in the single mode setting there is also an exponentially better scaling with respect to the system size.

\begin{figure}
    \centering
    \includegraphics[width=0.45\textwidth]{boson_growth_system_size.pdf}
    \includegraphics[width=0.45\textwidth]{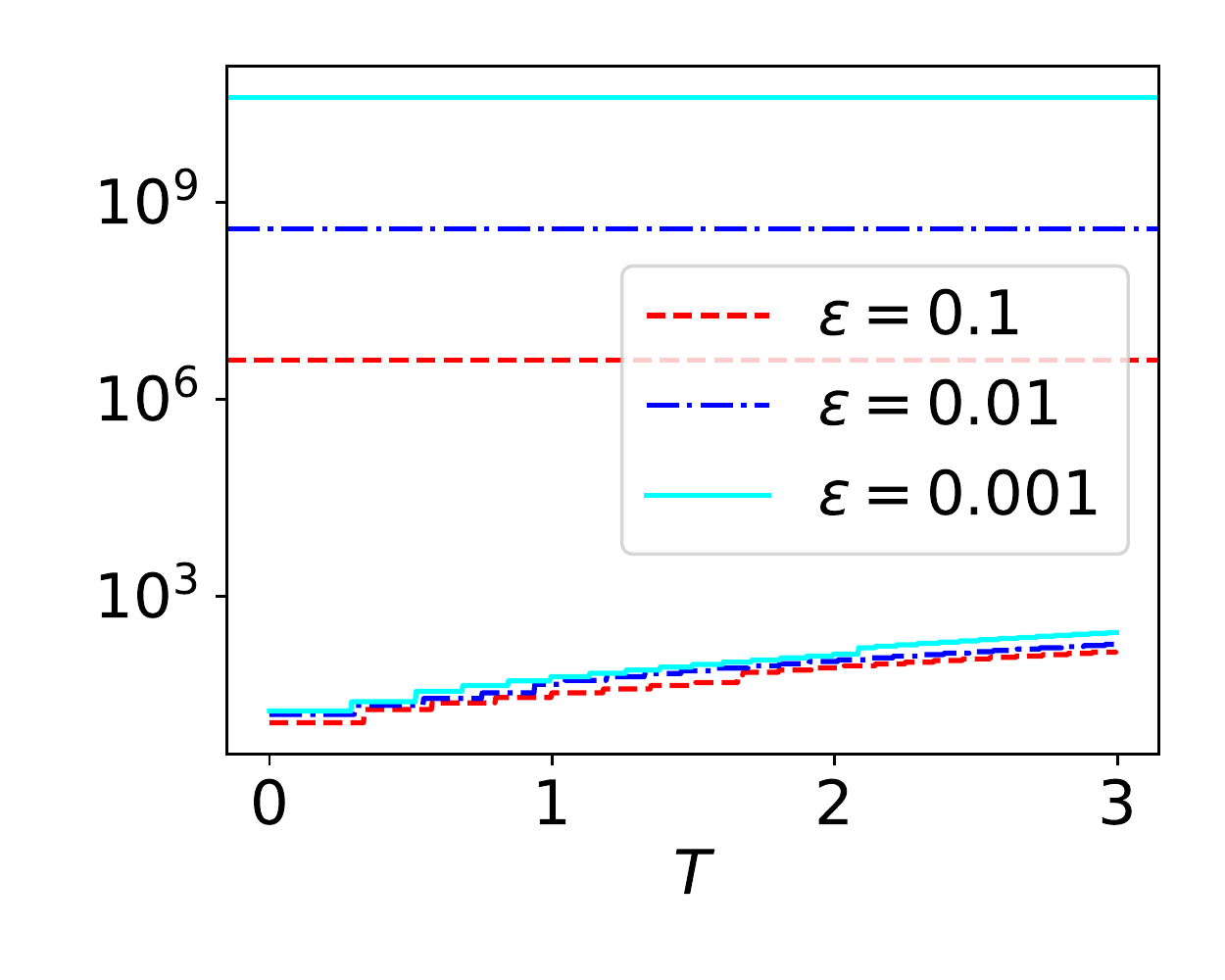}
    \includegraphics[width=0.45\textwidth]{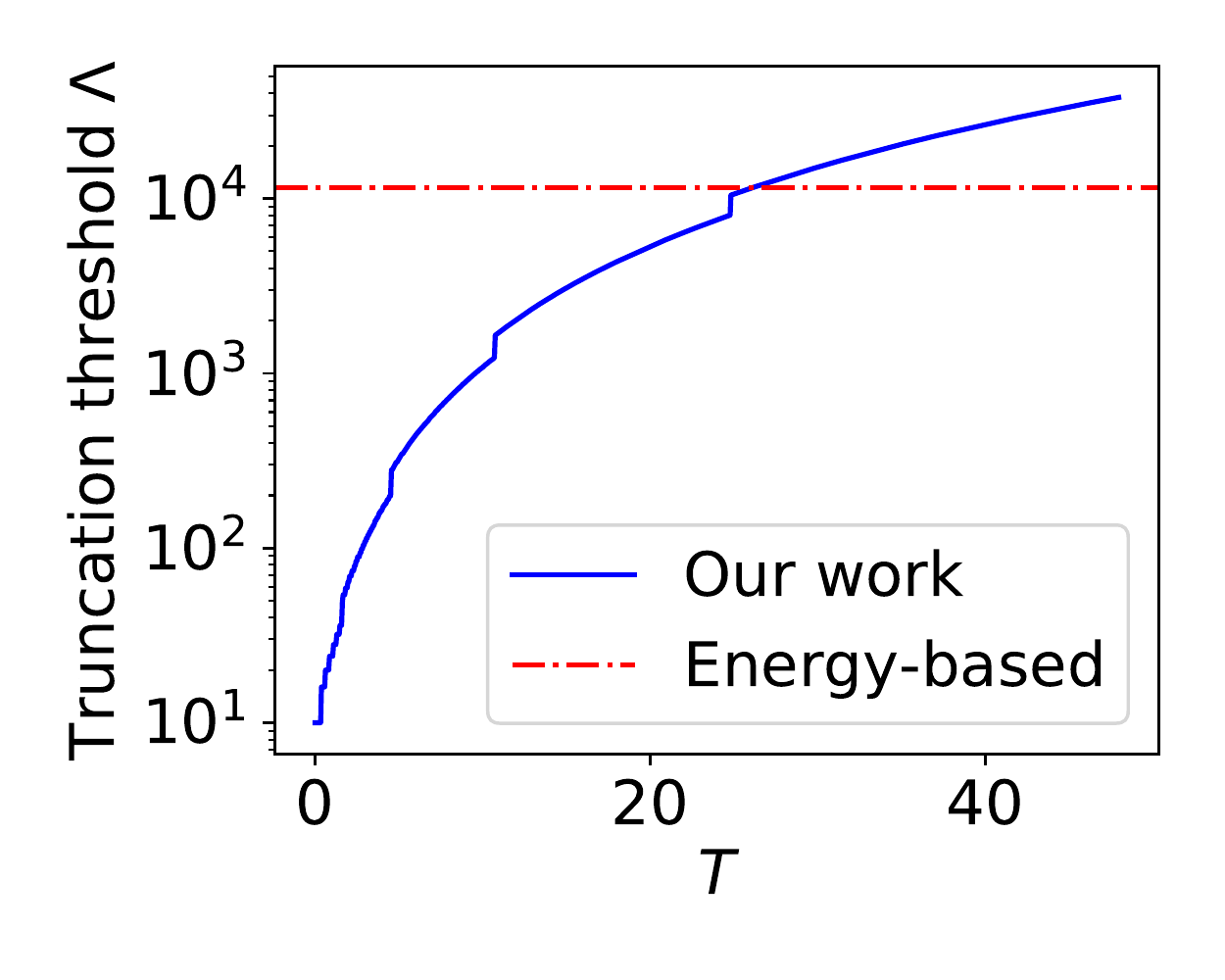}
    \caption{The truncation thresholds $\Lambda$ required for Hubbard-Holstein model of different sizes $N$ (upper left) and to achieve different precisions (upper right). The curves below show the $\Lambda$ obtained in this work and the horizontal lines above show the $\Lambda$ obtained using the energy-based method. The model parameters are chosen according to \cite{KlossReichmanTempelaar2019multiset}: $\omega_0=1,g=0.5,U=0,\mu=0$. We assume the initial state has at most $\Lambda_0=4$ particles in each bosonic mode. In the upper left panel we set $\epsilon=10^{-2}$ and in the upper right panel we set $N=100$. The panel below shows the cross-over of the two truncation thresholds using the two methods for $N=5$ and $\epsilon=0.1$.
    The truncation thresholds in this work are computed by using \eqref{eq:choose_M(t)} and choosing the smallest integer $\Delta$ to satisfy the precision requirement.
    } 
    \label{fig:boson_growth_size_and_precision}
\end{figure}

In Figure~\ref{fig:boson_growth_size_and_precision} we compare the truncation threshold $\Lambda$ computed using the method of this work and the energy-based method of \cite{JordanLeePreskill2012quantum} for the Holstein model, which is a special case of the Hubbard-Holstein model with $U=0$, with parameters chosen according to \cite{KlossReichmanTempelaar2019multiset}. We assume the initial state is a tensor product between the fermionic ground state and a quantum state of the bosonic modes that has at most $\Lambda_0=4$ particles in each mode. We clearly see that when the system size becomes larger or when the precision requirement is higher, our method yields a lower truncation threshold than the energy-based method.

\end{document}